\documentclass[%
 reprint,
 amsmath,amssymb,
 aps,
]{revtex4-2}
\usepackage{alemfonts}
\usepackage{mathtools}
\usepackage{graphicx}% Include figure files
\usepackage{dcolumn}% Align table columns on decimal point
\usepackage{bm}% bold math

\usepackage[T1]{fontenc}
\usepackage[english]{babel}
\usepackage{amsmath}
\usepackage{amsfonts}
\usepackage{amssymb}
\usepackage{mathrsfs}
\usepackage{amsthm}
\usepackage{graphicx}
\usepackage[parfill]{parskip}
\graphicspath{{./figures/}}
\usepackage{color}

\newcommand{\var}{\text{Var}}

\newcommand{\de}[1]{\frac{\partial}{\partial #1}}
\newcommand{\dde}[1]{\dot{#1}\frac{\partial}{\partial #1}}
\newcommand{\qE}[1]{\operatorname{\mathbb{E}_{#1}}}

\newcommand{\evalat}[2]{\left. #1 \right\vert_{#2}}

\newcommand{\mv}[1]{\langle #1\rangle}%

\DeclarePairedDelimiter\ton{(}{)}
\DeclarePairedDelimiter\qua{[}{]}

\DeclarePairedDelimiter\mean{\langle}{\rangle}

\newenvironment{alignLetter}{
    \setcounter{equation}{0}
    
    \align
}{
    \endalign
}

\theoremstyle{plain}
\newtheorem{remark}{Remark}
\newtheorem{theorem}{Theorem}

\newtheorem{proposition}{Proposition}
\newtheorem{corollary}{Corollary}
\newtheorem{definition}{Definition}

\begin{document}

\preprint{APS/123-QED}

\title{The emergence of a concept in shallow neural networks}

\author{Elena Agliari$^1$}
\author{Francesco Alemanno$^{2}$}
\author{Adriano Barra$^{2,3}$}
\author{Giordano De Marzo$^{4,5}$}
\affiliation{$^1$Dipartimento di Matematica, Sapienza Universit\`a di Roma, P.le A. Moro 5, 00185, Rome, Italy.}
\affiliation{$^2$Dipartimento di Matematica e Fisica, Universit\`a del Salento, Campus Ecotekne, via Monteroni, Lecce 73100, Italy.}
\affiliation{$^3$Istituto Nazionale di Fisica Nucleare, Sezione di Lecce, Campus Ecotekne, via Monteroni, Lecce 73100, Italy.}
\affiliation{$^4$Dipartimento di Fisica, Sapienza Universit\`a di Roma, P.le A. Moro 5, 00185, Rome, Italy.}
\affiliation{$^5$Centro Ricerche Enrico Fermi, Via Panisperna 89a, 00184 Rome, Italy.}

\date{\today}

\begin{abstract}We consider restricted Boltzmann machine (RBMs) trained over an unstructured dataset made of blurred copies of definite but unavailable ``archetypes'' and we show that there exists a critical sample size beyond which the RBM can learn archetypes, namely the machine can successfully play as a generative model or as a classifier, according to the operational routine.
In general, assessing a critical sample size (possibly in relation to the quality of the dataset) is still an open problem in machine learning. Here, restricting to the random theory, where shallow networks suffice and the grand-mother cell scenario is correct, we leverage the formal equivalence between RBMs and Hopfield networks, to obtain a phase diagram for both the neural architectures which highlights regions, in the space of the control parameters (i.e., number of archetypes, number of neurons,  size and quality of the training set), where learning can be accomplished. Our investigations are led by analytical methods based on the statistical-mechanics of disordered systems and results are further corroborated by extensive Monte Carlo simulations.
\end{abstract}

\keywords{Neural Networks $|$ Machine Learning $|$ Glassy Statistical Mechanics}

\maketitle

%\tableofcontents

In the past decades, the development of Artificial Intelligence has strongly benefited from the contributions of two inter-playing strands, that is, {\it neural networks} and {\it machine learning}. The former is meant as a mathematical modelling for brain abilities and, in particular, the celebrated Hopfield network \cite{Hopfield} implements Hebb's rule for synaptic plasticity and exhibits {\it pattern recognition} as emergent computational skill, hence it plays as a natural content-addressable memory 
\cite{CKS}. On the other hand, machine learning provides algorithms to make a machine able to {\it learn from experience}, namely to detect features hidden in the supplied datasets, whence making its own (compressed and probabilistic) representation of the dataset, and therefore to be able to generalize or correctly identify new examples. Machine learning  nowadays relies on a zoo of architectures among which the Boltzmann machine \cite{HHS} with its best-known algorithm, the contrastive divergence \cite{Hinton-1MC}, plays as a paradigmatic model in the statistical mechanical framework (see e.g., \cite{PRL2012,Mezard,Monasson,Aurelienne,Huang-PRL2020,Huang-PRE2020,Marullo,Love}). 
 
One of the main achievements in neural networks theory has been pioneered by Amit, Gutfreund and Sompolinsky (AGS) \cite{AGS} who first addressed the statistical mechanics of the Hopfield model obtaining for this network a phase diagram where its various operational regimes are shown as different phases (i.e., ergodic, spin-glass, retrieval regions), much as like ice, vapor and liquid for the water in thermodynamics. The phase diagram is painted in the plane of the control parameters: the fast noise $\beta$ and the load $\alpha=\lim_{N \to \infty} K/N$, namely, the number of patterns $K$ per neuron $N$ that the network stores in the thermodynamic limit. Remarkably, this knowledge allows setting a priori the system in the desired regime. Since that milestone, {\it phase transitions} entered the field of computer science in a broad variety of aspects \cite{Way1,Way2,Way3}.  Not surprisingly, thus, much efforts have been spent to outline phase diagrams also in machine learning \cite{angel-learning,sompo-learning} and, in particular, for Boltzmann machines as they can serve as building blocks of deep architectures \cite{Hinton1,Hinton} and one can possibly rely on a formal equivalence between Boltzmann machines and Hopfield networks \cite{BarraEquivalenceRBMeAHN,PRL2012,Barra-RBMsPriors2,Mezard,Monasson,Linda}.
 
Yet, a straight comparison is difficult: while in machine learning there are large datasets with features to be extracted, the Hopfield model is supplied with definite patterns, and it does not have a phase diagram with the size or the quality of the dataset as tuneable parameters.
Thus, the main conceptual problem is still to be overcome: the Hopfield network operates under the prior knowledge of prescribed patterns, that we call ``archetypes'', hence bypassing the inference process that shapes the archetype out of many examples, that is, it {\it stores} rather than {\it learns}. But what if we feed the Hopfield network on a sample $\mathcal S$ made of $K$ sub-samples, each related to a different archetype pattern and made of $M$ blurred examples of that pattern: is the network able to create its own representation of the archetypes? Driven by needs of resource optimization in machine learning and leveraging the duality between Boltzmann machines and Hopfield networks, this investigation could be a complementary approach to inspect the critical size of datasets in machine learning with respect to those achieved in the past and it is the object of the present study.
 
To this goal we first show numerically that, given the dataset $\mathcal S$ overall made of $M \times K$ blurred examples, there exists a threshold size $M_{\times}$, depending on the quality of the dataset, such that for $M>M_{\times}$ a Hopfield network -- whose Hebbian kernel is built over $\mathcal S$ -- correctly retrieves the archetypes (hence learning has successfully been accomplished), and such that a restricted Boltzmann machine (RBM) -- trained over $\mathcal S$ -- correctly learns the archetypes; remarkably, the threshold $M_{\times}$ is the same for both systems. 
 
Then, we focus on the former and we inspect how generalization, from examples to archetypes, takes place; our investigations are led by analytical techniques, from the heuristic signal-to-noise approach to the rigorous interpolation method within the statistical mechanics of spin glasses \cite{AABF-NN2020,Martino1}, and corroborated by Monte Carlo simulations; previous heuristic findings are also recovered \cite{Fontanari}.
In particular, we prove that, as the Hopfield network is provided with examples, it starts storing each of them as a distinct pattern, namely its free-energy minima match these highly-correlated examples (i.e., an ``overfitting regime''), but, beyond $M_{\times}$ minima approach archetypes and, further, once a critical amount $M_c$ of examples is reached, the attraction basins related to examples of the same archetype collapse into a new minimum corresponding to that archetype. Indeed, the network undergoes a second order phase transition: in a totally unsupervised manner, new minima are located in the landscape and these do not correspond to any of the examples used to build the Hebbian kernel; these basins eventually prevail and the corresponding archetypes will play as the patterns to be retrieved in later usage as the standard patterns of AGS theory. We stress that the critical sample size grows with the example fuzziness, parametrized by $r \in [0,1]$, and this scenario can be straightforwardly translated in the RBM framework. As a result, we can derive a phase diagram with control parameters $\alpha, \beta, M$, and $r$, where we highlight regions where the RBM trained over $\mathcal S$ can successfully perform the prescribed tasks (e.g., archetype classification, generation, or reconstruction) and this ultimately allows for resource optimization. 
 
Results are structured as follows: in the next Sec.~\ref{sec:soglia} we show computationally that the threshold size $M_{\times}$ for archetype's learning is the same for RBMs (Sec.~\ref{sec:RBM}) and Hopfield neural networks (Sec.~\ref{sec:Hop}). This suggests us to tackle the analytical investigation on the existence of different computational regimes within the Hopfield model setting and this is achieved in Sec.~\ref{sec:Hop_analitico}; next, in Secs.~\ref{sec:phase_diagram} and Secs.~\ref{sec:phase_diagram2} these analytical results are recovered and corroborated numerically in the Boltzmann machine scenario. 
%Finally, Sec.~\ref{sec:conclusions} is left for conclusions and outlooks. 
In the supplementary material all the mathematical details of both the signal to noise technique and the statistical mechanical approach are provided in detail.

\section*{Results}

\subsection{Numerical evidence of a dataset threshold size} \label{sec:soglia}
In this section we show numerically that a RBM trained over a sample of blurred examples and a Hopfield model storing the same sample of blurred examples are eventually (as the dataset gets large enough) able to generalize, namely, the former can be used as an archetype classifier/generator and the latter as an archetype retriever. Remarkably, we can detect a threshold $M_{\times}$ in the dataset size for the emergence of such a skill and this threshold turns out to be the same for both models. In the following subsections we will introduce and address the two models separately. 

\subsubsection{RBM learning from blurred samples} \label{sec:RBM}
We denote with $\{\boldsymbol{\xi}^{\mu} \}_{\mu=1,...,K}$ the $K$ archetypes, namely the patterns that we would like to see learnt by the RBM and with $\mathcal S = \{ \boldsymbol\eta^{\mu,a} \}_{\mu=1,...,K}^{a = 1,...,M}$ the related examples that the machine is actually supplied with. These objects are codified in terms of binary vectors of length $N$; pattern entries are Rademacher random variables drawn with probability 
\begin{equation}
\mathcal{P}(\xi_i^{\mu} = + 1) = \mathcal{P}(\xi_i^{\mu} = - 1) = 1/2, ~~ \forall i, \mu 
\end{equation}
while example entries are defined,  $\forall i, \mu, a$, as
\begin{eqnarray} \label{eq:sample}
& \eta_i^{\mu,a} =  \xi_i^{\mu} \chi_i^{\mu,a}, \\ \nonumber
& \textrm{with} ~~ \mathcal{P}(\chi_i^{\mu,a} = 1) = 1 - \mathcal{P}(\chi_i^{\mu,a} = -1) = p \in [1/2, 1], 
\end{eqnarray}
in such a way that the closer $p$ gets to $1/2$, the farther from the pattern gets the example\footnote{Although here we are working with random data-sets, for intuition guidance, we could look at a certain pattern $\boldsymbol \xi$ as the archetype of, say, a German Shepherd, while the set $\boldsymbol \eta^{a}$ would be a set of pictures of this dog, and similarly for the other patterns.}. We also introduce 
\begin{equation}
r:= 2p -1 \in [0,1],
\end{equation}
as an index for the dataset quality: $r$ ranges from $0$ (any example and the related archetype are uncorrelated) to $1$ (any example coincides with the related archetype).

\begin{figure*}[tb] 
    \centering
    \includegraphics[scale=0.475]{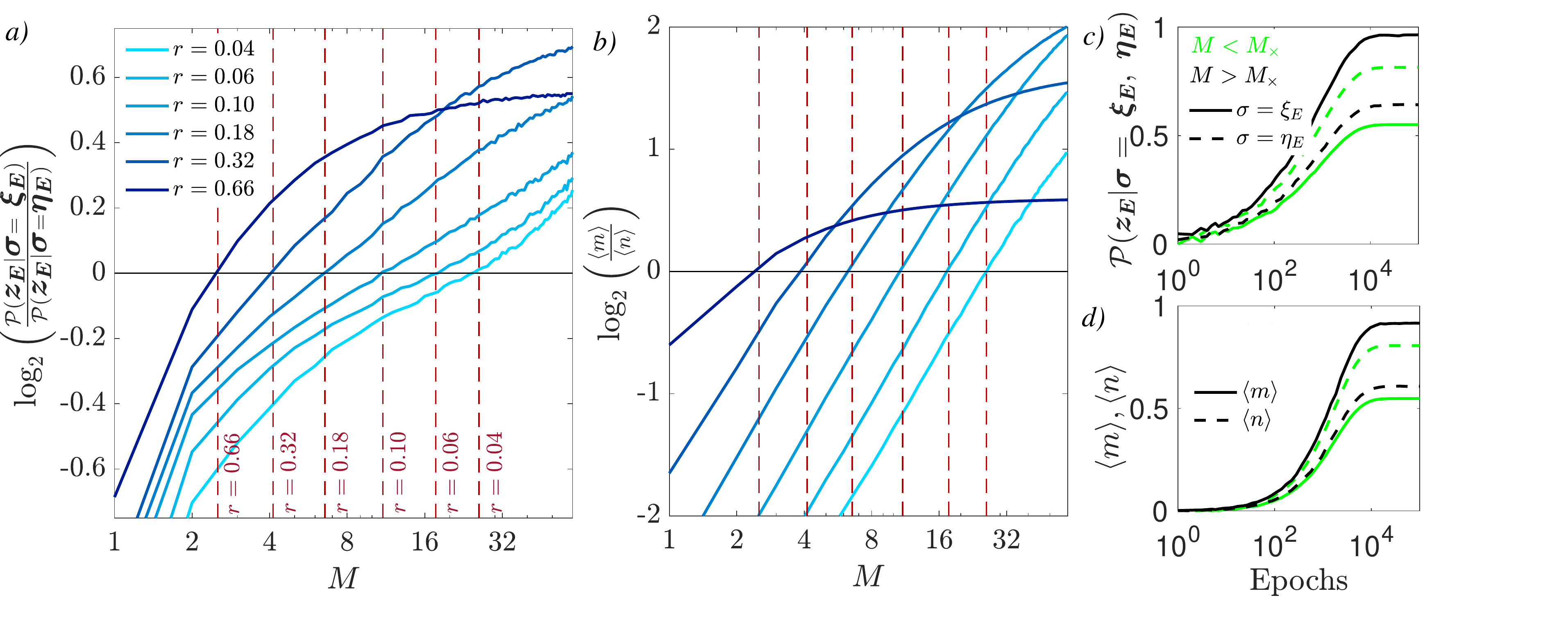}
    \caption{The larger panels provide a picture of the performance of a trained RBM used as a classifier (panel $a$) and as a generative model (panel $b$); in particular, the logarithm of $\mathcal P(\boldsymbol z_E |  \boldsymbol \sigma = \boldsymbol \xi_E) / \mathcal P(\boldsymbol z | \boldsymbol \sigma = \boldsymbol \eta_E)$ and the logarithm of 
    $\langle m \rangle /\langle n \rangle$, respectively, are shown versus $M$, for different choices of the parameters $r$ (as explained by the common legend in panel $a$), which quantifies the dataset quality. 
   The threshold value $M_{\times}$ corresponds to the interception between the curves and the horizontal axis. 
        In both panels the vertical dashed lines are obtained analytically by studying the dual Hopfield network and asking for the minimum value of $M$ such that archetype retrieval prevails over example retrieval (i.e., $\bar m > \bar n$, see Sec.~\ref{sec:Hop} and SM); this estimate is obtained for different choices of $r$, as reported. Note that the vertical lines intersect the experimental curves always when they also cross the horizontal line, showing that the threshold size is the same for these machines.
        The smaller panels provide a picture of the training routine for the RBM. As epochs run, we show the evolution of the classification probabilities (panel $c$)
        $\mathcal P(\boldsymbol z_E |  \boldsymbol \sigma = \boldsymbol \xi_E)$ and $\mathcal P(\boldsymbol z_E |  \boldsymbol \sigma = \boldsymbol \eta_E)$ (respectively, solid and dashed lines) and of the overlaps (panel $d$) $\langle m \rangle$ and $\langle n \rangle$ (respectively solid and dashed lines), distinguishing between the case $M> M_{\times}$ (bright color) and $M> M_{\times}$ (dark color), as explained in the legend.} \label{fig:one}
\end{figure*}
The machine  is  made of two layers: the visible one made of $N$ binary neurons $\sigma_i = \pm 1, \  \  i  \in (1,...,N)$ and  a hidden one built of by $K$ binary neurons $z_{\mu}= \pm 1, \ \mu \in (1,...,K)$; we denote with $(\boldsymbol \sigma , \boldsymbol z) \in \{ -1, + 1 \}^{N \times K}$ the overall configuration. Note that there are as many hidden neurons as archetypes and, as we will explain in the following, this architecture allows us to allocate one hidden neuron per archetype, configuring the network in the {\it grandmother cell} scenario \cite{Nonna1,Nonna2}. 
\newline
We also introduce the weight matrix $ \boldsymbol w \in \mathbb R^{N \times K}$, whose entry $w_i^{\mu}$ represents the weight associated to the connection between neurons $i$ and $j$ belonging to different layers. The cost function (or Hamiltonian to keep a physical jargon)  $H_{N,K}(\boldsymbol \sigma, \boldsymbol z | \boldsymbol w)$  related to this RBM reads as
\begin{equation}
  \label{eq:Ham-RBM}
  H_{N,K}(\boldsymbol \sigma, \boldsymbol z | \boldsymbol w) = -\frac{1}{\sqrt{N}}  \sum_{\mu=1}^K  \sum_{i=1}^N w_i^\mu  \sigma_i z_{\mu},
\end{equation}
where the factor $\sqrt{N}$ ensures the linear scaling of the cost function with respect to the size in the thermodynamic limit $N \to \infty$.
The equilibrium distribution for such a system is given by the Boltzmann-Gibbs measure
\begin{equation}
\mathcal P(\boldsymbol \sigma, \boldsymbol z| \boldsymbol w) = \frac{1}{Z_{N,K}} e^{ - \beta   H_{N,K}(\boldsymbol \sigma, \boldsymbol z | \boldsymbol w) },
\end{equation}
where $Z_{N,K}$ is the suitable normalization factor obtained by summing the exponential term over all possible configurations $(\boldsymbol \sigma, \boldsymbol z) \in \{ -1, +1\}^{N\times K}$.

We train this machine in a supervised mode, that is, during the training phase the clamped setting involves both the visible and the hidden degrees of freedom, namely $\boldsymbol{\sigma}$ is set to one of the examples in the dataset, say the $(\nu,a)$-th one, i.e., $\sigma_i = \eta_i^{\nu,a}$ for $i=1,...,N$, while  $\boldsymbol{z}$ is set to a one-hot vector where the entry related to the correct archetype is $1$ and the others $0$, i.e., $z_{\mu} = \delta_{\mu,\nu}$ for $\mu=1,...,K$; we call $\mathcal Z$ the set of all possible $K$ one-hot vectors. This kind of setting can be interpreted as a grandmother-cell setting, namely we establish a one-to-one correspondence between hidden neurons and archetypes and  -- in the clamped state -- we force solely one hidden neuron per archetype to be active, whence the constraint on the number of archetypes equal to the number of hidden neurons.

More specifically, the machine training is accomplished by means of the following Hinton's scheme of contrastive divergence \cite{Hinton-1MC}:
$$
\Delta w_i^{\mu} \propto \left(  \langle \sigma_i z_{\mu} \rangle_{\textrm{clamped}} - \langle \sigma_i z_{\mu} \rangle_{\textrm{free}} \right), \  \ \  \forall (i,\mu) \in (N \times K), 
$$	
where for each training step the ``free'' average is sampled via a single step of Gibbs dynamics, i.e. a random training example $(\boldsymbol{\sigma}_E,\boldsymbol{z}_E) \in \mathcal S \times \mathcal Z$ is selected, then the free mean is calculated single shot via a pair $(\boldsymbol{\sigma}_{\rm free},\boldsymbol{z}_{\rm free})$ sampled using the Gibbs-chain $\boldsymbol{z}_E\to\boldsymbol{\sigma}_{\rm free}\to\boldsymbol{z}_{\rm free}$; the ``clamped'' average is also evaluated single shot using the same pair $(\boldsymbol{\sigma}_E,\boldsymbol{z}_E)$.

The trained machine can be used as a classifier (i.e., as a pattern recognition device, by feeding the machine a noisy  $\boldsymbol{\sigma}$ configuration and letting the machine recover the $\boldsymbol{z}$ configuration whose entries indicate how the input signal has been classified), or as a generative model (by feeding the machine a $\boldsymbol{z}$ configuration and letting the machine output the  $\boldsymbol{\sigma}$ configuration of the corresponding archetype).
We inspect the success of the learning procedure by testing the machine as a classifier and as a generative model, as reported in panels $a$ and $b$ of Fig.~\ref{fig:one}. More precisely, we choose as performance measure for classification the logarithm of $\mathcal P(\boldsymbol z_E | \boldsymbol \sigma = \boldsymbol \xi_E)/ \mathcal P(\boldsymbol z_E | \boldsymbol \sigma = \boldsymbol \eta_E)$, where $\mathcal P(\boldsymbol z_E | \boldsymbol \sigma = \boldsymbol \xi_E)$ (respectively $\mathcal P(\boldsymbol z_E | \boldsymbol \sigma = \boldsymbol \eta_E)$) is the probability of reaching a correct hidden state $\boldsymbol z_E$ given a visible state clamped as $\boldsymbol \sigma = \boldsymbol \xi_E$ (respectively $\boldsymbol \sigma = \boldsymbol \eta_E$); the ratio between the two terms allows us to assess when one prevails over the other (see Fig.~\ref{fig:one}, panel $a$).  To evaluate computationally $\mathcal P(\boldsymbol z | \boldsymbol \sigma = \boldsymbol \xi)$ (and analogously $\mathcal P(\boldsymbol z | \boldsymbol \sigma = \boldsymbol \eta)$)  we provide the network with, respectively, the archetype and the example on the visible layer and we study the distribution of activations within the hidden layer (i.e., the entries of the $\boldsymbol z$ vector):  as training followed the grandmother-cell setting we expect to have just one positive  entry -- the hidden neuron coupled to the selected archetype -- if learning has been properly accomplished as empirically confirmed.  
\newline
When looking at the Boltzmann machine as a generative model, we use a different performance measure: having trained the machine as specified above, we clamp the hidden layer on a certain one-hot vector $\boldsymbol z_E \in \mathcal Z$ and we let the machine thermalize allowing visible neurons to evolve freely; we expect that the system relaxes to configurations where $\boldsymbol \sigma$ corresponds to the related archetype $\boldsymbol \xi_E$. To check whether this is the case we measure 
the overlap between the visible neuron configuration $\boldsymbol \sigma$ and $\boldsymbol \xi_E$ and compare it with the overlap between $\boldsymbol \sigma$ and the examples corresponding to the class of $\boldsymbol \xi_E$. To fix ideas, let us set  $\boldsymbol \xi_E = \boldsymbol \xi^{\nu}$, then we introduce $n^{\nu,a}$ as the overlap between $\boldsymbol \sigma$ and the $(\nu,a)$-th example $\boldsymbol \eta^{\nu,a}$ for $a=1,...,M$, namely
\begin{equation} \label{eq:n_Mattis}
n^{\nu,a}:=\frac{1}{N} \sum_{i=1}^N \xi_{i}^{\nu} \chi_{i}^{\nu,a} \sigma_{i}, 
\end{equation}
and $m^{\nu}$ as the overlap between $\boldsymbol \sigma$ and the $\nu$-th archetype $\boldsymbol \xi^{\nu}$, namely
\begin{equation} \label{eq:m_Mattis}
m^{\nu}:=\frac{1}{N} \sum_{i=1}^N \xi_{i}^{\nu} \sigma_{i}.
\end{equation}
 
To evaluate computationally these overlaps we first evaluate $n^{\nu,a}$ and $m^{\nu}$ as normalized dot product between the thermalized configuration $\boldsymbol \sigma$ and, respectively, $\boldsymbol \eta^{\nu,a}$ and $\boldsymbol \xi^{\nu}$ as per definitions \eqref{eq:n_Mattis} and \eqref{eq:m_Mattis}, then we average over different choices of clamped states (namely by varying $\nu \in [1,K]$) and over different realizations of archetypes (this is the analogous of a quenched average). These mean values are denoted as $\langle n \rangle$ and $\langle m \rangle$.
Their comparison allows us to evaluate whether the system is more prone to generate one of the examples it has been exposed to or to generate the unseen archetype (see Fig.~\ref{fig:one} , panel $b$).

For both operational modes we see that, if the number of examples provided to the network is relatively small, the system fails, that is the system can classify examples better than archetypes (i.e., $\mathcal P (\boldsymbol z_E | \boldsymbol \sigma =  \boldsymbol \eta_E) > \mathcal P (\boldsymbol z_E | \boldsymbol \sigma = \boldsymbol \xi_E)$) or the system generates examples rather than archetypes (i.e., $\langle n \rangle > \langle m \rangle$). However, if the number of examples is relatively large, the system succeeds, that is the system can classify archetypes better than examples or the system generates archetypes rather than examples. The threshold between a ``small'' and a ``large'' dataset is denoted with $M_{\times}$ and, as expected, $M_{\times}$ grows as the sample quality $r$ decreases. Empirically, we find that $M_{\times} \sim r^{-1}$.
On the other hand, the two extreme cases $M=1$ and $M \to \infty$ are trivial as for $M=1$ there is no difference between examples and archetypes (each example is also an archetype) while for $M\to \infty$  the archetype always prevails over examples by a standard central limit theorem argument. 

Analogous remarks can be drawn also from Fig.~\ref{fig:one} panels $c$, $d$ where we show the evolution of the classification probabilities $\mathcal P(\boldsymbol z_E |  \boldsymbol \sigma = \boldsymbol \xi_E)$ and $\mathcal P(\boldsymbol z_E |  \boldsymbol \sigma = \boldsymbol \eta_E)$ and of the mean overlaps $\langle m \rangle$ and $\langle n \rangle$ as the training is running. Interestingly, as long as $M < M_{\times}$, the saturation values for $\mathcal P(\boldsymbol z_E |  \boldsymbol \sigma = \boldsymbol \eta_E)$ and $\langle n \rangle$ are larger than those obtained for $\mathcal P(\boldsymbol z_E |  \boldsymbol \sigma = \boldsymbol \xi_E)$ and $\langle m \rangle$; the opposite holds as $M > M_{\times}$.

We conclude this section recalling that we can recast the problem of archetypes generation and classification exploiting the duality between Boltzmann machines and Hopfield networks: as largely discussed in the past decade \cite{BarraEquivalenceRBMeAHN,Barra-RBMsPriors2,Barra-RBMsPriors2,Linda,Monasson,Aurelienne,Mezard,Huang3,Marullo}: by marginalizing the probability distribution $\mathcal P(\boldsymbol \sigma, \boldsymbol z| \boldsymbol w)$ over the hidden layer, we end up with the probability distribution of a Hopfield network as long as we identify the weights $w_{i}^{\mu}$ in the former with the entries of the patterns stored by the latter, and we suitably rescale the temperature; in formulae
%\begin{equation}\label{eq:formale}
%\mathcal P(\boldsymbol \sigma, \boldsymbol z| \boldsymbol w)=\sum_{\sigma}^{2^N}\sum_{z}^{2^K} e^{{\frac{\beta}{\sqrt{N}}}  \sum_{\mu=1}^K  \sum_{i=1}^N w_i^\mu  \sigma_i z_{\mu}} \to \mathcal P(\boldsymbol \sigma| \boldsymbol w)  \propto \sum_{\sigma}^{2^N}e^{\frac{\beta^2}{2N}\sum_{i,j}^{N,N}\sum_{\mu}^K (w_i^{\mu}w_j^{\mu})\sigma_i \sigma_j}.
%\end{equation}
\begin{eqnarray}\label{eq:formale}
\mathcal P(\boldsymbol \sigma, \boldsymbol z| \boldsymbol w)&=&\sum_{\sigma}^{2^N}\sum_{z}^{2^K} e^{{\frac{\beta}{\sqrt{N}}}  \sum_{\mu=1}^K  \sum_{i=1}^N w_i^\mu  \sigma_i z_{\mu}} \\ \nonumber
&\to& \mathcal P(\boldsymbol \sigma| \boldsymbol w)  \propto \sum_{\sigma}^{2^N}e^{\frac{\beta^2}{2N}\sum_{i,j}^{N,N}\sum_{\mu}^K (w_i^{\mu}w_j^{\mu})\sigma_i \sigma_j}.
\end{eqnarray}

From this perspective we may want to check the ability of the system to retrieve an archetype, namely if we initialize the Hopfield network in a configuration $\boldsymbol \sigma$ corresponding to an example, say $\boldsymbol \eta^{\nu,a}$, and let it thermalize towards equilibrium, does it eventually end up ``close'' to the archetype $\boldsymbol \xi^{\nu}$?  This problem is faced in the next section.

\begin{figure}[tb]
\centering
\includegraphics[width=1.0\linewidth]{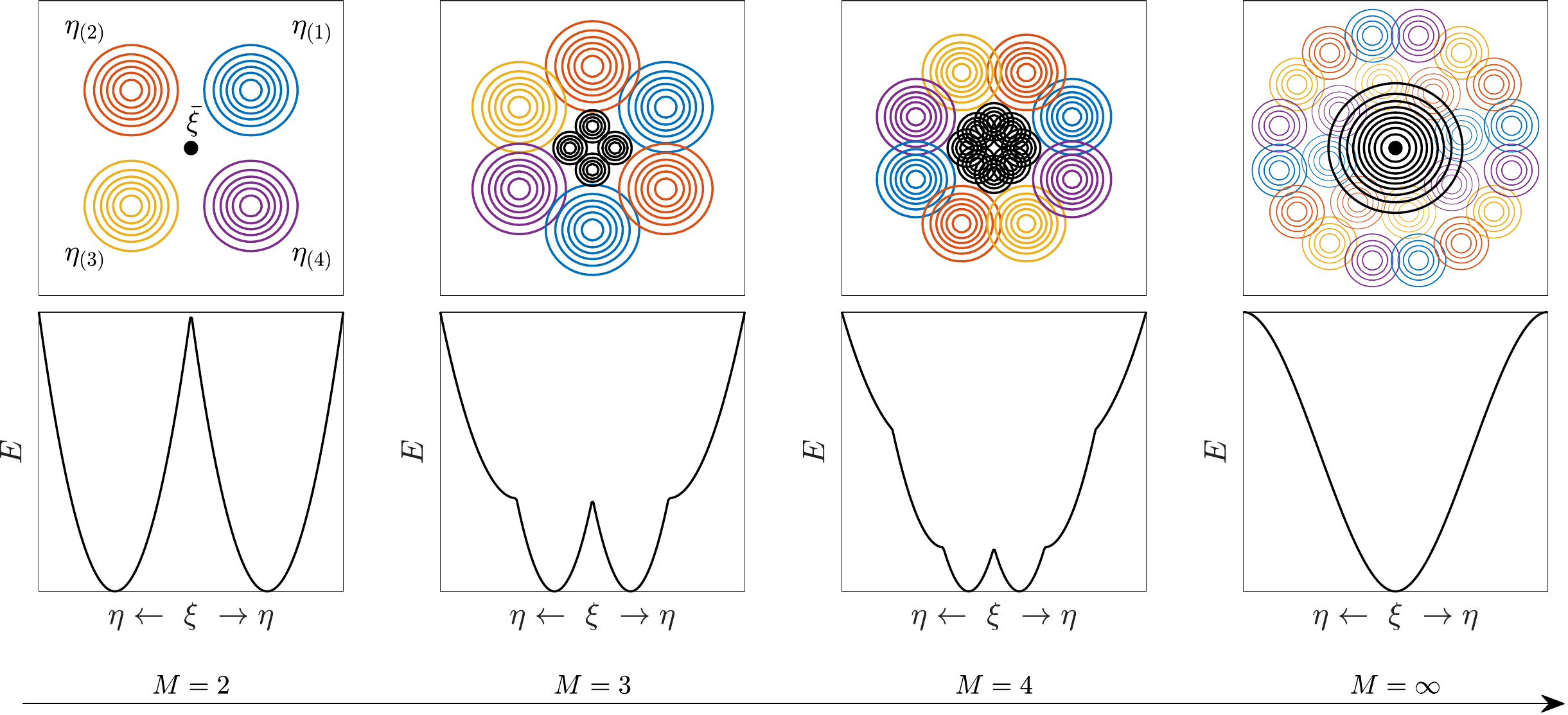}
\caption{\label{fig:minima}Schematic representation of the emergence of the archetype minimum in the energy landscape. The network is supplied with $M$ examples of the pattern $\boldsymbol{\xi}$, which, instead, is never presented to the network. As $M$ grows, from left to right, the network at first stores each single example but it is unable to retrieve $\boldsymbol{\xi}$ (left, $M<M_{\times}$), then, in the energy landscape, new minima, close to $\boldsymbol{\xi}$, appear and coexist with the minima corresponding to examples (center, $M > M_{\times}$) and, finally, a unique stable minimum corresponding to the archetype emerges (\textcolor{black}{right}, $M \approx M_c$).}
\end{figure}

\subsubsection{Hopfield network learning from blurred samples} \label{sec:Hop}
Let us consider a Hopfield neural network made of $N$ binary neurons, whose overall configuration is denoted with $\boldsymbol \sigma = (\sigma_1, \sigma_2,...,\sigma_N)$, and supplied with the sample $\mathcal S$ made of $K \times M$ examples $\eta_i^{\mu,a} =  \xi_i^{\mu} \chi_i^{\mu,a}$ as defined in \eqref{eq:sample}. We want to apply Hebb's rule to this sample and check whether the resulting system is able to generalize, namely to retrieve the archetype $\{\boldsymbol \xi^{\mu} \}_{\mu=1,...,K}$ once provided with an example\footnote{When this can be accomplished we say that also the Hopfield network can generalize because, starting from the inferred archetype, it generates variations on theme by taking advantage of the fast noise $\beta$}. We write the coupling $J_{ij}$ between the neurons $i$ and $j$ as
\begin{equation}
J_{ij}=\frac{1}{N} \sum_{\mu=1}^K \sum_{a=1}^M \eta^{\mu, a}_i\eta^{\mu, a}_j. \label{eq:couplings_single}
\end{equation}
Notice that, in this definition, we are simply summing over all the instances making up the sample $\mathcal S$, without caring of the class each term belongs to, in this sense, this kind of Hebbian learning is \emph{unsupervised}. 
The cost function (or Hamiltonian to keep a physical jargon) $H_{N,K,M}(\boldsymbol \sigma | \boldsymbol \chi, \boldsymbol \xi) $ of the model reads as follows
\begin{equation}\small
  \label{eq:Hamiltonian}
  H_{N,K,M}(\boldsymbol \sigma | \boldsymbol \chi, \boldsymbol \xi) = -\frac{1}{2N} \sum_{\mu,a}^{K,M} \big(\sum_{i=1}^N \xi_i^\mu \chi_i^{\mu,a}\sigma_i \big) ^2 = - \frac{N}{2} \sum_{a=1}^M   (\boldsymbol n^{a})^2,
\end{equation}
where $\boldsymbol n^{a} = (n^{1,a}, ..., n^{K,a})$ with entries defined in \eqref{eq:n_Mattis}.
Analogously, we pose $\boldsymbol m = (m^{1}, ..., m^{K})$ with entries defined in \eqref{eq:m_Mattis}. In this context, we shall also refer to $\boldsymbol n^{a}$ and $\boldsymbol m$ as Mattis magnetizations related to, respectively, examples and archetypes.
\newline
As we will see, these quantities play as key order parameters to quantify how (and what kind of) pattern recognition is accomplished by the network, implicitly quantifying the goodness of its learning too.
In the following we will denote with $\bar{m}$ and $\bar{n}$ their expectations with respect to the Boltzmann-Gibbs distribution related to the cost function \eqref{eq:Hamiltonian}, namely
\begin{equation}\label{eq:P_H}
\mathcal P(\boldsymbol \sigma  | \boldsymbol \chi, \boldsymbol \xi)= \frac{1}{Z_{N,K,M}}  e^{- \frac{\beta}{2N} H_{N,K,M}(\boldsymbol \sigma | \boldsymbol \chi, \boldsymbol \xi) },
\end{equation}
where $Z_{N,K,M}$ is the suitable normalization factor obtained by summing the exponential term over all possible configurations $\boldsymbol \sigma \in \{ -1,+1\}^N$.

As detailed in the next sections, this model can be addressed analytically and we can obtain -- in the thermodynamic limit and in the high-storage regime (i.e., $\alpha=\lim_{N \to \infty} K/N$ finite) -- self-consistent equations for its order parameters that can be then solved numerically. Following this route, we can compare $\bar n$ and $\bar m$ and check whether $\bar{m} > \bar{n}$ finding that for this condition to hold, $M$ must be larger than a certain threshold, represented by the vertical dashed lines in Fig.~\ref{fig:one}:  remarkably, this threshold corresponds to the threshold value $M_{\times}$ of the RBM. 
 
Before proceeding we anticipate that the analytical investigation performed on the Hopfield model \eqref{eq:Hamiltonian} highlights a rich phenomenology that here we try to summarize by means of Fig.~\ref{fig:minima} that sketches the evolution of (a cross section of) the cost-function landscape $E := H_{N,K,M}(\boldsymbol \sigma | \boldsymbol \chi, \boldsymbol \xi)$ as the dataset size is made larger; in this landscape, we especially care of minima since they play as attraction basins for the neural configuration $\boldsymbol \sigma$. When $M$ is small the landscape exhibits $K \times M$ minima\footnote{The number of minima is actually $2\times K \times M$ due to the gauge symmetry.} corresponding to the examples provided; as $M$ is made larger, minima get denser and their attraction basins possibly overlap; when $M > M_{\times}$ the minima corresponding to examples are only local while new and deeper minima emerge, whose location is closer to the archetype rather than any other example; as $M$ is further increased local minima get less and less stable while global minima get closer and closer to the archetypes; finally when $M$ is large enough, configurations corresponding to archetypes become stable. As we will see in the next section, the last passage can be related to a critical value for $M$, that scales with $r$ and that we denote with $M_c$. Interestingly, we also find out that setting $M=M_c$ determines the onset of a critical phase transition.   \\
Therefore, for the Hopfield network defined in \eqref{eq:Hamiltonian}, in addition to the traditional tuneable parameters, namely the fast noise $\beta$ and the load $\alpha=\lim_{N \to \infty} K/N$, we have the sample size $M$ and the sample quality $r$; we expect the system to correctly retrieve archetypes as long as
$\alpha < \alpha_c(\beta)$ and as long as $M > M_c(r)$.  
When translating this knowledge into the RBM scenario, we derive restrictions in the data-dimensionality reduction ability of Boltzmann machines (note that $\alpha$ corresponds to the ratio between the sizes of the hidden and the visible layers in RBM) \cite{Hinton} and an interplay between dataset quality and quantity.

\subsection{Analytical results} \label{sec:Hop_analitico}
As anticipated, our analytical investigations shall focus on the Hopfield counterpart for which we can rely on solid mathematical methods. 

We start with the signal-to-noise analysis (extensively reported in the SM) to check for local stability of the configurations $\boldsymbol \sigma = \boldsymbol \eta^{\mu, a}$ and $\boldsymbol \sigma = \boldsymbol \xi^{\mu}$, for arbitrary $\mu$ and $a$, in the noiseless limit $\beta \to \infty$. This is accomplished by studying if the internal field $h_i=\sum_{j=1}^{N}J_{ij}\sigma_j$, experienced by the neuron $i$, is aligned with the neural activity $\sigma_i$ and by monitoring the evolution of the relative energies associated to these configurations (see Figure $2$). We find that by increasing $M$, archetypes (examples) progressively gain (loose) stability at a rate depending on $r$. In particular, as for archetypes, the stability threshold $M_c$ increases according to the following scaling
\begin{equation}
%\label{uno}
%M_c &\sim& (2p-1)^{-2}, \ \ \  \textit{Low storage} ~(\alpha=0),\\
\label{due}
M_c \sim (2p-1)^{-4}, %\ \ \  \textit{High storage} ~(\alpha>0).
\end{equation}
To get sharper estimates and a characterization of a possible phase transition, we need to solve for the quenched free-energy of the model and inspect the related self-consistent equations for order parameters: here we simply report the main points, while we refer to the SM for technical details. 
In the limit of infinite volume $N$, but finite dataset size $M$, the quenched pressure (i.e., $- \beta$ times the free energy) of the model \eqref{eq:Hamiltonian} is defined as
\begin{equation}\label{FreeE}
A_{M}(\alpha,\beta) := \lim_{N \to \infty} \frac{1}{N}\mathbb{E}\ln Z_{N,K,M}(\beta| \boldsymbol \chi, \boldsymbol \xi),
\end{equation}
where $\mathbb{E}:=\mathbb{E}_\chi \mathbb{E}_{\xi}$ averages over both the quenched variables $\chi, \xi$, and $Z_{N,K,M}$ is the partition function %Z_{N,M}(\beta| \boldsymbol \chi,\boldsymbol \xi)$ 
given by
\begin{equation}
\label{def:Z_H}
Z_{N,K,M}(\beta| \boldsymbol \chi, \boldsymbol \xi) = \sum_{\sigma}^{2^N} \exp \Big[\frac{\beta}{2N} \sum_{a=1}^M \sum_{\mu=1}^K \big(\sum_{i=1}^N \xi_i^\mu \chi_i^{\mu,a}\sigma_i \big) ^2 \Big].
\end{equation}
Note that, as shown in the Supplementary Material (see Proposition One) by a trivial Hubbard-Stratonovich transformation,  this partition function coincides with that of a RBM equipped with Gaussian prior.
At the replica symmetric level of description, keeping $M$ fixed but sending both $K$ and $N$ to infinity in such a way that $\alpha$ is finite, and focusing on the retrieval of $\boldsymbol \xi^{1}$ with no loss of generality, we reach the following expressions for the quenched pressure
%\begin{equation}
%\label{eq:rssolution}
%\notag
%A_{M}(\alpha,\beta) =\log2 -\frac{\beta\alpha M}{2}\bar p(1-\bar q) -\frac{\beta M}{2} \bar n^{2} -\frac{\alpha M}{2} \big(\log[1-\beta (1-\bar q)] - \frac{\beta\bar q}{1-\beta(1-\bar q)}\big) +  \mathbb{E}_{\phi\chi}\log\cosh \big(\bar n \beta \sum_{a=1}^M \chi_{a} + \sqrt{\alpha\beta\bar p M}\phi\big),
%\end{equation}

\begin{eqnarray}\small
\label{eq:rssolution}\nonumber
\notag
&& A = \log2 -\frac{\beta\alpha M}{2}\bar p(1-\bar q) -\frac{\beta M}{2} \bar n^{2}  -\frac{\alpha M}{2} \big(\log[1-\beta (1-\bar q)]\\  \nonumber \small
&& -\frac{\beta\bar q}{1-\beta(1-\bar q)}\big) +  \mathbb{E}_{\phi\chi}\log\cosh \big(\bar n \beta \sum_{a=1}^M \chi_{a} + \sqrt{\alpha\beta\bar p M}\phi\big),
\end{eqnarray}

\begin{figure}[tb]
\centering
\includegraphics[width=1.15\linewidth]{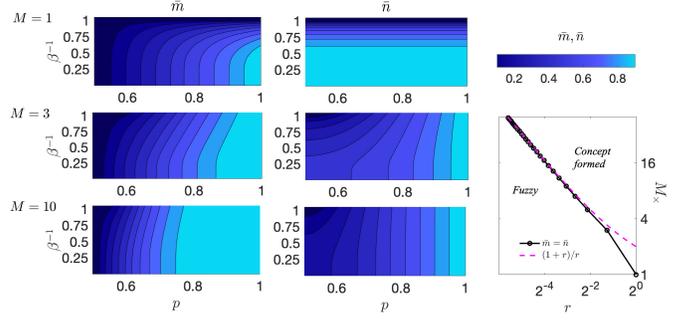}
\caption{\label{fig:soft_trans} Left: Contour plots for the magnetization of the archetype $\bar m$ (left panels) and of the examples $\bar n$ (right panels), obtained by solving the self-consistencies in eqs.~(\ref{eq:sce_n}) and (\ref{eq:relation}) for $\alpha =0$ and for $M=1, 3, 10$ (from top to bottom), versus $p$ ($x$-axis) and $\beta^{-1}$ ($y$-axis); analogous results are obtained for $\alpha >0$, see the SM. By comparing the values of $\bar{m}$ and $\bar{n}$ we see that, as the number of examples exceeds a bound $M_{\times}(r)$, the archetype retrieval dominates over the example retrieval. Right: the condition $\bar{m} = \bar{n}$ is recognized as the boundary between a {\it fuzzy} regime where the sample size is not enough for the archetype to be inferred and the network is only able to retrieve the examples it has been presented to, and a {\it concept-formed} regime where the network ``forgets'' about the examples and can retrieve the archetype. Consistency between the theoretical (dashed line) and the empirical (bullets, solid line is a guide for eyes) estimates is provided. \textcolor{black}{Notice that, at $\alpha=0$, the function $M_{\times}(r)$ is temperature independent.}}
\end{figure}
where $\phi$ is an auxiliary, mute, Gaussian field, while $\bar p$, $\bar q$ and $\bar n$ are, respectively, the expectation values for the order parameters $p_{12}$, $q_{12}$, $n_{1,a}$, being $p_{12}$ and $q_{12}$ overlaps between different replicas of the system (see the SM). These expectation values can be obtained by looking for the stationary points of the quenched pressure $\nabla_{\bar n,\bar q,\bar p}  A_{N,M}(\alpha,\beta) =0$ and turn out to fulfil the following self-consistent equations 
\begin{eqnarray}
  \label{eq:sce_n}
  \bar n &=& \mathbb{E}_{\phi\chi}\left(\frac{1}{M}\sum_{a=1}^M \chi_{a}\right) \tanh \big(\beta \bar n \sum_{a=1}^M \chi_{a} +\sqrt{\alpha\beta  \bar p M}\phi\big), \\
  \bar q &=& \mathbb{E}_{\phi\chi} \tanh^{2} \big(\beta \bar n \sum_{a=1}^M \chi_{a} +\sqrt{\alpha\beta  \bar p M}\phi\big),\\
  \label{eq:sce_p}
  \bar p &=& \frac{\beta \bar q}{[1-\beta(1-\bar q)]^{2}}.%\\
 % \bar m &=&  \mathbb{E}_{\phi\chi} \tanh \big(\beta \bar \eta \sum_{a=1}^M \chi_{a} +\sqrt{\alpha\beta  \bar p M}\phi\big).
\end{eqnarray}
Note that,  the example magnetization $n$ is embedded right in the expression of the model cost-function (see \eqref{eq:Hamiltonian}), much as the Mattis magnetization for the standard AGS theory.
On the other hand, $m$ does not play as a natural observable for the model as the system is, in principle, unaware of the archetypes. Having access to the archetypes, a practical way to compute $\bar{m}$ is to insert in the model a small field $J$ coupled to $m$ and then evaluate $\partial_J  A_{N,M}(\alpha,\beta,J)$ as $J \to 0$. More interestingly, as shown in the SM, in the limit of large $M$,   $\bar m$ spontaneously emerges and occurs to be directly related to $\bar n$; in particular, 
the two magnetizations, $\bar{m}$ and $\bar{n}$, get related as 
\begin{equation} \label{eq:relation}
\bar{n} = \frac{\bar{m} r}{1-\beta  (1-\bar q) (1- r^2)}.
\end{equation} 

In the next subsections we analyze the self-consistent equations under different conditions and try to derive analytically the existence of a threshold size $M_{\times}$ and of a critical size $M_c$ that determine the onset of different regimes as for the system ability to generalize.

\subsubsection{Finite dataset size} \label{sec:phase_diagram}
Let us resume eqs.~\eqref{eq:sce_n}-\eqref{eq:sce_p} and let us focus on the zero fast noise limit $\beta \to \infty$. Recalling that $M$ is large, we can introduce the random variable $S := \frac{1}{M} \sum_{a=1}^{M} \chi_a = r+\sqrt{\frac{1-r^2}{M}} Z ~~\rm{with} ~~ Z \sim \mathcal{N}(0,1)$, and, posing
\begin{equation}
\delta \bar Q  = \mathbb{E}_{Z} \frac{2}{\sqrt{\pi}}\exp\left[-\big(\frac{\bar n M S(Z)}{\delta \bar Q + \sqrt{2\alpha M}}\big)^2\right],
\end{equation}
where $\mathbb{E}_{Z}$ denotes the expectation over $Z$, the self-consistency equations for the magnetizations $\bar m$ and $\bar n$ become
\begin{eqnarray}
\bar  n &=& \mathbb{E}_{Z} S(Z)\operatorname{erf}\left(\frac{\bar  n M S(Z)}{\delta \bar Q + \sqrt{2\alpha M}}\right), \\
\bar m &=& \mathbb{E}_{Z}\operatorname{erf}\left(\frac{\bar  n M S(Z)}{\delta \bar Q + \sqrt{2\alpha M}}\right).
\end{eqnarray}
Via these equations it is possible to obtain an analytic expression for the threshold $M_{\times}$:  by requiring $\bar m > \bar  n$, we obtain the following inequality 
\begin{equation}
\mathbb{E}_{Z} \left[1-S(Z)\right]\operatorname{erf}\left(\frac{\bar  n M S(Z)}{\delta \bar Q + \sqrt{2\alpha M}}\right)>0
\end{equation}
which, to first order in $\bar n$, is satisfied if $\mathbb{E}_{Z} \left[1-S(Z)\right]S(Z)>0$,
namely, recalling the definition of $S$,
\begin{equation}
\mathbb{E}_{Z} \left[1-r - \sqrt{\frac{1-r^2}{M}}Z\right]\left[ r+\sqrt{\frac{1-r^2}{M}}Z \right]>0,
\end{equation}
whence $(1-r)r -  \frac{1-r^2}{M} \mathbb{E}_{Z}Z^2>0$. The latter inequality yields to
\begin{equation} \label{eq:threshold}
M >  \frac{1+r}{r} = M_{\times}. 
\end{equation}
Therefore, as expected, in order for the archetype magnetization to prevail over the example magnetization, the dataset size needs to be larger and larger as the sample gets more and more blurred, according to the above scaling.  This finding is corroborated by extensive computational checks and its robustness with respect to the fast noise is also tested, as we solved numerically the self-consistency equations for arbitrary, finite $\beta$ and derived an estimate of $M_{\times}$ by comparing the solutions of $\bar m$ and $\bar n$ obtaining analogous results as reported in Fig.~\ref{fig:soft_trans}.

Finally, we tested the validity of these results in the RBM framework, also exploring the robustness with respect to different loads. In the left panels of Fig.~\ref{fig:laSei} we compare the classification probabilities $\mathcal P(\boldsymbol z_E |  \boldsymbol \sigma = \boldsymbol \xi_E)$ and $\mathcal P(\boldsymbol z_E |  \boldsymbol \sigma = \boldsymbol \eta_E)$ versus $M$ and for different choices of $r$. The two probabilities display a monotonical behaviour as a function of $M$ that is, respectively, increasing and decreasing. This can be intuitively explained invoking the central limit theorem and recalling Fig.~\ref{fig:minima}: as $M$ increases, minima in the energy landscape become denser and denser in such a way that the system may eventually fall into a state other than $\boldsymbol \eta$, and this gets more and more likely as the dataset quality $r$ is lower. In the right panel of Fig.~\ref{fig:laSei} the threshold values obtained for different loads $\alpha$, analytically (i.e., investigating the Hopfield network, see Eq.~\eqref{eq:threshold}) and numerically (i.e., investigating the RBM) are shown to be perfectly consistent.

\begin{figure}[tb]
    \centering
    \includegraphics[scale=0.475]{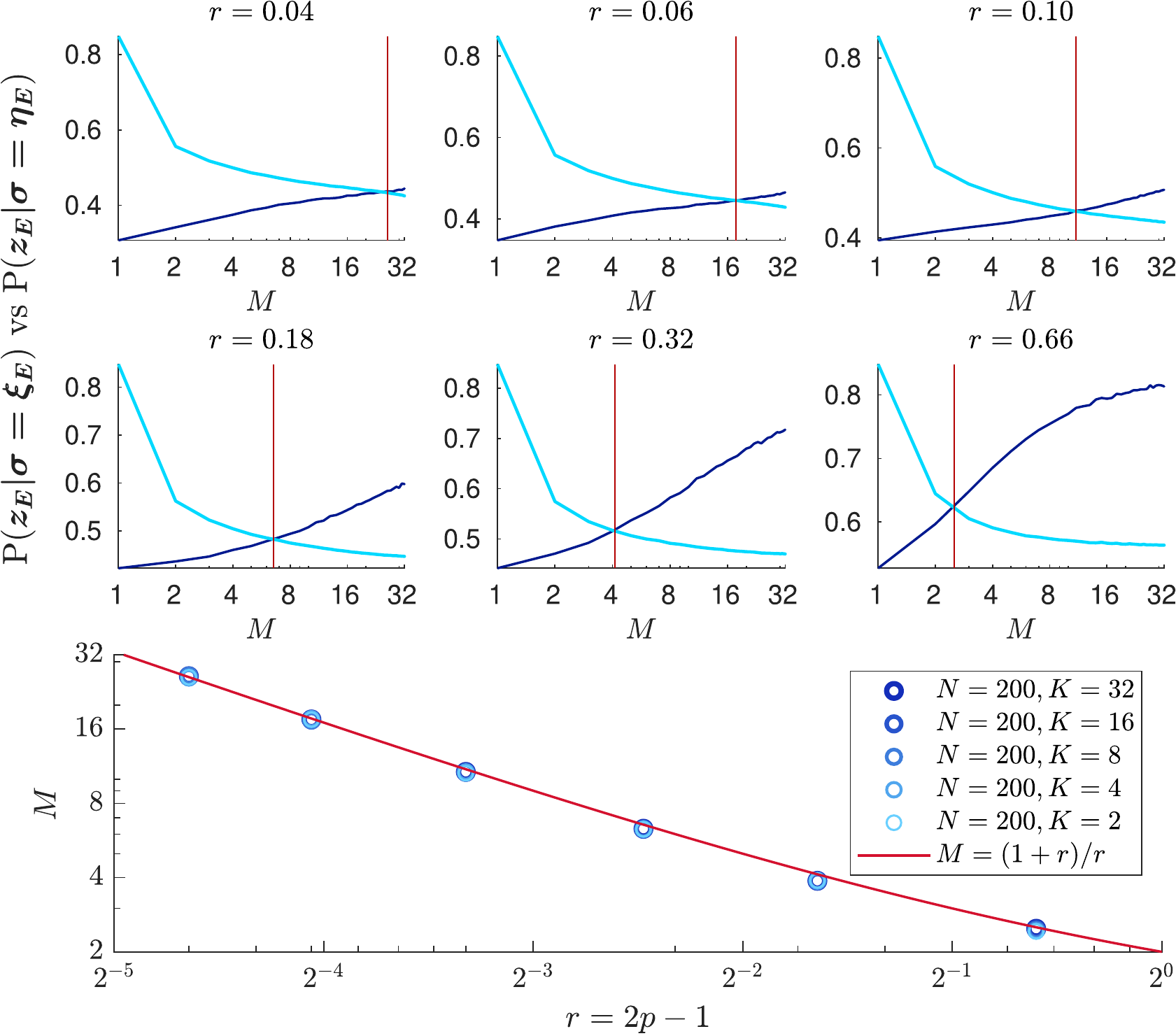}
    \caption{Each plot shows the probability of correctly classifying either an example or an archetype: blue lines are computational and drawn from Boltzmann machine learning (dark blue for the archetype, i.e. $\mathcal P(\boldsymbol z_E | \boldsymbol \sigma = \boldsymbol \xi_E)$, light blue for the example, i.e. $\mathcal P(\boldsymbol z_E | \boldsymbol \sigma = \boldsymbol \eta_E)$) while red lines are theoretical and draft from Hopfield network learning.  Different plots show different noise $r=2p-1$ levels and the vertical red line is evaluated via the relation $M_{\times} = (1+r)/r$. Right: This phase diagram shows two regions split by the threshold line $M_{\times} =  (1+r)/r$ (where $\bar{m}=\bar{n}$) and above that threshold the concept of the archetype emerges (and the network can successfully generalizes) while below a fuzzy misture where all the examples still preserve their characteristics persists. The threshold is shown to be universal: different values of $\alpha$ are computationally simulated for Boltzmann learning and shown as spots (different blue circles), while the theoretical prediction by the Hopfield network is presented as a continuous red line and the two perfectly coincide as expected.  On the vertical axes we report the critical size of the training set required for a successful learning while on the horizontal axes we report the degree of noise in the data-set.} \label{fig:laSei}
\end{figure}

\subsubsection{Infinite dataset size} \label{sec:phase_diagram2}
Let us now retain a finite noise $\beta$ and apply the rescaling of the noise $\beta \to  \frac{\beta }{r^2 + \beta  (1 -q) \left(1 - r^2\right)}$ to eqs.~\eqref{eq:sce_n}-\eqref{eq:sce_p}, thus, we reach expressions for the magnetization and the overlap whose content finally shines:
\begin{eqnarray}\small 
\label{eq:mm}
\bar m &=&  \mathbb{E}_{Z} \tanh \left[\beta \bar m M + Z \beta \sqrt{M \frac{1-r^{2}}{r^{2} }{\bar m}^{2} +\frac{  \alpha  \bar p}{r^4 \beta }M} \right]\\ \small \nonumber
\label{eq:qq}
\bar q &=&  \mathbb{E}_{Z} \tanh^{2} \left[\beta \bar m M + Z \beta \sqrt{M \frac{1-r^{2}}{r^{2} }{\bar m}^{2} +  \frac{  \alpha \bar p}{r^4 \beta }M} \right],
\end{eqnarray}
where $Z \sim \mathcal{N}(0,1)$ and $\bar{p}$ was given in (\ref{eq:sce_p}). As arguments of the hyperbolic tangent there are now three contributions and no longer just two as in the standard AGS theory. Indeed, beyond the signal carried by $\bar{m}$ there are two sources of (slow) noise: a classic one, proportional to $\alpha$, \textcolor{black}{stemming from} the other patterns not retrieved (pattern interference), and a new one \textcolor{black}{stemming from} the examples \textcolor{black}{making up the sample} related to the pattern (example interference). Note that, as consistency check,  if the network is not provided with datasets, but just noiseless patterns (i.e. $M=1$ and $r=1$), the whole theory collapses over the standard AGS one of the Hopfield model as it should. Further we stress that at $\alpha=0$ there is not a real phase transition (as a glance at these self-consistencies reveal), rather we need $\alpha>0$ (namely examples of different archetypes produce reciprocal attenuation of their retrieval, promoting as a result the emergence of the archetypes themselves).
We now inspect in more details the self-consistency for $\bar m$ and we check when the signal contribution prevails over the noise, 
namely we require that
$\beta M \bar m> \beta \sqrt{M} |Z| \sqrt{  {\bar m}^{2} (1-r^{2}) / r^{2}+ \alpha \bar{p}/(r^4 \beta)} $
holds almost surely. A solution to this inequality is given by
 \begin{equation}
   \label{eq:solMdis1}
     M > \frac{\gamma^{2}}{r^{2}} \left[ 1-r^{2}+\frac{\bar q}{\bar m^{2} (1 - \beta (1-\bar q))^2}\frac{\alpha} {r^{2}}\right],
 \end{equation}
 where $\gamma$ assesses the confidence level (in fact, the last condition implies $|Z|<\gamma$ which can be satisfied up to an exceedingly small probability at finite $M$). \textcolor{black}{Setting $\beta \to \infty$ this} result recovers the scaling in \eqref{due}) obtained via signal-to-noise analysis. Therefore, a large enough database ensures the stability of the archetype.

\begin{figure}[tb]
\centering
\includegraphics[width=0.95\linewidth]{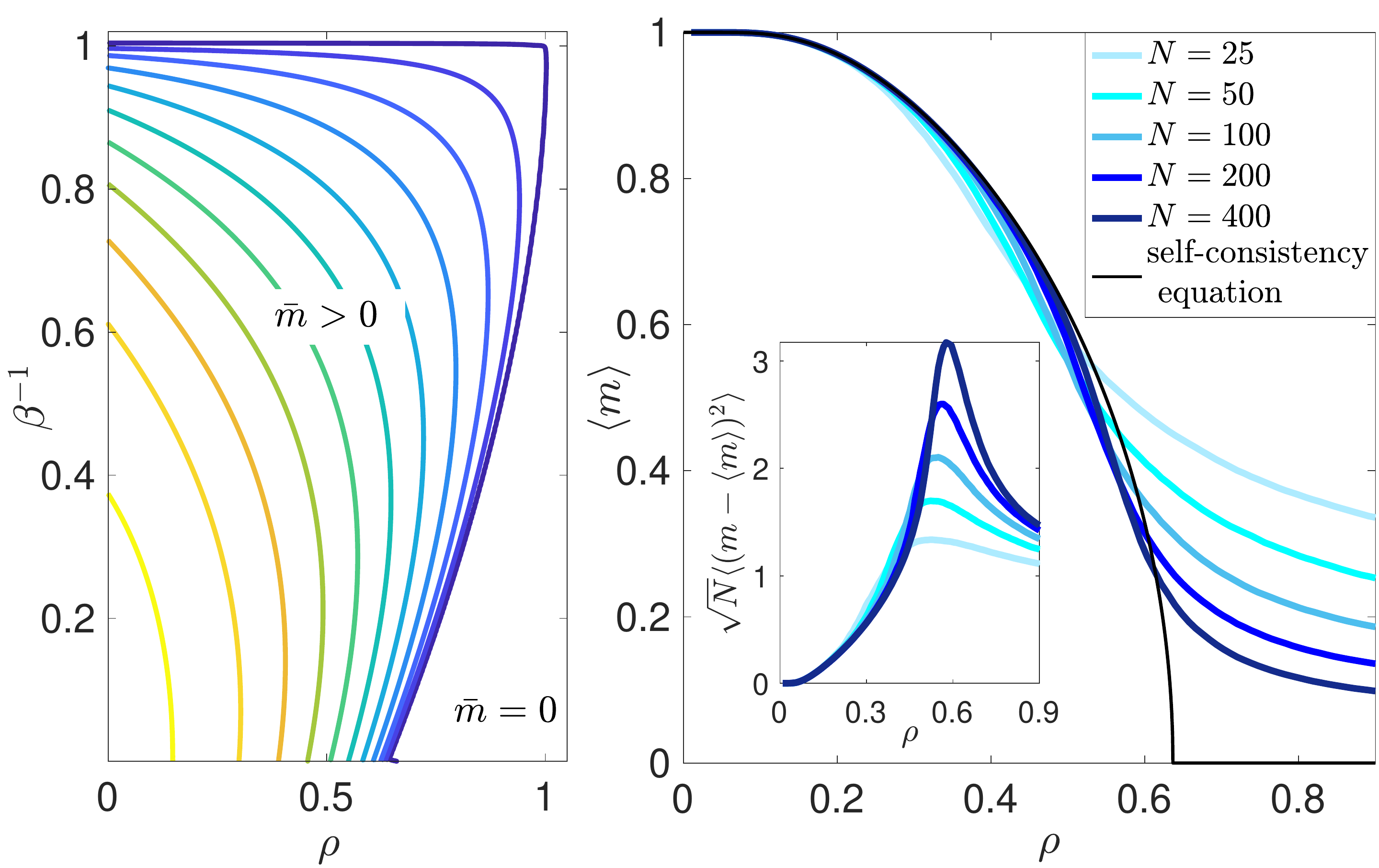}
\caption{\label{fig:scaling} Left panel: Phase diagram in the $(\beta, \rho)$ plane obtained by solving numerically equations \eqref{eq:criticalSCEbis}-\eqref{eq:criticalSCEter}. The outer, darkest line corresponds to the onset of a non-null magnetization $\bar m >0$, the remaining contour lines, in brighter and brighter colors, correspond to larger and larger values of magnetization. Right panel: the main plot shows a comparison
between the numerical solution of the self-consistency equation (\ref{eq:criticalSCEbis}) in the noiseless limit (thin and darkest solid line), and zero-temperature Monte Carlo runs at different sizes (from brighter to darker nuances, $N=25, 50, 100, 200, 400$, as shown by legend), while the inset shows the same finite-size-scaling for the susceptibility; in both figures, we set $\alpha=0.08$, $M=80$, and $r^2=\sqrt{\alpha/(M \rho)}$, and, to determine each point, a quenched average over \(50\) independent coupling matrices was performed.}
\end{figure}   

In order to evidence a possible, genuine phase transition, we have to study the limit \textcolor{black}{$(N,M,K)\to \infty, r \to 0$} and rephrase the whole theory intensive. In this limit we find that $\rho:=\alpha / (Mr^4)$ is a suitable control parameter (ruling the overall slow noise) able to trigger a phase transition and the self-consistency equations can be recast as
\begin{eqnarray}
  \label{eq:criticalSCEbis}
  \bar{m} &=& \mathbb{E}_Z \tanh (\beta \bar m + \beta Z \sqrt{\rho \bar q}) \xrightarrow{\beta \to \infty} \textrm{erf} \left( \frac{\bar m}{\sqrt{2 \rho}} \right),\\
    \label{eq:criticalSCEter}
  \bar{q} &=& \mathbb{E}_Z \tanh^2 (\beta \bar m + \beta Z \sqrt{\rho \bar q})  \xrightarrow{\beta \to \infty} 1.
  \end{eqnarray}
The numerical solution of eq.~\eqref{eq:criticalSCEbis} is sketched in Figure \ref{fig:scaling} (left panel), where we highlight a region in the $(\beta,\rho)$ plane where $\bar m$ is non null.

Focusing on the fast noiseless limit $\beta \to \infty$ and expanding at $\bar{m}=0$, a critical behavior is found at $\rho_{c} = \frac{2}{\pi}$ with the critical exponent $1/2$ (i.e., $\bar m \sim \sqrt{\frac{3}{\pi}} \sqrt{2 - \pi \rho}$ near the critical point): the concept is not abruptly formed, rather it stems gradually by a continuous contribution provided by all the examples.
\newline
The scenario painted above is corroborated by numerics: in Figure \ref{fig:scaling} (right panel) we plot a finite-size-scaling of the Mattis magnetization of the archetype, along with the related susceptibility, obtained via  Monte Carlo simulations. Signatures of criticality just occur at $\rho \approx \rho_c$, according to the theory.

\section*{Conclusions} \label{sec:conclusions}
The  a-priori knowledge  of the minimal data-set size to ensure a successful learning, is not yet known in general, despite the pivotal importance of such information en route toward optimized artificial intelligence.  In this work we try to contribute towards this goal and we restrict to the simplest random data-sets scenario, where shallow networks suffice and a general theory can be worked out. 
\newline
First, we prove that the supervised Boltzmann learning based on the grandmother-cell setting mirrors unsupervised Hopfield learning: the grand-mother cell scenario \cite{Nonna1,Nonna2} was originally introduced in a biological context (and adopted here to the machine learning counterpart) and it assumes that one single neuron (here one single hidden neuron) gets active when a pattern is presented to the network (here when a pattern is inputted in the visible layer); while this theory was criticized in the biological context, this simple setting naturally works here for structureless data-sets as all the patterns are equivalent under permutations and we can arbitrarily associate any of the patterns to be learnt to any of the hidden neurons; the unsupervised Hopfield learning generalizes the standard Hopfield model in the case where, instead of having a set of definite patterns (archetypes), only a sample of blurred versions are available and these are overall combined in a Hebbian kernel.
\newline
Next, we show numerically that these two learning schemes are successful as long as the training dataset is large enough and the related threshold sizes are the same.
This corroborates the formal equivalence between the two systems and we therefore proceeded with the analytical investigation of the Hopfield networks; note that, by applying Hubbard-Stratonovich  transformation to the Hopfield partition function (\ref{def:Z_H})  de facto with end up with the partition function of a Boltzmann machine equipped with gaussian prior, thus the related findings are then safety interpreted also in the Boltzmann learning framework.  We find that, as the number $M$ of examples provided grows, at first examples are wrongly stored as archetypes, but, as the data-set size gets large enough, the system eventually builds its own representation of the archetypes lying behind the  provided information: this signs the onset of a successful training. Clearly the larger the noise in the data-set the larger the required data-set size and we find out sharp scalings to be respected in order for learning to take place successfully:  the crossover among  examples vs archetype retrieval -- the formation of the concept of the archetype -- happens at $M_{\times} \sim r^{-1}$, archetype stability requires  $M_c \sim r^{-4}$.
\newline
We remark that, beyond the theory at finite $M$ (pivotal for practical purposes as infinite volume neural networks or data-sets are not available), we have also rephrased the whole approach in the $M,N,K \to \infty$ and $r \to 0$ limit: by introducing an effective control parameter $\rho=K/MNr^4$  we have shown that there exists a critical $\rho_c= 2/\pi$ where the Mattis magnetization of the archetype continuously raises from zero, accompanied by a divergence of the relative susceptibility: there is a true phase transition underlying  the formation of the concept of the archetype and it is of second-order, highlighting that this inference process is of continuous nature. 
\newline 
Finally, as there are arguments by which a random theory can be of vague utility (mainly due to the structure-less nature of the data-sets), yet there are at least three reasons by which it can be a meaningful starting point for future explorations:  it is universal  (we do have just one theory for the random scenario), it acts as a {\it bound} for structured theories (as, for a trivial Shannon compression argument, if the network is able to cope with  $K$ random patterns, it would be possibly able to cope with structured patterns), finally an argument of historical and methodological continuity: retrieval theory, namely the celebrated AGS theory, is a random theory too \cite{AGS}. 
\newline
Beyond the computational reward of the a-priori knowledge of the minimal data-set size for training, we believe that the whole approach is by itself another dowel in the mosaic en route toward a systemic theory of Artificial Intelligence, where all its emerging capabilities can be explained, that is under construction via the statistical mechanics of complex systems \cite{Florent1,Florent2,Zecchina,Structure1,Lenka,Montanari}.

%
%\section*{Data availability}
%The data-sets  that are used in this study are fully available upon reasonable request. Request for data access should be addressed to Loretta L. del Mercato.
%
%
%\section*{Code availability}
%The code and models that are described in this study are fully available upon reasonable request. Request for code access should be addressed to Adriano Barra.

\section*{Acknowledgments}\label{acknowledgments}
EA and GD acknowledge partial financial support from Sapienza University of Rome (RM120172B8066CB0).
\newline
AB  and FA are grateful to MUR (PRIN 2017, Project no. 2017JFFHSH) and to UniSalento (Prot. n. 148919-III8) for financial support and to  INFN, Sezione di Lecce  (FIELDTURB) for providing computational facilities.
 
\section*{Competing interests statement}
The Authors declare no competing interests.

\newpage
\onecolumngrid

\begin{center}
\bf Supplementary Information
\end{center}
\appendix
\section{Signal to noise approach}
Although not permitting a sharp control as statistical mechanics does, the signal-to-noise technique still represents an optimal tradeoff to start examining the system with relatively cheap analytical and computational expenses. In the next subsec.~\ref{ssec:low}, we study the low storage regime (i.e., $\alpha=0$) whose inspection can be achieved by considering solely one archetype pattern, along with the related set of $M$ examples; in the following subsec.~\ref{ssec:high}, we focus on the high storage (i.e. $\alpha >0$) where $K=\alpha N$ patterns are taken into account, along with the related $K \times M$ examples, $M$ examples each pattern. In all the cases, the network shall be built of by $N$ Ising neurons.

\subsection{Single archetype (low storage)} \label{ssec:low}
Let us consider an archetype pattern $\boldsymbol{\xi}$ that we would like to be spontaneously stored by the Hopfield network upon providing it with a set of $M$ noisy examples $\{ \boldsymbol{\eta}^{a} \}_{a=1,...,N}$ of the archetype. In order to analytically approach the problem we make the following schematization for the noisy examples
	\[
		\begin{cases}
			\eta^{a}_i=\xi_i\quad\text{with probability}\quad p\\
			\eta^{a}_i=-\xi_i\quad\text{with probability}\quad 1-p.
		\end{cases}
	\]
	In other words the noisy examples are obtained flipping some of the components of the archetype such that for $p \to 0$ and for $p \to 1$ there is no noise and the example matches perfectly the archetype (or its spin-flipped dual version), while for $p \to 1/2$ any reminiscence of the archetype in the examples gets lost. Since the network is fed by the noisy patterns, according to the Hebbian prescription we can introduce the coupling $J_{ij}$ as
	\begin{equation}
		J_{ij}=\frac{1}{N}\sum_{a=1}^M\eta^{a}_i\eta^{a}_j.
		\label{eq:couplings_single}
	\end{equation}
	%where $M$ is the number of noisy patterns, while $N$ the size of the network (number of neurons).
	The questions we want to answer to are $i.$ whether the archetype is dynamically stable and $ii.$ whether the examples are dynamically stable.
	\newline
	 In order to understand these points we firstly analyze the statistics of couplings.
		The average value of the coupling $J_{ij}$ is
		\[
			\mean*{J_{ij}}=\frac{1}{N}\sum_{a=1}^M\mean*{\eta^{a}_i\eta^{a}_j},
		\]
		where $\mean{\cdot}$ is the average over the noise affecting the patterns. Since the noise acts over single entries independently and identically, and since the self-interactions are absent, we can factorize the expectation so to obtain
		\begin{equation}
			\mean*{J_{ij}}=\frac{1}{N}\sum_{a=1}^M\mean*{\eta^{a}_i}\mean*{\eta^{a}_j}=\frac{1}{N}\sum^M_{a=1}\qua*{p\xi_i-(1-p)\xi_i}\qua*{p\xi_j-(1-p)\xi_j},
		\end{equation}
		that is
		\begin{equation}
			\mean*{J_{ij}}=\frac{M}{N}\ton*{4p^2+1-4p}\xi_i\xi_j=\frac{4M}{N}\ton*{p-\frac{1}{2}}^2\xi_i\xi_j.
			\label{eq:mean_J_single}
		\end{equation}
		If the archetype were available, the couplings $J_{ij}^0$ obtained using it would be
		\[
			J_{ij}^0=\frac{1}{N}\xi_i\xi_j,
		\]
		in such a way that the expectation of the coupling stemming from examples can be written as
		\[
			\mean*{J_{ij}}=4M\ton*{p-\frac{1}{2}}^2J_{ij}^0.
		\]
		In order to inspect the stability of the archetype, we have also to consider the fluctuations of the noisy couplings: this can be easily done computing their variance; by definition
		\[
			\var\qua*{J_{ij}}=\mean*{J_{ij}^2}-\mean*{J_{ij}}^2.
		\]
		Let us consider the squared couplings, we can write them as
		\[
			J_{ij}^2=\frac{1}{N^2}\left(\sum_{a=1}^M\eta^{a}_i\eta^{a}_j \right)^2=\frac{1}{N^2}\sum_{a,b=1}^M\eta^{a}_i\eta^{a}_j\eta^{b}_i\eta^{b}_j,
		\]
		that is
		\[
			J_{ij}^2=\frac{1}{N^2} \left(M+\sum_{\substack{a, b= 1 \\ a \neq b}}^{M}\eta^{a}_i\eta^{a}_j\eta^{b}_i\eta^{b}_j \right).
		\]
		We can now take the expectation over the noise, but we have to consider the fact that $\eta^{a}_i$ and $\eta^{b}_i$ are correlated and, as a consequence, the average can not be fully factorized
		\[
			\mean*{J_{ij}^2}=\frac{M}{N^2}+\frac{1}{N^2}\sum_{\substack{a, b=1 \\ a \neq b}}^{M}\mean*{\eta^{a}_i\eta^{b}_i}\mean*{\eta^{a}_j\eta^{b}_j}.
		\]
		The correlation between the two realizations of the noise is
		\begin{equation}
			\mean*{\eta^{a}_i\eta^{b}_i}=p^2\xi_i^2+(1-p)^2(-\xi_i)^2+2p(1-p)\xi_i(-\xi_i)=4\ton*{p-\frac{1}{2}}^2
			\label{eq:correlation_noise_noise}
		\end{equation}
		and so we get
		\[
			\mean*{J_{ij}^2}=\frac{M}{N^2}\qua*{1+16(M-1)\ton*{p-\frac{1}{2}}^4}.
		\]
		We have all the ingredients for computing the variance, recalling Eq.~\eqref{eq:mean_J_single} we obtain
		\begin{eqnarray}
		\notag
			\var\qua*{J_{ij}}&=&\frac{M}{N^2}\qua*{1+16(M-1)\ton*{p-\frac{1}{2}}^4}-\frac{16M^2}{N^2}\ton*{p-\frac{1}{2}}^4\\
		\label{eq:variance_J_single}
			&=&\frac{M}{N^2}\qua*{1-16\ton*{p-\frac{1}{2}}^4}.
		\end{eqnarray}

		In order to understand if the archetype is stable we have to determine if the expectation is larger or smaller than the square root of the variance. Indeed, if the variance is small we can safely set $J_{ij}\approx\mean*{J_{ij}}$ and, for the aforementioned considerations, the archetype is stable. Conversely, if the variance is large, the couplings will substantially be random variables not carrying any signal; in this case the archetype is then expected not to be dynamically stable. Using Eqs.~\eqref{eq:mean_J_single} and \eqref{eq:variance_J_single} we obtain the following equation for the crossover between these two situations
		\[
			\mean*{J_{ij}}=\sqrt{\var\qua*{J_{ij}}},
		\]
		that is
		\[
			\frac{4M}{N}\ton*{p-\frac{1}{2}}^2\xi_i\xi_j=\frac{\sqrt{M}}{N}\sqrt{1-16\ton*{p-\frac{1}{2}}^4}.
		\]
		Solving this equation we obtain an expression for the crossover value of $M$, referred to as $M_c$:
		\[
			M_c=\frac{1-4\ton*{p-\frac{1}{2}}^2}{4\ton*{p-\frac{1}{2}}^2}.
		\]
		As one would expect, $M_c$ diverges for $p\to\frac{1}{2}$, because more and more examples are necessary for inferring the archetype if almost half the components are flipped. We also notice that $M_c$ is symmetric around $p=\frac{1}{2}$, because, as a consequence of the Hebbian rule, storing a pattern or its flipped version is the same.
	\subsubsection{Stability of the archetype}
		\label{subsec:stability_archetype}
		In order to understand if the archetype is stable under the dynamics induced by the couplings $J_{ij}$ defined above, we have to consider the local field $h_i$ acting on it. It holds
		\[
			h_i\xi_i=\frac{1}{N}\sum_{a=1}^M\sum_{j\neq i}^N\eta_i^a\eta_j^a\xi_i\xi_j
		\]
		and taking the expectation over the noise we obtain
		\[
			\mean*{h_i\xi_i}=\frac{1}{N}\sum_{a=1}^M\sum_{j\neq i}^N\mean*{\eta_i^a\xi_i}\mean*{\eta_j^a\xi_j}.
		\]
		The correlation between the archetype and one of its noisy versions is
		\begin{equation}
			\mean*{\eta_i^a\xi_i}=p\ton*{\xi_i}^2+(1-p)\ton*{-\xi_i\xi_i}=2\ton*{p-\frac{1}{2}},
			\label{eq:correlation_noise_archetype}
		\end{equation}
		this yields
		\begin{equation}
			\mean*{h_i\xi_i}=4M\ton*{p-\frac{1}{2}}^2.
			\label{eq:mean_field_archetype}
		\end{equation}
		This quantity is always positive meaning that, on average, the archetype is dynamically stable. However, as in the case of the couplings, we have to consider also the variance of $h_i\xi_i$, so to determine if it is a self-averaging quantity. We can write
		\[
			\ton*{h_i\xi_i}^2=h_i^2=\frac{1}{N^2}\sum_{a=1}^M\sum_{j\neq i}^N\sum_{b=1}^M\sum_{k\neq i}^N\eta_i^a\eta_j^a\xi_j\eta_i^b\eta_k^b\xi_k.
		\]
		In order to average this quantity we have to take out the self-interactions from the summation. Multiplying each addend by the factor $\qua*{(1-\delta_{ab})+\delta_{ab}}\qua*{(1-\delta_{jk})+\delta_{jk}}$ we get four terms
		\begin{equation}
		\notag
			h_i^2=\frac{1}{N^2} \left[MN+\sum_{j\neq i}^N\sum_{k\neq i, j}^N\sum_a^M\xi_j\eta_j^a\xi_k\eta_k^a+
			+ \sum_{j\neq i}^N\sum_a^M\sum_{b\neq a}^M\eta_j^a\eta_j^b\eta_i^a\eta_i^b+
		\sum_{j\neq i}^N\sum_{k\neq i, j}^N\sum_a^M\sum_{b\neq a}^M\xi_j\eta_j^a\xi_k\eta_k^b\eta_i^a\eta_i^b\right].
		\end{equation}
		Each summation does not contain any self-interaction term, so we can take the expectation over the noise strightforwardly; exploiting Eqs.~\eqref{eq:correlation_noise_noise} and \eqref{eq:correlation_noise_archetype} we obtain
		\begin{equation}
			\mean*{\ton*{h_i\xi_i}^2}=\frac{1}{N^2}\left[MN+4MN^2\ton*{p-\frac{1}{2}}^2+
			16M^2N\ton*{p-\frac{1}{2}}^4+16M^2N^2\ton*{p-\frac{1}{2}}^4\right].
		\end{equation}
		Using this result and Eq.~\eqref{eq:mean_field_archetype} we can now compute the variance
		\begin{equation}
			\var\qua*{h_i\xi_i}=\mean*{\ton*{h_i\xi_i}^2}-\mean*{h_i\xi_i}^2=
			 \frac{M}{N}+4M\ton*{p-\frac{1}{2}}^2+16\frac{M^2}{N}\ton*{p-\frac{1}{2}}^4.
			\label{eq:variance_field_archetype}
		\end{equation}

		The archetype is stable only if the mean value is larger than the variance, otherwise the fluctuations of the noise are dominating. Combining Eqs.~\eqref{eq:mean_field_archetype} and \eqref{eq:variance_field_archetype} we obtain the following condition determining if the network successfully retrieves the archetype
		\[
			\mean*{h_i\xi_i}>\sqrt{\var\qua*{h_i\xi_i}},
		\]
		that is
		\[
			16\ton*{p-\frac{1}{2}}^4\ton*{1-\frac{1}{N}}M^2-\qua*{\frac{1}{N}+4\ton*{p-\frac{1}{2}}^2}M>0.
		\]
		Being the coefficient of $M^2$ always positive, the conclusion is that the archetype is stable provided that
		\begin{equation}
			M>\frac{\frac{1}{N}+4\ton*{p-\frac{1}{2}}^2}{16\ton*{p-\frac{1}{2}}^4\ton*{1-\frac{1}{N}}}\approx\frac{\frac{1}{N}+4\ton*{p-\frac{1}{2}}^2}{16\ton*{p-\frac{1}{2}}^4} = M_{a,1}.
			\label{eq:stability_archetype}
		\end{equation}
		This expression shows that if the sample is sufficiently large the archetype is stable; more precisely, the sample size has to be related to the underlying degree of noise: in the large $N$ limit and in the low-load ($\alpha=0$) regime under consideration, we need to scale the number of examples as $M \propto 1/(2p-1)^2$ in order to reliably store the archetype; as we will see, this is no longer true in the high load ($\alpha >0$) where we  will need $M \propto 1/(2p-1)^4$.

The estimate given in \eqref{eq:stability_archetype} is successfully compared with simulations in Fig.~\ref{fig:esempinumericiK1} (left panel).

\begin{figure}[t!]
	\centering
    \includegraphics[width=0.45\textwidth]{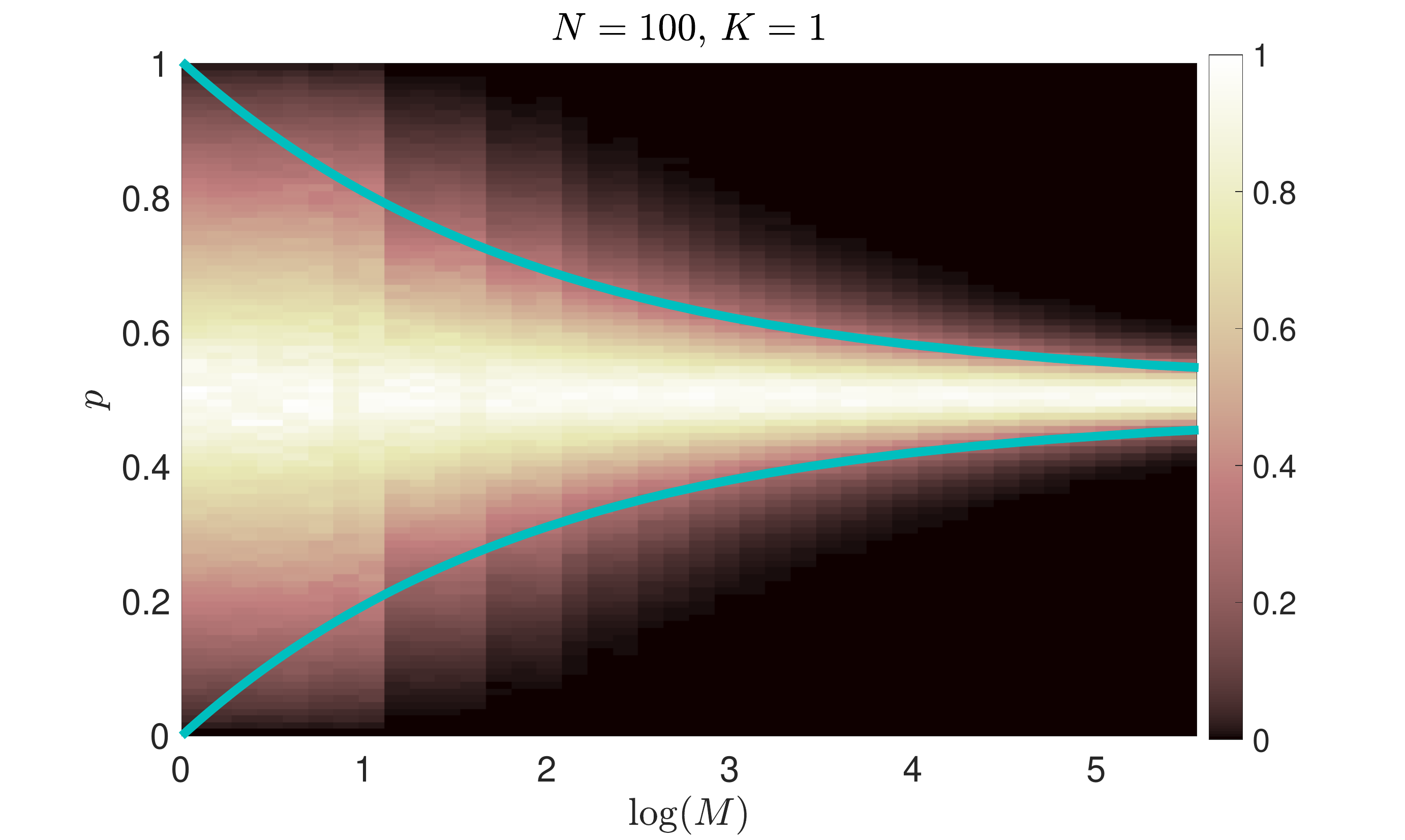}
    \includegraphics[width=0.45\textwidth]{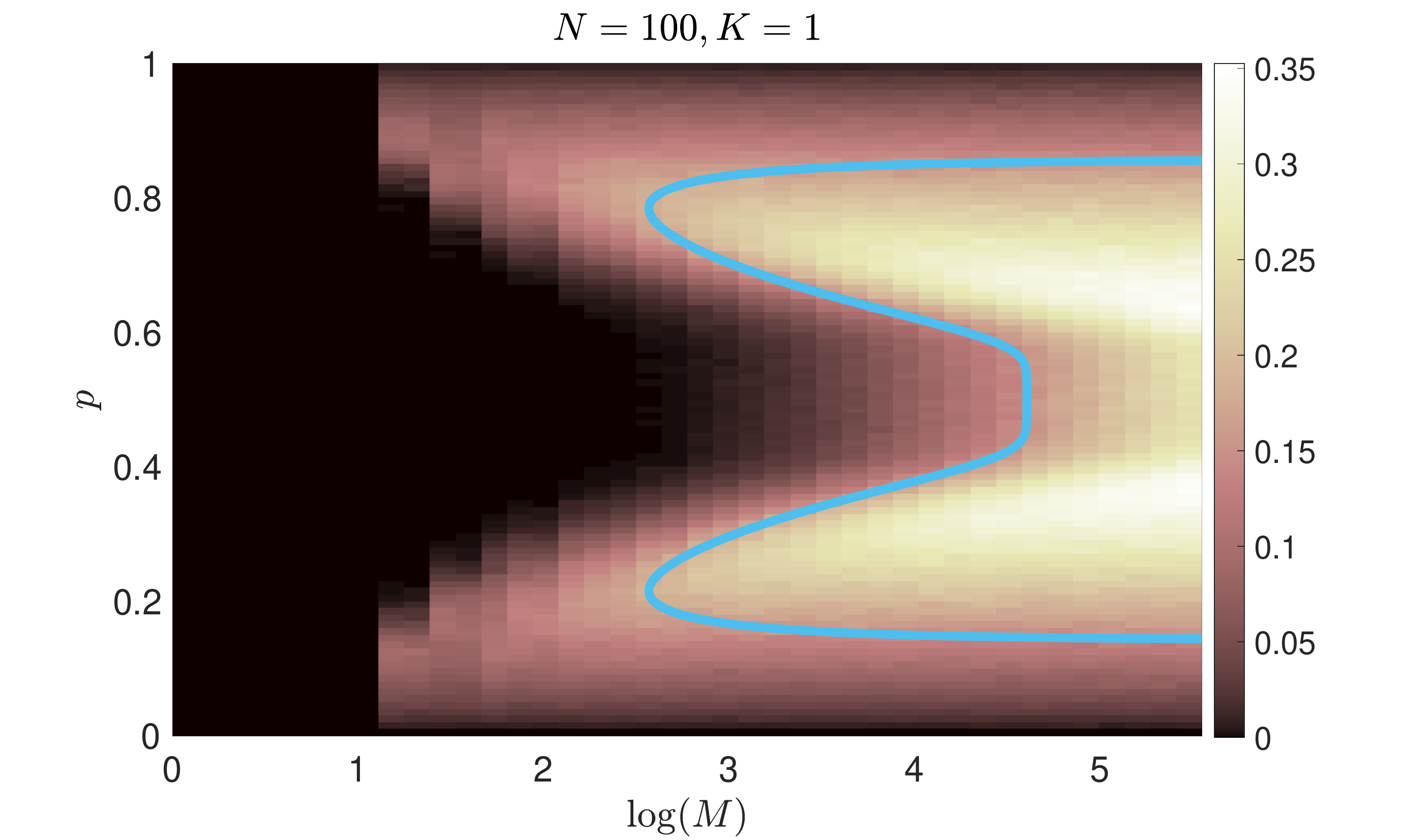}
	\caption{
	Left: Archetype stability versus $p$ and $M$.
The colormap shows the fraction of neurons that occur to be unstable when set in the configuration $\boldsymbol \sigma = \boldsymbol \xi$, as a function of $p$ and of $M$ (logarithmic axis); here $\boldsymbol \xi$ is the unique archetype which the $M$ examples refer to and data depicted have been obtained by averaging over a sample of $100$ different realizations of $\boldsymbol \xi$ and related examples, while the number $N$ of neurons is set equal to 100. The solid lines represent the roots of equation \eqref{eq:stability_archetype} and correctly demarcate the region of instability corresponding to values of $p$ close to $1/2$.
Right: Example stability versus $p$ and $M$.
The colormap shows the fraction of neurons that occur to be unstable when set in the configuration $\boldsymbol \sigma = \boldsymbol \eta$, as a function of $p$ and of $M$ (logarithmic axis); here $\boldsymbol \eta$ is one of the $M$ examples related to the unique archetype and data depicted have been obtained by averaging over a sample of $100$ different realizations of $\boldsymbol \xi$ and related examples, while the number $N$ of neurons is set equal to 100. The solid line represents the root of equation (\ref{eq:cons}) and correctly demarcates the region of instability corresponding to values of $p$ close to $1/2$ or relatively large values of $M$.}\label{fig:esempinumericiK1}
\end{figure}

	\subsubsection{Stability of the noisy examples}
		In the previous subsection we have proven that, as long as $M$ is sufficiently large, the network can store the archetype, however, we should also consider the dynamical stability of the noisy patterns we used for training the network. Indeed, we would like the network to learn the archetype \emph{and} to forget the noisy patterns. In order to understand if this is the case we proceed as done before, analyzing the dynamical stability of the arbitrary example $\boldsymbol \eta^c$. For the $i$-th spin we get
		\begin{equation}
			h_i\eta_i^c =\frac{1}{N}\sum_{a=1}^M\sum_{j\neq i}^N\eta_j^a\eta_i^a\eta_j^c\eta_i^c
			=\frac{1}{N}\sum_{a=1}^M\sum_{j\neq i}^N\eta_j^a\eta_i^a\eta_j^c\eta_i^c\qua*{\ton*{1-\delta_{ac}}+\delta_{ac}}
			=1+\frac{1}{N}\sum_{a\neq c}^M\sum_{j\neq i}^N\eta_j^a\eta_i^a\eta_j^c\eta_i^c.
		\end{equation}
		Taking the expectation over the noise and recalling Eq.~\eqref{eq:correlation_noise_noise} we obtain
		\begin{equation}
			\mean*{h_i\eta_i^c}=1+16(M-1)\ton*{p-\frac{1}{2}}^4,
			\label{eq:mean_field_noise}
		\end{equation}
		as expected, for $M=1$, this quantity reduces to one. We now turn to the variance, which reads as
		\[
			\ton*{h_i\eta_i^c}^2=h_i^2=\left( \frac{1}{N}\sum_{a=1}^M\sum_{j\neq i}^N\eta_j^a\eta_i^a\eta_j^c \right)^2.
		\]
		Multiplying each addend by $\qua*{\ton*{1-\delta_{ac}}+\delta_{ac}}$ we can recast this quantity as
		\begin{equation}
		\notag
			\ton*{h_i\eta_i^c}^2=\frac{1}{N^2}\qua*{N\eta_i^c+\sum_{a\neq c}^M\sum_{j\neq i}^N\eta_j^a\eta_i^a\eta_j^c}^2
			= 1+\frac{2}{N}\sum_{a\neq c}^M\sum_{j\neq i}^N\eta_j^a\eta_i^a\eta_j^c\eta_i^c
			+\frac{1}{N^2}\sum_{a\neq c}^M\sum_{b\neq c}^M\sum_{j\neq i}^N\sum_{k\neq i}^N\eta_j^a\eta_i^a\eta_j^c\eta_k^b\eta_i^b\eta_k^c,
		\end{equation}
		that is
		\[
			\ton*{h_i\eta_i^c}^2=2h_i\eta_i^c-1+\frac{1}{N^2}\sum_{a\neq c}^M\sum_{b\neq c}^M\sum_{j\neq i}^N\sum_{k\neq i}^N\eta_j^a\eta_i^a\eta_j^c\eta_k^b\eta_i^b\eta_k^c.
		\]
		Splitting the summation into four terms by introducing the factor $\qua*{\ton*{1-\delta_{ab}}+\delta_{ab}}\qua*{\ton*{1-\delta_{jk}}+\delta_{jk}}$ we get rid of the self-interaction terms arriving at the following expression
\begin{eqnarray} \nonumber
\ton*{h_i\eta_i^c}^2&=&2h_i\eta_i^c-1+\frac{1}{N^2} [MN+\sum_{a\neq c}^M\sum_{b\neq a, c}^M\sum_{j\neq i}^N\eta_i^a\eta_i^b\eta_j^a\eta_j^b \\
			 &+&\sum_{a\neq c}^M\sum_{j\neq i}^N\sum_{k\neq i,j}^N\eta_j^a\eta_j^c\eta_k^a\eta_k^c+
			\sum_{a\neq c}^M\sum_{b\neq a, c}^M\sum_{j\neq i}^N\sum_{k\neq i,j}^N\eta_i^a\eta_i^b\eta_j^a\eta_j^c\eta_k^b\eta_k^c].
\end{eqnarray}
		Using Eqs.~\eqref{eq:correlation_noise_noise} and \eqref{eq:mean_field_noise} we compute the expectation obtaining
		\begin{equation}
			\mean*{\ton*{h_i\eta_i^c}^2}= 1+32M\ton*{p-\frac{1}{2}}^4+\frac{M}{N}+16\frac{M^2}{N}\ton*{p-\frac{1}{2}}^4+16M\ton*{p-\frac{1}{2}}^4+64M^2\ton*{p-\frac{1}{2}}^6.
		\end{equation}
		Thus, recalling Eq.~\eqref{eq:mean_field_noise}, the variance is
		\begin{equation}
			\var\qua*{h_i\eta_i^c}=\frac{M}{N}+16\frac{M^2}{N}\ton*{p-\frac{1}{2}}^4+16M\ton*{p-\frac{1}{2}}^4+ 64M^2\ton*{p-\frac{1}{2}}^6-256M^2\ton*{p-\frac{1}{2}}^8.
			\label{eq:variance_field_noise}
		\end{equation}

		As usual, the condition for the dynamical stability of the noisy pattern is $\mean*{h_i\eta_i^c}>\sqrt{\var\qua*{h_i\eta_i^c}}$, and, exploiting Eqs.~\eqref{eq:mean_field_noise} and \eqref{eq:variance_field_noise}, this constraint reads
		\begin{eqnarray}\nonumber
			1+32M\ton*{p-\frac{1}{2}}^4+256M^2\ton*{p-\frac{1}{2}}^8
			&>& \frac{M}{N}+16\frac{M^2}{N}\ton*{p-\frac{1}{2}}^4+16M\ton*{p-\frac{1}{2}}^4\\
			\label{eq:cons}
			 &+&64M^2\ton*{p-\frac{1}{2}}^6-256M^2\ton*{p-\frac{1}{2}}^8.
		\end{eqnarray}
		We rewrite this expression as
		\[
			aM^2+bM+c>0,
		\]
		where
		\begin{equation}
			\begin{cases}
				a=512\ton*{p-\frac{1}{2}}^8-64\ton*{p-\frac{1}{2}}^6-\frac{16}{N}\ton*{p-\frac{1}{2}}^4\\
				b=16\ton*{p-\frac{1}{2}}^4-\frac{1}{N}\\
				c=1.
			\end{cases}
			\label{eq:coefficients}
		\end{equation}
With some algebra one can see that the previous inequality is always satisfied for $p$ that is either relatively large or relatively small (neglecting $1/N$ terms this is for $p<1/2-1/\sqrt{7}$ and $p>1/2+1/\sqrt{7}$). In between, the inequality can be satisfied provided that $M$ is relatively small. A comparison with numerical simulations is shown in Fig.~\ref{fig:esempinumericiK1} (right panel).

We can further deepen the behavior of the network as the size of the dataset varies by inspecting the values of the energies corresponding to configurations $\boldsymbol \sigma = \boldsymbol \xi$ and $\boldsymbol \sigma = \boldsymbol \eta^c$ (for arbitrary $c$); in fact, if we assume the network to be in one of those states we can simply compare the related energies without taking care of thermodynamic expectations and therefore no statistical mechanics is yet needed.

		Exploiting Eq.~\eqref{eq:couplings_single} we can write the Hamiltonian of the system as
		\begin{equation}
			\mathcal {H}_{N,M}(\boldsymbol \sigma | \boldsymbol \eta)=-\frac{1}{N}\sum_{i, j\neq i}^N\sum_a^M\eta_i^a\eta_j^a\sigma_i\sigma_j
			\label{eq:hamiltonian_single}
		\end{equation}
		Introducing the random variables $\chi_i^a$ satisfying
		\[
			\chi_i^a=
			\begin{cases}
				1 \ \text{with probability}\ p\\
				-1 \ \text{with probability}\ 1-p
			\end{cases}
		\]
		we can rewrite the Hamiltonian as
		\[
			\mathcal{H}_{N,M}(\boldsymbol \sigma | \boldsymbol \chi, \boldsymbol \xi )=-\frac{1}{N}\sum_{i, j\neq i}^N\sum_a^M\chi_i^a\chi_j^a\xi_i\xi_j\sigma_i\sigma_j.
		\]
		We can now compute the energy of the examples and of the archetype. For what concerns the latter it holds
		\begin{equation}
			E_a =\mathcal{H}_{N,M}(\boldsymbol \xi |\boldsymbol \chi, \boldsymbol \xi )=-\frac{1}{N}\sum_{i, j\neq i}^N\sum_a^M\chi_i^a\chi_j^a\xi_i\xi_j\xi_i\xi_j=-\frac{1}{N}\sum_{i, j\neq i}^N\sum_a^M\chi_i^a\chi_j^a.
		\end{equation}
		Noting that
		\[
			\chi_i^a\chi_j^a=
			\begin{cases}
				1 \ \text{with prob.}\ p^2+(1-p)^2\\
				-1 \ \text{with prob.}\ 2p(1-p)
			\end{cases}
		\]
		we get, after taking the expectation
		\begin{equation}
			\mean*{E_a}=-MN\qua*{p^2+(1-p)^2-2p(1-p)}=-4MN\ton*{p-\frac{1}{2}}^2.
			\label{eq:E_a_single}
		\end{equation}
		Analogously the energy of an example $\eta^c$ is
		\begin{equation}
		\notag
			\mean*{E_e}=\mathcal{H}_{N,M}(\boldsymbol \eta^c | \boldsymbol \chi, \boldsymbol \xi )=-\frac{1}{N}\sum_{\substack{i, j=1 \\ i\neq j}}^N\sum_{a=1}^M\chi_i^a\chi_j^a\xi_i\xi_j\chi_i^c\chi_j^c\xi_i\xi_j= -\frac{1}{N}\qua*{\sum_{i, j\neq i}^N\chi_i^c\chi_j^c\chi_i^c\chi_j^c+\sum_{i, j\neq i}^N\sum_{a\neq c}^M\chi_i^a\chi_j^a\chi_i^c\chi_j^c}.
		\end{equation}
		It holds
		\[
			\chi_i^a\chi_j^a\chi_i^c\chi_j^c=
			\begin{cases}
				1 \ \text{with probability}\ p^4+(1-p)^4+6(p-1)^2p^2\\
				-1 \ \text{with probability}\ 4p^3(1-p)+4p(1-p)^3
			\end{cases}
		\]
		and consequently
		\begin{equation}
			\mean*{E_e}=-N\qua*{1+16(M-1)\ton*{p-\frac{1}{2}}^4}.
			\label{eq:E_e_single}
		\end{equation}

		Comparing Eqs.~\eqref{eq:E_a_single} and \eqref{eq:E_e_single} we can compute the energy difference $\Delta_{N,M}$ between the archetype and a given example
		\begin{align*}
			\Delta_{N,M}&=\mean*{E_a}-\mean*{E_e} =N\qua*{1+16(M-1)\ton*{p-\frac{1}{2}}^4-4M\ton*{p-\frac{1}{2}}^2}.
		\end{align*}
		This expression has two implications
		\begin{enumerate}
			\item for $M$ sufficiently small and $p$ close to $1/2$ it holds $\Delta_{N,M} \approx N>0$ and so the energy of the examples, as expected, is lower than that of the archetype;
			\item the energy difference diverges for $N\to\infty$ suggesting a true phase transition to happen.
		\end{enumerate}
		We can define the critical number of examples $M_E$ such that $\Delta_{N,M_E}=0$, this yields
		\begin{equation}\label{GiordanoBoundE}
			M_E=\frac{16\ton*{p-\frac{1}{2}}^4-1}{16\ton*{p-\frac{1}{2}}^4-4\ton*{p-\frac{1}{2}}^2}.
		\end{equation}
		Note that it holds $M_E=M_{a,1}+1$, where $M_{a,1}$, defined in Eq.~\eqref{eq:stability_archetype}, is the value of $M$ for which the archetype becomes dynamically stable.
		Figure \ref{fig:J_single_Energie} (first row) provides a picture of this scenario.

			\begin{figure*}[tb]
	            \includegraphics[width=1.0\textwidth]{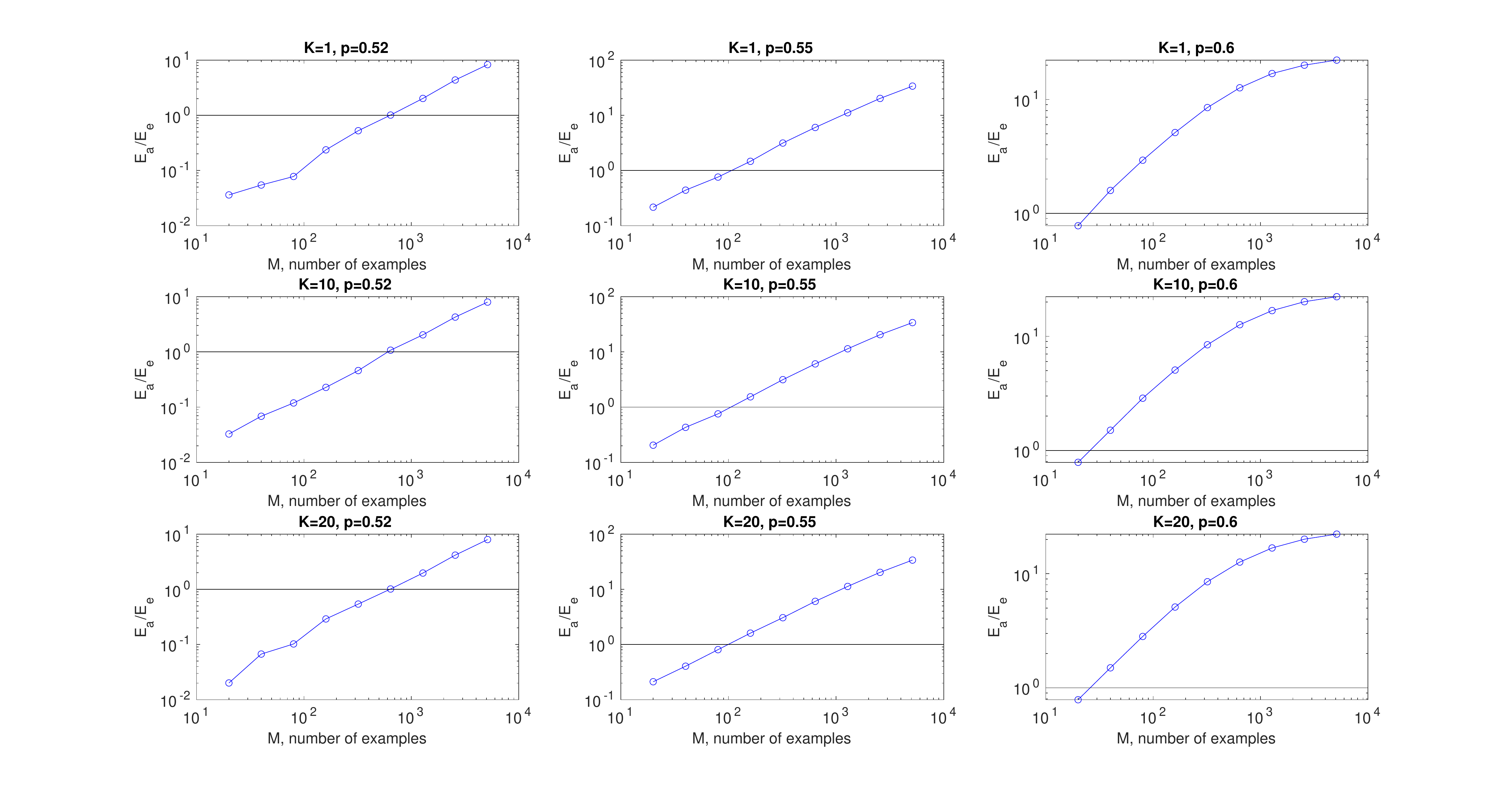}
	           % \label{fig:energy_p06_N1000}
	       % \end{subfigure}
	        \caption{Comparison between the energy of the system evaluated when the neural configuration corresponds to the archetype ($\mean*{E_a}$) and to an example ($\mean*{E_e}$), as the size $M$ of the dataset (per pattern) is varied. The panels are arranged in such a way that different columns correspond to different choices for the parameter $p$ and different rows correspond to different choices for the parameter $K$. Notice that, in any case, $|E_a|$ eventually overcomes $|E_e|$ (this is highlighted by the intersection with the horizontal line corresponding to unity) and therefore the configuration corresponding to the archetype is energetically more convenient. This data are obtained by averaging over the energies of all the $K$ archetypes and the $M\cdot K$ noisy examples.}
        	\label{fig:J_single_Energie}
		\end{figure*}

\subsection{Many archetypes (high storage)} \label{ssec:high}
	We now turn to the case where we aim to storing $K$ distinct archetypes $\{\boldsymbol{\xi}^{\mu}\}_{\mu=1,...,K}$ starting from $M$ noisy realizations $\{\boldsymbol{\eta}^{\mu, a}\}_{\mu=1,...,K}^{a=1,...,M}$ of each of them. Exploiting the Hebbian rule we can write the couplings as
	\[
		J_{ij}=\frac{1}{N}\sum_{a=1}^M\sum_{\mu=1}^K\eta_i^{\mu, a}\eta_j^{\mu, a}.
	\]
	In the following, mirroring the previous subsection, we analyse the stability of the archetypic patterns and of the noisy patterns, respectively.

	\subsubsection{Stability of the archetypes}
		The procedure for determining if the archetypes are stable is the same as the one performed in Subsec.~\ref{subsec:stability_archetype}. First of all we have to evaluate the product between the local field and a given archetype, say $\boldsymbol{\xi}^1$ without loss of generality. Focusing on the $i$-th component we get
		\[
			h_i\xi_i^1=\frac{1}{N}\sum_{j\neq i}^N\sum_{\mu=1}^K\sum_{a=1}^M\eta_i^{\mu, a}\eta_j^{\mu, a}\xi_i^1\xi_j^1
		\]
		and taking the expectation over the noise we obtain
		\[
			\mean*{h_i\xi_i^1}=\frac{1}{N}\sum_{j\neq i}^N\sum_{\mu=1}^K\sum_{a=1}^M\mean*{\eta_i^{\mu, a}\xi_i^1}\mean*{\eta_j^{\mu, a}\xi_j^1}.
		\]
		Generalizing Eq.~\eqref{eq:correlation_noise_archetype} to the case of multiple archetypes we can write the expectation appearing in the sum as
		\begin{equation}
			\mean*{\eta_i^{\mu, a}\xi_i^1}=2\ton*{p-\frac{1}{2}}\delta_{\mu, 1}.
			\label{eq:correlation_many_noise_archetype}
		\end{equation}
		It then follows
		\begin{equation}
			\mean*{h_i\xi_i^1}=\frac{1}{N}\sum_{j\neq i}\sum_{a=1}^M\mean*{\eta_i^{1, a}\xi_i^1}\mean*{\eta_j^{1, a}\xi_j^1}=4M\ton*{p-\frac{1}{2}}^2.
			\label{eq:mean_field_many_archetypes}
		\end{equation}
		This result coincides with that derived in Subsec.~\ref{subsec:stability_archetype}.

		As previously done we have to consider also the variance of $h_i\xi_i^1$. We have
		\[
			\ton*{h_i\xi_i^1}^2=h_i^2=\frac{1}{N^2}\sum_{j,k\neq i}^N\sum_{\mu, \rho=1}^K\sum_{a,b=1}^M\eta_j^{\mu, a}\eta_i^{\mu, a}\xi_j^1\eta_k^{\rho, b}\eta_i^{\rho, b}\xi_k^1
		\]
		In order to compute the expectation over the noise we have to get rid of the self interaction terms, splitting the sums we obtain $8$ terms
		\begin{alignLetter}
			(h_i\xi_i^1)^2=&\nonumber\\
			\frac{1}{N^2}\Bigg[&NKM+\\
			+&\sum_{j\neq i}^N\sum_{k\neq j, i}^N\sum_{\mu=1}^K\sum_{a=1}^M\xi_j^1\xi_k^1\eta_j^{\mu, a}\eta_k^{\mu, a}+\\
			+&\sum_{j\neq i}^N\sum_{\mu=1}^K\sum_{\rho\neq\mu}^K\sum_{a=1}^M\eta_j^{\mu, a}\eta_i^{\mu, a}\eta_j^{\rho, a}\eta_i^{\rho, a}+\\
			+&\sum_{j\neq i}^N\sum_{k\neq j, i}^N\sum_{\mu=1}^K\sum_{\rho\neq\mu}^K\sum_{a=1}^M\xi_j^1\xi_k^1\eta_j^{\mu, a}\eta_i^{\mu, a}\eta_k^{\rho, a}\eta_i^{\rho, a}+\\
			+&\sum_{j\neq i}^N\sum_{\mu=1}^K\sum_{a=1}^M\sum_{b\neq a}^M\eta_j^{\mu a}\eta_i^{\mu, a}\eta_j^{\mu, b}\eta_i^{\mu b}+\\
			+&\sum_{j\neq i}^N\sum_{k\neq j, i}^N\sum_{\mu=1}^K\sum_{a=1}^M\sum_{b\neq a}^M\xi_j^1\xi_k^1\eta_j^{\mu, a}\eta_i^{\mu, a}\eta_k^{\mu, b}\eta_i^{\mu, b}+\\
			+&\sum_{j\neq i}^N\sum_{\mu=1}^K\sum_{\rho\neq\mu}^K\sum_{a=1}^M\sum_{b\neq a}^M\eta_j^{\mu, a}\eta_i^{\mu, a}\eta_j^{\rho, b}\eta_i^{\rho, b}+\\
			+&\sum_{j\neq i}^N\sum_{k\neq j, i}^N\sum_{\mu=1}^K\sum_{\rho\neq\mu}^K\sum_{a=1}^M\sum_{b\neq a}^M\xi_j^1\xi_k^1\eta_j^{\mu, a}\eta_i^{\mu, a}\eta_k^{\rho, b}\eta_i^{\rho, b}\Bigg].
		\end{alignLetter}
		After the expectation is taken, only four terms survive. Indeed, generalizing Eq.~\eqref{eq:correlation_noise_noise} to the present case we can write
		\begin{equation}
			\mean*{\eta^{\mu, a}_i\eta^{\rho, b}_i}=4\ton*{p-\frac{1}{2}}^2\delta_{\mu, \rho}
			\label{eq:correlation_many_noise_noise}
		\end{equation}
		and thus the average of the terms $(C)$, $(D)$, $(G)$ and $(H)$ are null, because of the constraint $\mu\neq\rho$. In conclusion, exploiting Eqs.~\eqref{eq:correlation_many_noise_archetype} and \eqref{eq:correlation_many_noise_noise}, we arrive at
		\begin{equation}\nonumber
			\mean*{(h_i\xi_i^1)^2}=\frac{KM}{N}+4M\ton*{p-\frac{1}{2}}^2+16\frac{KM^2}{N}\ton*{p-\frac{1}{2}}^4+16M^2\ton*{p-\frac{1}{2}}^4
		\end{equation}
		and so, recalling Eq.~\eqref{eq:mean_field_many_archetypes}, we obtain for the variance
		\begin{equation}
		\label{eq:variance_field_many_archetype}
			\var\qua*{h_i\xi_i^1}=\frac{KM}{N}+4M\ton*{p-\frac{1}{2}}^2+16\frac{KM^2}{N}\ton*{p-\frac{1}{2}}^4.
		\end{equation}

		The final step consists in comparing the square root of the variance, Eq.~\eqref{eq:variance_field_many_archetype}, and the mean value Eq.~\eqref{eq:mean_field_many_archetypes}; the archetypes are stable provided that the former is smaller than the latter.
		\[
			\mean*{h_i\xi_i^1}>\sqrt{\var\qua*{h_i\xi_i^1}},
		\]
		and this yields to
		\begin{equation}\label{eq:stability_archetype_K}
			M >\frac{\frac{K}{N}+4\ton*{p-\frac{1}{2}}^2}{16\ton*{p-\frac{1}{2}}^4\ton*{1-\frac{K}{N}}} = M_{a,K}.
		\end{equation}
		As long as $K$ is not too large, there exists a finite threshold $M_{a,K}$ which ensures the stability of the archetype for a relatively large number of examples, on the other hand, as one would expect, if $K\to N$, the threshold diverges and the archetypes never get dynamically stable, independently of $M$.
	Moreover the expression in Eq.~\eqref{eq:stability_archetype_K} generalizes Eq.~\eqref{eq:stability_archetype} derived for a single archetype: now, as the term $K/N$ is not vanishing in the thermodynamic limit, when $p$ is close to $1/2$, the scaling $M \propto 1/(2p-1)^2$ is not enough and we need $M \propto 1/(2p-1)^4$ in order to ensure a reliable storing of the archetypes. Remarkably, these scalings for the low ($\alpha=0$) and the high ($\alpha>0$) load are accordingly recovered in the statistical mechanical analysis (see Proposition \ref{MSscaling} in Sec.~\ref{sec:SM}).

		These results are corroborated in Fig.~\ref{fig:esempinumericiK10} (left panel), where one can see that the instability region (bright colors) around $p=1/2$ is wider than the one in the analogous panel in Fig.\ref{fig:esempinumericiK1}.

		\begin{figure}[t!]
	\centering
    \includegraphics[width=0.45\textwidth]{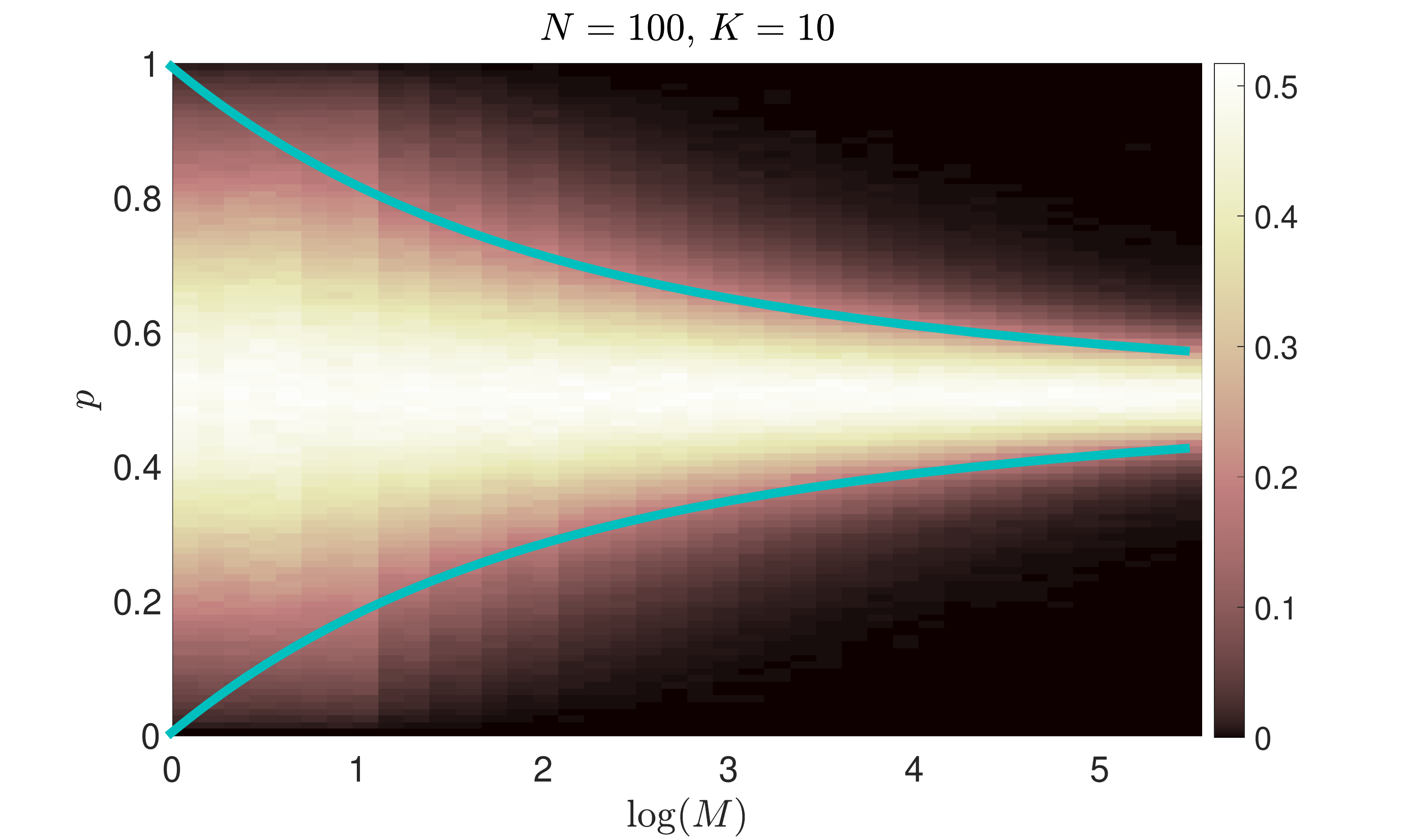}
    \includegraphics[width=0.45\textwidth]{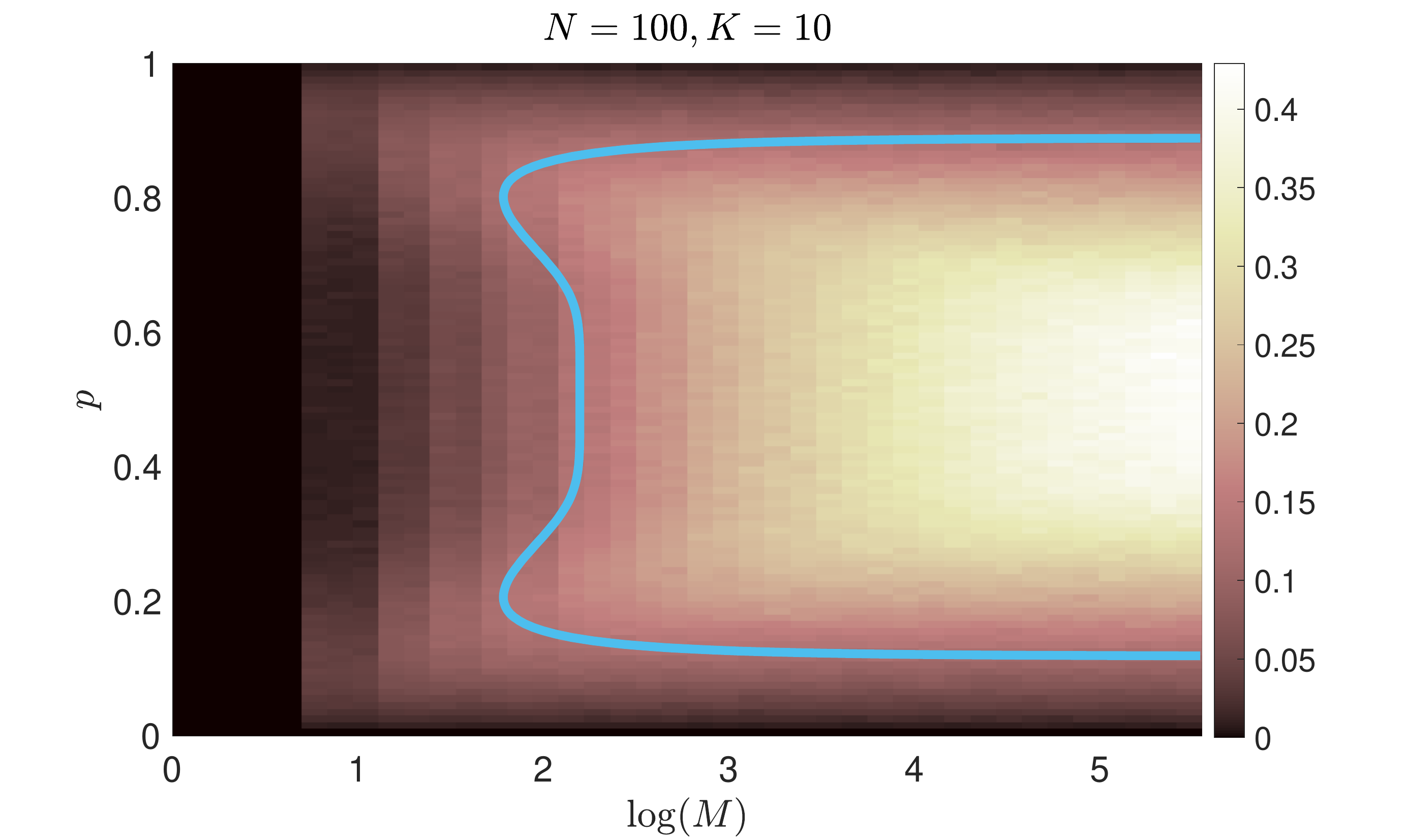}
	\caption{
	Left: Archetype stability versus $p$ and $M$.
The colormap shows the fraction of neurons that occur to be unstable when set in the configuration $\boldsymbol \sigma = \boldsymbol \xi^1$, as a function of $p$ and of $M$ (logarithmic axis); here $\boldsymbol \xi^1$ is one of the $K=10$ archetypes considered and, for each, $M$ examples are drawn randomly, while the number $N$ of neurons is set equal to 100. The data depicted have been obtained by averaging over a sample of $100$ different realizations of $\{\boldsymbol \xi^{\mu} \}_{\mu=1,...,K}$ and of related examples. The solid lines represent the roots of equation \eqref{eq:stability_archetype_K} and correctly demarcate the region of instability corresponding to values of $p$ close to $1/2$.
Right: Example stability versus $p$ and $M$.
The colormap shows the fraction of neurons that occur to be unstable when set in the configuration $\boldsymbol \sigma = \boldsymbol \eta^{1,1}$, as a function of $p$ and of $M$ (logarithmic axis); here $\boldsymbol \eta^{1,1}$ is one of the $M$ examples related to one (i.e., $\boldsymbol \xi^{1}$, without loss of generality) of the $K$ archetypes considered and data depicted have been obtained by averaging over a sample of $100$ different realizations of $\{\boldsymbol  \xi^{\mu} \}_{\mu=1,...,K}$ and of the related examples, while the number of neurons is set equal to 100. The solid line represents the root of equation (\ref{eq:cons_K}) and correctly demarcates the region of instability corresponding to values of $p$ close to $1/2$.}\label{fig:esempinumericiK10}
\end{figure}

	\subsubsection{Stability of the noisy examples}
		The product between the local field $h_i$ and the corresponding component of one of the noisy examples, say for instance $\eta_i^{1,1}$, is
		\begin{equation}
			h_i\eta_i^{1,1}=\frac{1}{N}\sum_{j\neq i}^N\sum_{\mu=1}^K\sum_{a=1}^M\eta_i^{\mu, a}\eta_j^{\mu, a}\eta_i^{1,1}\eta_j^{1,1}=\frac{1}{N}\sum_{j\neq i}^N\sum_{\mu=1}^K\sum_{a=1}^M\eta_i^{\mu, a}\eta_j^{\mu, a}\eta_i^{1,1}\eta_j^{1,1}\qua*{(1-\delta_{a,1})+\delta_{a,1}}.
		\end{equation}
		Taking the expectation over the noise and exploiting Eq.~\eqref{eq:correlation_many_noise_noise} we obtain
		\begin{equation}
			\mean*{h_i\eta_i^{1,1}}=1+16M\ton*{p-\frac{1}{2}}^4,
			\label{eq:mean_field_many_noise}
		\end{equation}
		that is equivalent to the expression in Eq.~\eqref{eq:mean_field_noise} derived for $K=1$. Moving to the variance we have to compute the expectation of $\ton*{h_i\eta_i^{1,1}}^2$: it holds
		\begin{equation}
		\notag
			\ton*{h_i\eta_i^{1,1}}^2 =h_i^2=\frac{1}{N^2}\qua*{\sum_{j\neq i}^N\sum_{\mu=1}^K\sum_{a=1}^M\eta_i^{\mu, a}\eta_j^{\mu, a}\eta_j^{1,1}}^2
			 =\frac{1}{N^2}\qua*{\sum_{j\neq i}^N\sum_{\mu=1}^K\sum_{a=1}^M\eta_i^{\mu, a}\eta_j^{\mu, a}\eta_j^{1,1}\qua*{(1-\delta_{a,1})+\delta_{a,1}}}^2.
		\end{equation}
		This yields
		\begin{alignLetter}
			\ton*{h_i\eta_i^{1,1}}^2&=\nonumber\\
			\frac{1}{N^2}\Bigg[&\sum_{j,k\neq i}^N\sum_{\mu,\nu}^K\eta_j^{\mu, 1}\eta_i^{\mu, 1}\eta_j^{1,1}\eta_k^{\nu, 1}\eta_i^{\nu, 1}\eta_k^{1,1}+\\
			&2\sum_{a\neq 1}^M\sum_{j,k\neq i}^N\sum_{\mu,\nu}^K\eta_j^{\mu, a}\eta_i^{\mu, a}\eta_j^{1,1}\eta_k^{\nu, 1}\eta_i^{\nu, 1}\eta_k^{1,1}+\\
			&\sum_{a,b\neq 1}^M\sum_{j,k\neq i}^N\sum_{\mu,\nu}^K\eta_j^{\mu, a}\eta_i^{\mu, a}\eta_j^{1,1}\eta_k^{\nu, b}\eta_i^{\nu, b}\eta_k^{1,1}\Bigg]=\\
			\frac{1}{N^2}&\qua*{A+B+C} \nonumber
		\end{alignLetter}
		Taking the expectation and proceeding as done before we get
			\begin{align*}
				\mean*{A}=&N(N+K-1)\\
				\mean*{B}=&32\ton*{p-\frac{1}{2}}^4MN(N+K-1)\\
				\mean*{C}=&MNK+16\ton*{p-\frac{1}{2}}^4MN^2+16\ton*{p-\frac{1}{2}}^4M^2KN+64\ton*{p-\frac{1}{2}}^6M^2N^2.
			\end{align*}
			Combining these expressions with Eq.~\eqref{eq:mean_field_many_noise} we can compute the variance $\var{\qua*{h_i\eta_i^{1,1}}}$ as
			\begin{align}
				\var{\qua*{h_i\eta_i^{1,1}}}=&\frac{K-1}{N}+32\ton*{p-\frac{1}{2}}^4M\frac{K-1}{N}+\nonumber\\
				&+\frac{MK}{N}+16\ton*{p-\frac{1}{2}}^4M+16\ton*{p-\frac{1}{2}}^4\frac{M^2K}{N}+\nonumber\\
				&+64\ton*{p-\frac{1}{2}}^6M^2-256\ton*{p-\frac{1}{2}}^8M^2
			\end{align}
			and then write down the condition for the dynamical stability
			\begin{equation}\label{eq:cons_K}
			\mean*{h_i\eta_i^{1,1}}>\sqrt{\var{\qua*{h_i\eta_i^{1,1}}}}.
			\end{equation}
			This last inequality is given by
			\[
				aM^2+bM+c>0,
			\]
			where
			\begin{equation}
				\begin{cases}
					a=512\ton*{p-\frac{1}{2}}^8-64\ton*{p-\frac{1}{2}}^6-16\ton*{p-\frac{1}{2}}^4\frac{K}{N}\\
					b=16\ton*{p-\frac{1}{2}}^4-32\ton*{p-\frac{1}{2}}^4\frac{K-1}{N}-\frac{K}{N}\\
					c=1-\frac{K-1}{N}
				\end{cases}
				\label{eq:coefficients_many}
			\end{equation}
			Note that these expressions recover those derived for a single archetype (Eq.~\eqref{eq:coefficients}) when $K=1$. These results are corroborated in Fig.~\ref{fig:esempinumericiK10} (right panel). In particular, we notice that, for relatively small values of $M$, stability is always ensured and therefore the noisy patterns are effectively stored in the network. Conversely, as $M$ gets large noisy examples progressively loose stability while the archetype patterns progressively get more stable.

				\begin{figure*}[tb]
	            \includegraphics[width=1.0\textwidth]{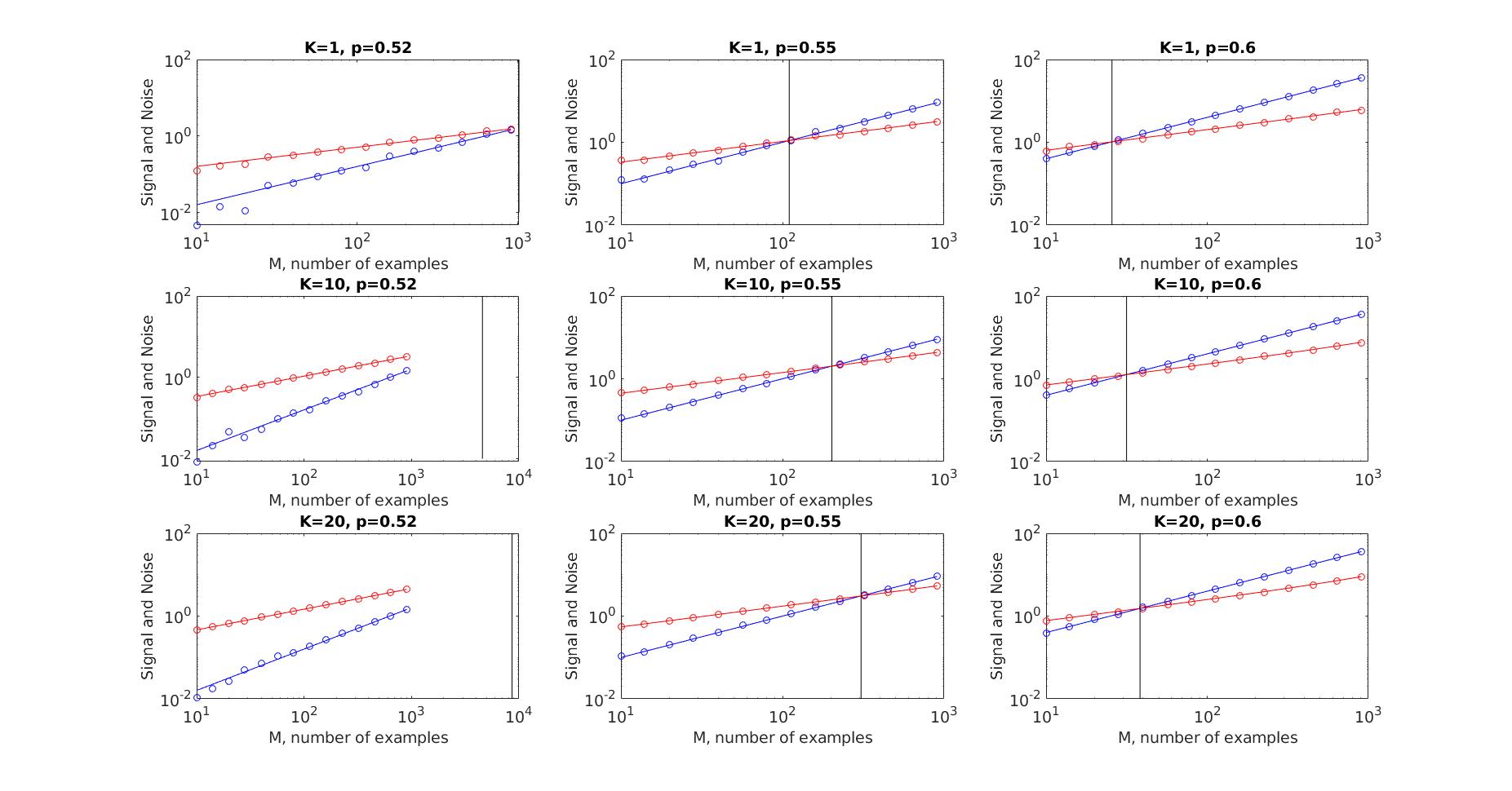}
	        	        \caption{Comparison between the signal $\langle h_i \xi_i^{\mu}\rangle$ acting on the $i$-th spin evaluated when the neural configuration corresponds to the archetype (dark color) and its standard deviation $\sqrt{\textrm{Var} (h_i \xi_i^{\mu})}$ (bright color), as the size $M$ of the dataset (per pattern) is varied. The panels are arranged in such a way that different columns correspond to different choices for the parameter $p$ and different rows correspond to different choices for the parameter $K$. Theoretical results (solid lines) representing Eqs.~\eqref{eq:variance_field_many_archetype} and \eqref{eq:mean_field_many_archetypes} are nicely overlapped by numerical results (bullets) obtained via simulations. Each curve has been obtained averaging over the $K\cdot N$ products $h_i \xi_i^{\mu}$. Notice that, in any case, $\langle h_i \xi_i^{\mu}\rangle$ and $\sqrt{\textrm{Var} (h_i \xi_i^{\mu})}$ intersect and, eventually, $\langle h_i \xi_i^{\mu}\rangle$ turns out to be smaller and therefore the configuration corresponding to the archetype is energetically more convenient.}
        	\label{fig:J_single}
		\end{figure*}

	A summary of the overall signal-to-noise analysis, along with numerical checks, is given in Fig.~\ref{fig:J_single} and in Fig.~\ref{fig:checkS2N}.
\begin{figure}[tb]
\centering
\includegraphics[width=0.35\textwidth]{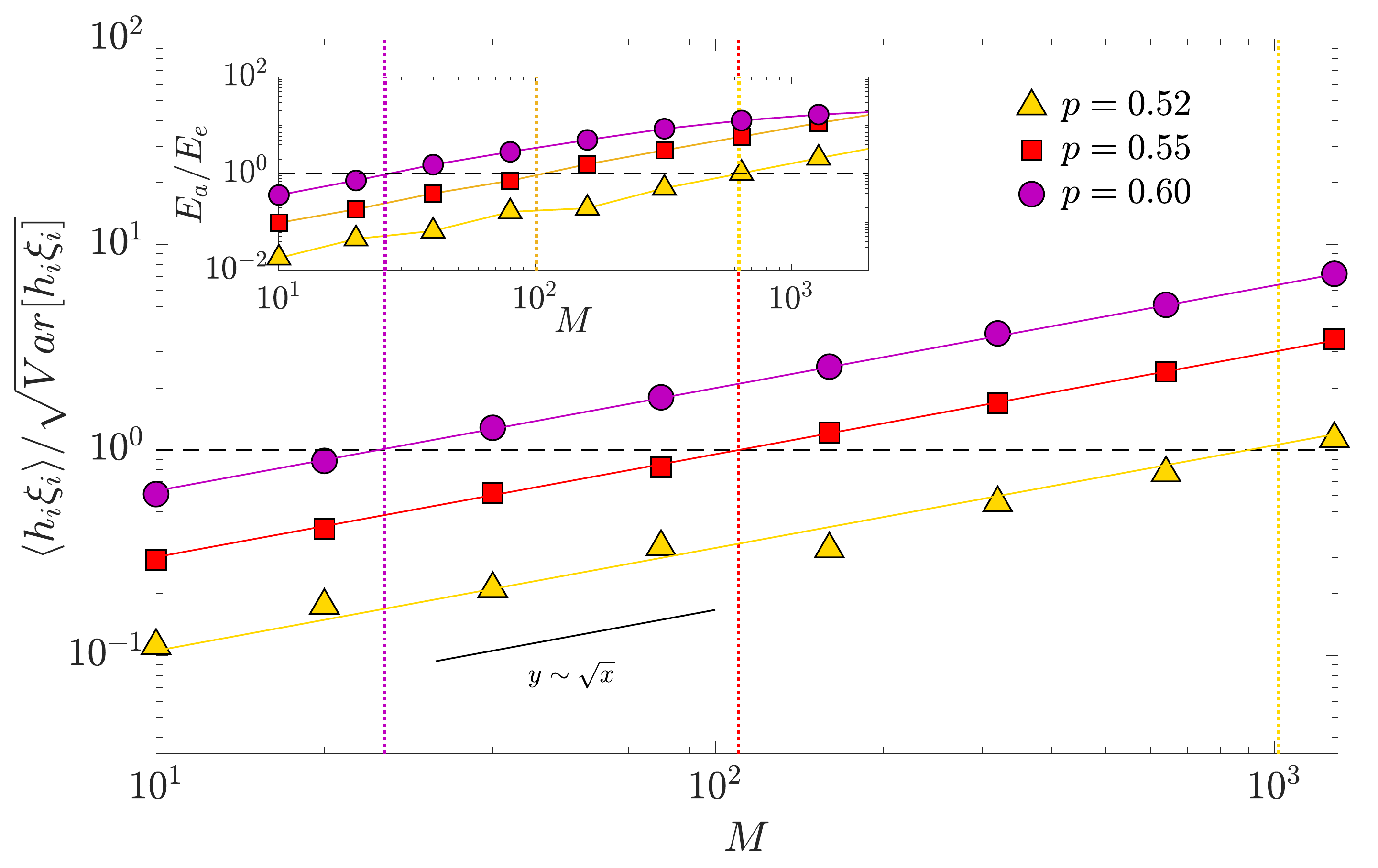}
\caption{\label{fig:checkS2N} Main figure: ratio between the signal, evaluated for a configuration retrieving the archetype, and its variance, versus $M$ and for different values of $p$, as explained by the legend; fitting curves highlight a square-root scaling. As $M$ grows, the signal prevails over the noise and the crossover (highlighted by the vertical dotted lines) occurs at a value $M_a$ which decreases with $p \in [0.5, 1]$. Inset: ratio between the energies $E_a$ and $E_e$ evaluated for configurations corresponding, respectively, to an archetype and to an examples; lines are guides to the eye. As $M$ gets larger, the archetype gets energetically more favorable than any example. For both figures, data depicted are obtained for a system of size $N=1000$ and averaged over all the $M$ examples.}
\end{figure}

	Finally, we studied the averages energies $E_a$ and $E_e$ corresponding to, respectively, an archetype configuration and an example configuration also for the high-load regime finding, again, that for large enough values of $M$ the former is energetically more favorable. Results are collected in Fig.~\ref{fig:J_single_Energie}.

\section{Statistical mechanics approach} \label{sec:SM}
To better inspect the crossover between archetype and example stabilities evidenced by signal-to-noise analysis and, possibly, to frame such a phenomenon into a classical phase transition setting where fast noise is also accounted, we must rely on statistical mechanics of spin glasses. In particular, we will use a reformulation \cite{AABF-NN2020,Agliari-Barattolo} of the celebrated Guerra's interpolation technique \cite{Guerra}.
\subsection{General setting and main definitions}
Let us consider a network made of $N$ Ising neurons $\sigma_i = \pm 1$, with $i \in (1,...,N)$, $K= \alpha N$ archetype patterns $\boldsymbol{\xi}_i^{\mu} \in  \{ -1, +1\}$ with $\mu \in (1,...,K)$, and $M$ noisy examples per archetype $\boldsymbol \eta^{\mu,a}$ with $\mu \in (1,...K)$ and $a \in (1,...,M)$. The latter constitute a stochastic, perturbed version of the archetypes, that are still binary and the arbitrary $i$-th component can be written as $\eta_i^{\mu,a} = \xi_i^{\mu}\chi_i^{\mu,a}$ for $i=1,...,N$, where $\chi_i^{\mu,a}$ is a Bernoullian random variable taking value $-1$ or $+1$.
We will assume that, for each component, $\mathcal{P}(\xi = \pm 1) = 1/2$ and $\mathcal{P}(\chi = 1) = 1 - \mathcal{P}(\chi = -1) = p$ namely, the closer $p$ is to $1/2$ and the higher the noise in the example (viceversa for $p \to 0$ and $p \to 1$, as the network stores equally a pattern and its flipped version, due to the spin-flip symmetry $\sigma_i \to -\sigma_i$).
The network is fed by the $M \times K$ noisy patterns and has no direct access to the $K$ archetype patterns.
\begin{definition} \label{def:H}
The Hamiltonian of the model is defined as
\begin{equation}
  \label{eq:Hamiltonian2}
  \mathcal{H}_{N,M}(\boldsymbol \sigma | \boldsymbol \chi, \boldsymbol \xi) = -\frac{1}{2N} \sum_{a=1}^M \sum_{\mu=1}^K \big(\sum_{i=1}^N \xi_i^\mu \chi_i^{\mu,a}\sigma_i \big) ^2.
\end{equation}
The partition function coupled to the Hamiltonian \eqref{eq:Hamiltonian2} is defined as
  \begin{equation}
    \label{def:Z}
    Z_{N,M}(\alpha, \beta | \boldsymbol \chi, \boldsymbol \xi)=\sum_{\sigma} \exp(-\beta  \mathcal{H}_{N,M}(\boldsymbol \sigma| \boldsymbol \chi,\boldsymbol\xi)) = \sum_{\sigma} \exp \Big[\frac{\beta}{2N} \sum_{a=1}^M \sum_{\mu=1}^K \big(\sum_{i=1}^N \xi_i^\mu \chi_i^{\mu,a}\sigma_i \big) ^2 \Big].
\end{equation}
At finite network volume $N$ and sample size $M$, the quenched pressure (i.e., the free energy times $- \beta$ \cite{Guerra}) of this model reads as
\begin{equation}\label{FreeE}
A_{N,M}(\alpha,\beta)=\frac{1}{N}\mathbb{E}\ln Z_{N,M}(\alpha, \beta| \boldsymbol \chi, \boldsymbol \xi),
\end{equation}
where $\mathbb{E}:=\mathbb{E}_\chi\mathbb{E}_{\xi}$, being
\begin{eqnarray}
\label{eq:quenchedaverages}
\mathbb{E}_{\xi}\, G(\xi) &=& \int_{\bbR} \Big( \prod_{i=1}^{N}\prod_{\mu=1}^{K}\frac{d \xi_{i}^{\mu}}{2} \big[\delta(\xi_{i}^{\mu} + 1) + \delta(\xi_{i}^{\mu} - 1)\big]\Big) G(\xi) \\
\mathbb{E}_{\chi} G(\chi) &=&\left( \prod_{\mu=1}^{K} \prod_{a=1}^{M}\prod_{i=1}^{N} \mathbb{E}_{\chi_{i}^{\mu,a}}\right) G(\chi)\\
\mathbb{E}_{\chi_{i}^{\mu,a}} f(\chi) &=& \begin{cases}\int_{-\infty}^{+\infty}d\chi_{i}^{\mu,a}\left( p \delta(\chi_{i}^{\mu,a}-1) + (1-p) \delta(\chi_{i}^{\mu,a}+1)\right)f(\chi)& \mu = 1\\
                                    \int_{-\infty}^{+\infty}\frac{d\chi_{i}^{\mu,a}}{\sqrt{2\pi}} \exp(-\frac{(\chi_{i}^{\mu,a})^{2}}{2}) \,f(\chi)& \mu =2,\cdots,K
                        \end{cases}
\end{eqnarray}
Finally, for a generic observable $O(\boldsymbol \sigma | \boldsymbol {\xi},\boldsymbol {\chi})$, we define the brackets as $\langle O \rangle := \mathbb{E}\Omega \left(O(\boldsymbol \sigma|\boldsymbol{\xi},\boldsymbol{\chi})\right)$, being $\Omega$ the (replicated) Boltzmann average.
\end{definition}
\bigskip
\begin{remark}
In equation \eqref{eq:quenchedaverages} we approximated the noise terms $\chi_{i}^{\mu,a}$ for $\mu=2,\cdots,K$ as standard Gaussian variables; in the thermodynamic limit this assumption fits the worst case ($p=1/2$) and, in general, it plays as a bound: if the network is able to infer an archetype out of this noisiest example sample, it will certainly works also in less challenging ($p>1/2$) cases. 
\end{remark}
\bigskip
\begin{definition}
In order to quantify both the retrieval of the archetype and the retrieval of the examples, we define the related Mattis magnetizations as, respectively,
\begin{eqnarray}
m_{\mu} &=& \frac{1}{N}\sum_{i=1}^{N} {\xi}_i^{\mu} \sigma_i,\\
n_{\mu,a} &=&\frac{1}{N}\sum_{i=1}^{N} {\xi}_i^{\mu} \chi_i^{\mu,a} \sigma_i.
\end{eqnarray}
\end{definition}
\begin{proposition}
The partition function defined in ~\eqref{def:Z} can be recast as
  \begin{eqnarray}
    \label{eq:ZJ}
    \notag
    Z_{N,M}(\beta,\alpha | \boldsymbol \chi, \boldsymbol \xi) &=& \lim_{J \to 0} Z_{N,M}(\beta,\alpha, J | \boldsymbol \chi, \boldsymbol \xi)=  \lim_{J \to 0} \sum_{\sigma} \int \prod_{\mu=2,a=1}^{K,M}\frac{dz_{\mu,a}}{\sqrt{2\pi}}\exp \Big[-\frac{1}{2}\sum_{\mu=2,a=1}^{K,M}z_{\mu,a}^2 \\
        \notag
    &+& J \sum_{i=1}^N {\xi}_i^1\sigma_i
    + \sqrt{\frac{\beta}{N}} \sum_{a=1}^M \sum_{\mu=2}^K \sum_{i=1}^N  \xi_i^\mu \chi_i^{\mu,a}\sigma_i z_{\mu,a} + \frac{\beta}{2N} \sum_{a=1}^M \big(\sum_{i=1}^N \xi_i^1 \chi_i^{1,a}\sigma_i \big) ^2 \Big],\nonumber
\end{eqnarray}
which corresponds to the partition function of a restricted Boltzmann machine with $N$ visible binary neurons $\sigma_i \in \{ -1, +1\}$, $M \times K$ hidden Gaussian neurons $z_{\mu,a} \sim \mathcal{N}(0,1)$, and weights $\chi_i^{\mu,a} {\xi}_i^{\mu}$, for any $i =1,...,N$, $\mu=1,...,K$, and $a=1,...,M$.
\end{proposition}
\bigskip
\begin{remark}
In the expression above we added the last term $J \sum_{i=1}^N\xi_i^1\sigma_i$ to generate the expectations of the Mattis magnetization $m_1$, by evaluating the derivative of the quenched pressure w.r.t. J at $J=0$. In fact, we need to quantify both the retrieval of the archetype and the retrieval of the examples, but, while the noisy examples {\em exist} and are supplied to the network (in fact, the Hamiltonian itself can be written in terms of the examples $\{\boldsymbol \eta^{\mu,a}\}$), the archetype is a network's abstraction, nor it exists by itself neither it is coded in the Hamiltonian, hence we need to use the functional generator trick. However, as we will see in Sec.~\ref{ssec:MM}, as far as $M \gg 1$, we can bypass this artifice and obtain the expectation value of $m$ by exploiting is direct proportionality with the expectation value of $n$, which, instead, is a natural order parameter for the model.
\end{remark}
\begin{proof}
We chose as ``marked'' (or ``condensate'') patterns \cite{Amit, CKS} those related to the archetype labelled as $\mu=1$ and, accordingly, we re-write eq. \eqref{def:Z} as
\begin{equation}
    Z_{N,M}(\beta,\alpha | \boldsymbol \chi, \boldsymbol \xi) =\sum_{\sigma} \exp \Big[\frac{\beta}{2N} \sum_{a=1}^M \big(\sum_{i=1}^N \xi_i^1 \chi_i^{1,a}\sigma_i \big) ^2 +  \frac{\beta}{2N} \sum_{a=1}^M \sum_{\mu=2}^K \big(\sum_{i=1}^N \xi_i^\mu \chi_i^{\mu,a}\sigma_i \big) ^2 \Big].
\end{equation}
Since we are interested in extracting the magnetization for both the noisy examples and the archetypes, we introduce a source field $J$ such that the partition function is generalized as
\begin{equation}
\notag
Z_{N,M}(\beta,\alpha, J | \boldsymbol \chi, \boldsymbol \xi)=\sum_{\sigma} \exp \Big[\frac{\beta}{2N} \sum_{a=1}^M \big(\sum_{i=1}^N \xi_i^1 \chi_i^{1,a}\sigma_i \big) ^2 +  \frac{\beta}{2N} \sum_{a=1}^M \sum_{\mu=2}^K \big(\sum_{i=1}^N \xi_i^\mu \chi_i^{\mu,a}\sigma_i \big) ^2 + J \sum_{i=1}^N\xi_i^1\sigma_i  \Big].
\end{equation}
Then, we apply the relation
\begin{equation}
  \label{eq:gaussianlinearization}
  \exp\left(\frac{X^{2}}{2}\right)  = \int_{-\infty}^{+\infty}\frac{dz}{\sqrt{2\pi}}\exp(-\frac{z^{2}}{2} + X z)
\end{equation}
to each squared term appearing in the argument of the exponential and this directly yields to Eq. \eqref{eq:ZJ}.
\end{proof}

\subsection{Guerra's interpolation for the quenched pressure}
\label{sec:interpolation}
The strategy that we follow to solve the model is based on Guerra's interpolation technique \cite{AABF-NN2020,Agliari-Barattolo} and ultimately consists in exploiting the mean-field nature of the model to properly compare the original model with an effective one-body model that shares the same statistical features of the original one in the thermodynamics limit.
\bigskip
\begin{definition} The Guerra interpolating functional for the quenched pressure related to the cost-function \ref{eq:Hamiltonian2} is defined as
  \label{def:interpolant}
\begin{eqnarray}
    %A_{J}(t)
    \label{eq:aa}
    A_{N,M}(\alpha, \beta, J ; t) &=& \frac{1}{N} \mathbb{E}_\phi\mathbb{E}_\chi\mathbb{E}_{\xi} \ln\Big[\sum_{\sigma} \int \prod_{\mu=2,a=1}^{K,M}\frac{dz_{\mu,a}}{\sqrt{2\pi}}\exp \Big(-\frac{\psi(t)}{2}\sum_{\mu=2,a=1}^{K,M}z_{\mu,a}^2+ J \sum_{i=1}^N\xi_i^1\sigma_i\\
    &+& \Gamma(t)\sqrt{\frac{\beta}{N}} \sum_{a=1}^M \sum_{\mu=2}^K \sum_{i=1}^N \xi_i^\mu \chi_i^{\mu,a}\sigma_i z_{\mu,a} +\rho(t)\frac{\beta}{2N} \sum_{a=1}^M \big(\sum_{i=1}^N \xi_i^1 \chi_i^{1,a}\sigma_i \big) ^2 +N W_{N,M}(t)  \Big) \Big],\nonumber
\end{eqnarray}
where $\psi(t), \Gamma(t), \rho(t)$ are auxiliary fields to be set a posteriori, and $W_{N,M}(t):= W(\boldsymbol \sigma, \boldsymbol z, \boldsymbol \phi, \boldsymbol \xi, \boldsymbol \chi; t)$ is a source term whose specific expression will be set a posteriori too.
\end{definition}

\bigskip

In the following, to lighten the notation, we will set $A_J := A_{N,M}(\alpha, \beta, J ; t)$.
\newline
Note that the original model can be recovered by setting
\begin{eqnarray}
\label{eq:necessarycond}
  \psi(t=1)=\Gamma(t=1)=\rho(t=1)=1,\\
    W_{N,M}(t) = 0,
\end{eqnarray}
and, as standard, we approach  $\psi(t=1)=\Gamma(t=1)=\rho(t=1)=1$ by evaluating  the factorized case  $\psi(t=0), \Gamma(t=0), \rho(t=0)=0$ and then integrating back in $t$ from $0$ to $1$ by using the fundamental theorem of calculus.

To accomplish this plan, denoting by $\langle . \rangle_t$ the averages evaluated in this extended framework (and clearly $\langle \rangle_t \to \langle \rangle$ as $t \to 1$), let us start working out the streaming of $A_J$:
\begin{eqnarray}\nonumber
    \de{\Gamma} A_{J} &=& \frac{1}{N}\sqrt{\frac{\beta}{N}} \sum_{a=1}^M\sum_{\mu=2}^K\sum_{i=1}^N \qE{\phi}\qE{\chi}\qE{ \xi} \mv{ \xi_i^{\mu} \chi_i^{\mu,a} z_{\mu,a} \sigma_i}_t = \frac{\beta}{N^{2}}\Gamma_{t} \sum_{a=1}^{M}\sum_{\mu=2}^{K}\sum_{i=1}^{N}\qE{\phi}\qE{\chi}\qE{\xi} \big(\mv{z_{\mu,a}^{2}}_t -\mv{z_{\mu,a}\sigma_{i}}_t \big)\\ \nonumber
%\end{eqnarray}
%\begin{equation}
    \de{\rho} A_{J} &=& \frac{\beta}{2} \sum_{a=1}^M \qE{\phi}\qE{\chi}\qE{ \xi} \mv{\big(\frac{1}{N}\sum_{i=1}^N \xi_i^1 \chi_i^{1,a} \sigma_i\big)^2}_t\\ \nonumber
%\end{equation}
%\begin{equation}
    \de{\psi} A_{J}  &=& -\frac{1}{2N} \sum_{a=1}^M\sum_{\mu=2}^K \qE{\phi}\qE{\chi}\qE{ \xi} \mv{z_{\mu,a}^2}_t
\end{eqnarray}
such that
\begin{equation}
\frac{dA_{J}}{dt}=\dde{\Gamma}A_{J}+\dde{\rho}A_{J}+\dde{\psi}A_{J}+\mv{\dot W}_{t}.
\end{equation}
We still have the freedom of choice for the source term $W_{N,M}(t)$: the idea is the classical one in Guerra's interpolation, as we are explaining hereafter. By taking advantage of the mean-field nature of the model, it should be possible to linearize the ``nasty'' quadratic interactions appearing in (\ref{eq:ZJ}) by properly balancing them with extra one-body terms (i.e., those introduced in \ref{eq:aa}) such that each contribution within the source term has to match the second order moments of the order parameters. In this way we can calculate and tune the effective one-body contributions -- that are easy to evaluate -- and, in the thermodynamic limit, under the replica symmetric assumption disregard fluctuations around those means.  To this task we choose $W_{N,M}(t)$ as:
\begin{equation}
W_{N,M}(t) = \frac{\lambda(t)}{N} \sum_{i=1}^N \phi_{i} \sigma_{i} + \frac{\mu(t)}{N}\sum_{\mu=2}^K \sum_{a=1}^M \phi_{\mu,a} z_{\mu,a}+ \frac{\tau(t)}{N}\sum_{a=1}^M \sum_{i=1}^N  \xi_{i}^{1}\chi_{i}^{1,a}\sigma_{i}
\end{equation}
With this choice for the source term a few more derivatives must be calculated,
\begin{eqnarray}
  \de{\lambda}A_{J} &=&  \frac{1}{N}\mathbb{E}_{\phi,\chi,\xi} \sum_{i=1}^N \phi_{i}\mv{\sigma_{i}}_t =\frac{\lambda(t)}{N}\mathbb{E}_{\phi,\chi,\xi} \sum_{i=1}^N \big(1-\mv{\sigma_{i}}_t^{2}\big)\\
  \de{\mu}A_{J} &=& \frac{1}{N}\mathbb{E}_{\phi,\chi,\xi}\sum_{\mu=2}^K \sum_{a=1}^M \phi_{\mu,a} \mv{z_{\mu,a}}_t = \frac{\mu(t)}{N} \sum_{\mu=2}^K \sum_{a=1}^M\mathbb{E}_{\phi,\chi,\xi}\big(\mv{z_{\mu,a}^{2}}_t-\mv{z_{\mu,a}}_t^{2}\big) \\
  \de{\tau}A_{J} &=& \frac{1}{N}\mathbb{E}_{\phi,\chi,\xi} \sum_{i=1}^N \sum_{a=1}^M\mv{ \xi_{i}^{1}\chi_{i}^{1,a}\sigma_{i}}_t
\end{eqnarray}
\begin{proposition}
By inspecting the moments generated by differentiating $A_J$ we naturally introduce a complete set of order parameters to characterize the system, namely,  the two replica overlaps $p_{lm}$ for the $z$ variables, the two replica overlaps $q_{lm}$ for the $\sigma$ variables (accounting for the slow noise in the system) and two sets of quantifiers of the retrieval, namely the standard Mattis magnetization of the archetype $m_{\mu}$ and a generalized Mattis magnetization for the noise example $n_{\mu,a}$:
\begin{eqnarray}
\label{eq:2}
 p_{lm} &=& \frac{1}{KM}\sum_{\mu=1}^K \sum_{a=1}^M z_{\mu,a}^{(l)}z_{\mu,a}^{(m)}  \\
  q_{lm} &=& \frac{1}{N}\sum_{i=1}^N  \sigma_{i}^{(l)}\sigma_{i}^{(m)}\\
  n_{\mu,a} &=& \frac{1}{N} \sum_{i=1}^N \xi_{i}^{\mu} \chi_{i}^{\mu,a} \sigma_{i}\\
   m_{\mu} &=& \frac{1}{N} \sum_{i=1}^N \xi_{i}^{\mu} \sigma_{i}.
\end{eqnarray}
\end{proposition}
By these definitions each differential can be rewritten as
\begin{eqnarray}
  \de{\psi}A_{J}&=&-\frac{KM}{2N}\mathbb{E}_{\phi,\chi, \xi}\mv{p_{11}}_{t}; \ \ \de{\Gamma}A_J = \beta \Gamma(t) \frac{KM}{N} \mathbb{E}_{\phi,\chi, \xi} \big(\mv{p_{11}}_{t} -\mv{p_{12} q_{12}}_{t}\big)\\
  \de{\rho}A_{J} &=& \frac{\beta}{2}\sum_{a=1}^M \mathbb{E}_{\phi,\chi,\xi} \mv{ n_{1,a}^{2}}_{t}; \ \  \de{\lambda}A_{J} = \lambda(t) \mathbb{E}_{\phi,\chi,\xi}\big(\mv{q_{11}}_{t}-\mv{q_{12}}_{t}\big)\\
  \de{\mu} A_{J} &=& \mu(t)\frac{KM}{N} \big(\mv{p_{11}}_{t} - \mv{ p_{12}}_{t}\big); \ \  \de{\tau} A_{J} = \sum_{a=1}^M\mathbb{E}_{\phi,\chi,\xi} \mv{n_{1,a}}_{t}.
\end{eqnarray}
We are now ready to explicitly write $\frac{dA_{J}}{dt}$:
\begin{eqnarray}
\frac{dA_{J}}{dt} &=& -\frac{\alpha}{2}M \mv{p_{11}} \dot \psi+\beta \alpha M \Gamma \dot \Gamma (\mv{p_{11}}-\mv{p_{12}q_{12}})+\dot \rho \frac{\beta}{2}\sum_{a=1}^M \mv{n_{1,a}^{2}}+\lambda \dot \lambda (\mv{q_{11}}-\mv{q_{12}})\\
&+&\mu\dot\mu\alpha M (\mv{p_{11}}-\mv{p_{12}})+\dot\tau \sum_{a=1}^M \mv{n_{1,a}}.\nonumber
\end{eqnarray}

As in the replica symmetric regime we can discard fluctuations of the order parameters, assuming the latter to self-average around their mean values, that we indicate by a bar in the following, i.e. $\lim_{N \to \infty}\mathcal{P}(q_{12})=\delta(q_{12}- \bar{q}), \lim_{N \to \infty}\mathcal{P}(p_{12})=\delta(p_{12}- \bar{p})$, the strategy now is to write correlations as a source term, made of by mean values (that we will keep in the asymptotic limit), and fluctuations around these means (that will be discarded in the asymptotic limit), thus we write
\begin{eqnarray}
  \label{eq:ope}
  \mv{p_{12}q_{12}} &=& \mv{(p_{12}-\bar p)(q_{12}-\bar q)}-\bar p \bar q +\bar p \mv{q_{12}} + \bar q \mv{p_{12}} \\
   \mv{n_{1,a}^{2}} &=&  \mv{(n_{1,a}-\bar n)^{2}} - \bar n^{2}  +2\bar n \mv{n_{1,a}}. \nonumber
\end{eqnarray}

We plug the previous expressions in the streaming equation for $A_J$
\begin{eqnarray}\nonumber
  \frac{dA_{J}}{dt} &=& -\frac{\alpha}{2}M \mv{p_{11}} \dot \psi+\beta \alpha M \Gamma \dot \Gamma (\mv{p_{11}} - \mv{(p_{12}-\bar p)(q_{12}-\bar q)}+\bar p \bar q -\bar p \mv{q_{12}} - \bar q \mv{p_{12}})+\\
                &+& \dot \rho \frac{\beta}{2}\sum_{a=1}^M ( \mv{(n_{1,a}-\bar n)^{2}} - \bar n^{2}  +2\bar n \mv{n_{1,a}})+\lambda \dot \lambda (1 - \mv{q_{12}}) + \\
  &+& \mu\dot\mu\alpha M (\mv{p_{11}}-\mv{p_{12}})+\dot\tau \sum_{a=1}^M \mv{n_{1,a}}\nonumber
\end{eqnarray}
and we set to zero each coefficient coupled to a first order moment of any of the order parameters, namely
\begin{eqnarray}
  \label{eq:fluctsplitting}
  \mv{p_{11}} &: & -\frac{1}{2}\dot \psi+\beta\Gamma \dot \Gamma + \mu \dot \mu = 0, \\ %\label{eq:fluctsplitting1}
  \mv{p_{12}} &: & \beta \Gamma \dot \Gamma \bar q + \mu \dot \mu  = 0,\\ %\label{eq:fluctsplitting2}
  \mv{q_{12}} &: & \bar p \beta  \alpha M \Gamma \dot \Gamma + \lambda \dot \lambda = 0,\\ %\label{eq:fluctsplitting3}
  \mv{n_{1,a}} &: & \dot \tau +\bar n \dot \rho \beta = 0.
\end{eqnarray}
This PDE system is under-determined:  it is sufficient to find a solution which solves it and that also satisfies the Cauchy condition for $A_J$ \eqref{eq:necessarycond} and
\begin{eqnarray}
  \label{eq:solvabilitycond}
  \Gamma_{t=0} &=& 0, \\
  \rho_{t=0} &=& 0.
\end{eqnarray}
The last two constraints allow us to further simplify the solution of the model, and make it exactly solvable at the replica symmetric level.
It is easy to solve this PDE system: one can verify that the solution we are looking for is given by
\begin{eqnarray}
  \label{eq:pdesol}
  \Gamma(t) &=& \sqrt{t}, \\
  \rho(t) &=& t, \\
  \psi(t) &=& 1-(1-t)\beta (1-\bar q),\\
  \mu(t) &=& \sqrt{\beta \bar q (1-t)},\\
  \lambda(t) &=& \sqrt{\alpha \beta \bar p M (1-t)},\\
  \tau(t) &=& \bar n (1-t).
\end{eqnarray}
\begin{remark}
We point out a difference between our approach and the original Guerra's route:  in the latter, the interpolation parameter associated to glassy terms appears under the square root, while when associated to the signal terms it appears linearly; in our approach, the interpolants are general functions of $t$  and  we obtain Guerra's prescriptions as the result of the resolution of the differential equation system coded in the eq.s \ref{eq:fluctsplitting}.%,eq:fluctsplitting1,eq:fluctsplitting2,eq:fluctsplitting3}.
\end{remark}
These terms have to be plugged in the streaming equation for $A_J$, whose  final expression is given by
\begin{equation}
  \label{eq:simplstreaming}
  \frac{dA_{J}}{dt} = -\frac{1}{2}\beta\alpha M\bar p(1-\bar q) -\frac{\beta M}{2} \bar n^{2} -\frac{1}{2}\beta \alpha  M \mv{(p_{12}-\bar p)(q_{12}-\bar q)}  + \frac{\beta}{2}\sum_{a=1}^M \mv{(n_{1,a}-\bar n)^{2}}.
\end{equation}
As, under the replica symmetric ansatz, we can disregard the fluctuations asymptotically, we can state the next
\bigskip
\begin{theorem}
In the high storage ($K = \alpha N$) and in the infinite volume of the network limit ($N \to \infty$), but finite dataset size $M$, the quenched replica symmetric pressure of the model (\ref{eq:Hamiltonian2}) is given by the following expression in terms of the natural order parameters of the theory:
\begin{eqnarray}
  \label{eq:rssolution}
  A_{N,M}(\alpha, \beta, J ; t) &=&\log2 -\frac{\beta\alpha M}{2}\bar p(1-\bar q) -\frac{\beta M}{2} \bar n^{2} -\frac{\alpha M}{2} \big(\log[1-\beta (1-\bar q)] - \frac{\beta\bar q}{1-\beta(1-\bar q)}\big) +\nonumber\\
  &+& \mathbb{E}_{\phi\chi}\log\cosh \big(J+\bar n \beta \sum_{a=1}^M \chi_{a} + \sqrt{\alpha\beta\bar p M}\phi\big)
\end{eqnarray}
\end{theorem}
\begin{proof}
Note that, with the expression (\ref{eq:simplstreaming}) for the streaming of $A_J$ we express the flux of $A_J$ in $t$ by two kinds of object: average values of the order parameters, i.e. $\bar{q},\ \bar{p},\  \bar{n}$, that contribute to the source term, and all the remaining terms that are fluctuations around these means, i.e. $\langle (p_{12}-\bar p)(q_{12}-\bar q) \rangle$ and $\mv{(n_{1,a}- \bar n)^{2}}$: the latter can be discarded in the thermodynamic limit, under replica-symmetric assumption. Note further that, so far, the Mattis magnetization for the archetype has played no role.
\newline
For the sake of completeness we write also the interpolating structure in its final form that reads
\begin{eqnarray}
   \label{eq:solvedinterpolant}
   A_{J} &=& \frac{1}{N} \mathbb{E}_{\phi,\chi,\xi} \log \Big[\sum_{\sigma} \int \prod_{\mu=2}^K \prod_{a=1}^M \frac{dz_{\mu,a}}{\sqrt{2\pi}} \exp \Big(-\frac{1-\beta(1-\bar q)(1-t)}{2}\sum_{\mu=2}^K \sum_{a=1}^M z_{\mu,a}^{2}+\\ \nonumber
        &+&\sqrt{t}\sqrt{\frac{\beta}{N}}\sum_{a=1}^M \sum_{\mu=2}^K  \xi_{i}^{\mu} \chi_{i}^{\mu,a} z_{\mu,a} \sigma_{i} + t\frac{\beta N}{2} \sum_{a=1}^M (\frac{1}{N}\sum_{i=1}^N  \xi_{i}^{1}\chi_{i}^{1,a}\sigma_{i} )^{2} +J \sum_{i=1}^N  \xi_{i}^{1}\sigma_{i} +\\
        &+& \sqrt{\alpha \beta \bar p M(1-t)}\sum_{i=1}^N \phi_{i }\sigma_{i} + \sqrt{\beta \bar q (1-t)}\sum_{a=1}^M \sum_{\mu=2}^K \phi_{\mu,a}z_{\mu,a} +\bar n \beta(1-t)\sum_{a=1}^M \sum_{i=1}^N  \xi_{i}^{1} \chi_{i}^{1,a}\sigma_{i} \Big)\Big].\nonumber
\end{eqnarray}
The true power of the interpolation scheme now shines: the solution of the model can be recast as a simple integration problem. Recalling that we are interested in the original model (which can be recovered by setting $t=1, J=0$ inside the interpolating structure \ref{eq:solvedinterpolant}), we can exploit the fundamental theorem of calculus now, as
\begin{equation}
  \label{eq:cauchycond}
  A_{J}(t=1)=A_{J}(t=0)+\int_{0}^{1}ds\, \left. \frac{dA_{J}}{dt} \right|_{t=s},
\end{equation}
thus all that is left to do is evaluating the trivial 1-body problem $A_J(t=0)$: this is  a routinely integration procedure and it is performed as follows
\begin{eqnarray}
  \label{eq:onebody}
  A_{J}(t=0) &=& \frac{1}{N} \mathbb{E}_{\phi,\chi,\xi} \log \Big[\sum_{\sigma} \int \prod_{\mu=2}^K \prod_{a=1}^M \frac{dz_{\mu,a}}{\sqrt{2\pi}} \exp \Big(-\frac{1-\beta(1-\bar q)}{2}\sum_{\mu=2}^K \sum_{a=1}^M z_{\mu,a}^{2}+ \\
       &+& J \sum_{i=1}^N  \xi_{i}^{1}\sigma_{i} +\sqrt{\alpha \beta \bar p M}\sum_{i=1}^N \phi_{i }\sigma_{i} + \sqrt{\beta \bar q}\sum_{a=1}^M \sum_{\mu=2}^K \phi_{\mu,a}z_{\mu,a} +\bar n \beta\sum_{a=1}^M \sum_{i=1}^N  \xi_{i}^{1} \chi_{i}^{1,a}\sigma_{i} \Big)\Big] =\nonumber \\
       &=& -\frac{\alpha M}{2} \left(\log[1-\beta (1-\bar q)] - \frac{\beta\bar q}{1-\beta(1-\bar q)}\right) + \mathbb{E}_{\phi\chi}\log\cosh \big(J+\bar n \beta \sum_{a=1}^M \chi_{a} + \sqrt{\alpha\beta\bar p M}\phi\big),\nonumber
\end{eqnarray}
thus ending the proof.
\end{proof}
\begin{corollary}
The self-consistency equations related to the model introduced in Definition (\ref{def:H}) are obtained by looking for the stationary points of the quenched pressure $\evalat{\nabla_{\bar n,\bar q,\bar p} A_{J}}{J=0} = 0$. These equations are given by
\begin{eqnarray}
  \label{eq:sce}
    \bar p &=& \frac{\beta \bar q}{[1-\beta(1-\bar q)]^{2}},\\
  \bar q &=& \mathbb{E}_{\phi\chi} \tanh^{2} \big(\beta \bar n \sum_{a=1}^M \chi_{a} +\sqrt{\alpha\beta  \bar p M}\phi\big), \\
  \label{eq:sce_last}
  \bar n &=& \mathbb{E}_{\phi\chi}\big(\frac{1}{M}\sum_{a=1}^M \chi_{a}\big) \tanh \big(\beta \bar n \sum_{a=1}^M \chi_{a} +\sqrt{\alpha\beta  \bar p M}\phi\big).
\end{eqnarray}
Further, exploiting the auxiliary field $J$, inserted by hand in such a way that $\bar m = \nabla_J A_{J}$, we obtain
\begin{equation}
  \bar m = \mathbb{E}_{\phi\chi} \tanh \big(\beta \bar n \sum_{a=1}^M \chi_{a} +\sqrt{\alpha\beta  \bar p M}\phi\big).
\end{equation}
\end{corollary}
\begin{proof}
The proof works by straightforward derivation of $A_J$ in (\ref{eq:rssolution}).
\end{proof}

\subsection{Network behavior in the noiseless limit $\beta \to \infty$}
As standard also for the classic Hopfield scenario, namely within the AGS theory \cite{Amit,CKS}, en route to the ground-state solution (namely the self-consistencies for $\beta \to \infty$), we now assume that $\lim_{\beta \to \infty}\beta(1-\bar{q})$ is finite. This gives rise to the following
\bigskip
\begin{theorem} \label{th:2}
The zero-temperature self-consistency equations for the order parameters read as
\begin{eqnarray}
  \label{eq:ztsce}
  \bar K &:=& \frac{\sqrt{2 \alpha M} \beta (1 - \bar q) }{\beta (1 - \bar q) -1} = \mathbb{E}_{\chi} \operatorname{erf}^{\prime}\big(\frac{\bar n \sum_{a=1}^{M}\chi_{a}}{\bar K + \sqrt{2\alpha M}}\big) \\
  \label{eq:th2_1}
  \bar n &=& \mathbb{E}_{\chi} \frac{\sum_{a=1}^{M}\chi_{a}}{M}\operatorname{erf}\big(\frac{\bar n \sum_{a=1}^{M}\chi_{a}}{\bar K + \sqrt{2\alpha M}}\big) \\
  \label{eq:th2_2}
  \bar m &=& \mathbb{E}_{\chi} \operatorname{erf}\big(\frac{\bar n \sum_{a=1}^{M}\chi_{a}}{\bar K + \sqrt{2\alpha M}}\big)
\end{eqnarray}
where $\rm erf$ is the error function and ${\rm erf}^{\prime}$ is it's first derivative ${\rm erf}^\prime (x) := \frac{2}{\sqrt{\pi}}\exp(-x^2)$.
\end{theorem}
\begin{proof}
As a first step we introduce an additional term \(\beta x\) in the argument of the hyperbolic tangent appearing in the self-consistency equations (\ref{eq:sce}):
\begin{eqnarray}
\label{eq:orig-sce}
\bar q &=& \mathbb{E}_{\chi,\phi}\tanh^{2}\Big(\beta \bar n \sum_{a=1}^{M}\chi_{a} +\beta \phi \sqrt{\frac{\alpha M \bar q}{[1-\beta(1-\bar q)]^{2}}} +\beta x\Big)\\
\bar n &=& \mathbb{E}_{\chi,\phi}\big( \frac{1}{M}\sum_{a=1}^{M} \chi_{a} \big)\tanh\Big(\beta \bar n \sum_{a=1}^{M}\chi_{a} +\beta \phi \sqrt{\frac{\alpha M \bar q}{[1-\beta(1-\bar q)]^{2}}} +\beta x \Big)\\
\bar m &=& \mathbb{E}_{\chi,\phi}\tanh\Big(\beta \bar n \sum_{a=1}^{M}\chi_{a} +\beta \phi \sqrt{\frac{\alpha M \bar q}{[1-\beta(1-\bar q)]^{2}}} +\beta x\Big).
\end{eqnarray}
We also recognize that at \(\beta \to \infty\) we also have \(q \to 1\) thus in order to correctly perform the limit a reparametrization is in order,
\begin{equation}
\label{eq:q-lim}
\bar q=1-\frac{\delta q}{\beta} ~~~ \textrm{as} ~~ \beta \to \infty
\end{equation}
Via this reparametrization we obtain
\begin{eqnarray}
\label{eq:repar-sce}
1-\frac{\delta q}{\beta} &=& \mathbb{E}_{\chi,\phi}\tanh^{2}\Big(\beta \bar n \sum_{a=1}^{M}\chi_{a} +\beta \phi \sqrt{\frac{\alpha M (1-\frac{\delta q}{\beta})}{(1-\delta q)^{2}}} +\beta x\Big)\\
\bar n &=& \mathbb{E}_{\chi,\phi}\big( \frac{1}{M}\sum_{a=1}^{M} \chi_{a} \big)\tanh\Big(\beta \bar n \sum_{a=1}^{M}\chi_{a} +\beta \phi \sqrt{\frac{\alpha M (1-\frac{\delta q}{\beta})}{(1-\delta q)^{2}}} +\beta x \Big)\\
\bar m &=& \mathbb{E}_{\chi,\phi}\tanh\Big(\beta \bar n \sum_{a=1}^{M}\chi_{a} +\beta \phi \sqrt{\frac{\alpha M (1-\frac{\delta q}{\beta})}{(1-\delta q)^{2}}} +\beta x\Big).
\end{eqnarray}
Taking advantage of the new parameter \(x\) we can recast the last equation in \(\delta q\) as a derivative of the magnetization \(\bar m\):
\begin{equation}
\label{eq:deriv-sce}
\frac{\partial \bar m}{\partial x} = \beta [1-(1-\frac{\delta q}{\beta})] = \delta q
\end{equation}
where we used both the self-consistencies for $\bar m$ and $\delta q$.
Thanks to this correspondence between \(\bar m\) and \(\delta q\), we can proceed with our limit without worrying about \(\bar q\): the limiting equations for \(\bar m, \bar n\) are now for $\beta \to \infty$:
\begin{eqnarray}
\label{eq:n-sce}
\bar n &=& \mathbb{E}_{\chi,\phi}\big(\frac{1}{M} \sum_{a=1}^{M}\chi_{a}\big)\operatorname{sign}\Big( \bar n \sum_{a=1}^{M}\chi_{a} +\phi \sqrt{\frac{\alpha M }{[1-\delta q]^{2}}} +x\Big),\\
\label{eq:m-sce}
\bar m &=& \mathbb{E}_{\chi,\phi}\operatorname{sign}\Big( \bar n \sum_{a=1}^{M}\chi_{a} +\phi \sqrt{\frac{\alpha M }{[1-\delta q]^{2}}} +x\Big).
\end{eqnarray}
These equations can be further simplified by evaluating the Gaussian integral in \(\phi\), via the relation:
$$
\mathbb{E}_{\phi} \operatorname{sign}(A\phi+B) = \text{erf}\left(\frac{B}{\sqrt{2} A}\right)
$$
to get
\begin{eqnarray}
\label{eq:m-sce}
\bar m &=& \mathbb{E}_{\chi}\operatorname{erf}\Big[ (\bar n \sum_{a=1}^{M}\chi_{a} +x)  \frac{1-\delta q}{\sqrt{2\alpha M }} \Big]\\
\bar n &=& \mathbb{E}_{\chi}\big(\frac{1}{M}\sum_{a=1}^{M}\chi_{a}\big)\operatorname{erf}\Big [ (\bar n \sum_{a=1}^{M}\chi_{a} +x)  \frac{1-\delta q}{\sqrt{2\alpha M }} \Big]
\end{eqnarray}
while \(\delta q\), thanks to \eqref{eq:deriv-sce}, becomes
\begin{equation}
\label{eq:q-ztsce}
\delta q = \frac{\partial \bar m}{\partial x} = \mathbb{E}_{\chi} \frac{2}{\sqrt{\pi}}  \frac{1-\delta q}{\sqrt{2\alpha M }}\exp \left \{ -\left[ (\bar n \sum_{a=1}^{M}\chi_{a} +x)  \frac{1-\delta q}{\sqrt{2\alpha M }}\right]^{2} \right \}.
\end{equation}
In order to simplify the equation in \(\delta q\) we make one last change of variables,
$$
\delta q = \frac{\delta Q}{\delta Q + \sqrt{2\alpha M}}
$$
yielding to
\begin{eqnarray}
\label{eq:ztsce-final}
\bar m &=& \mathbb{E}_{\chi}\operatorname{erf}\Big( \frac{\bar n \sum_{a=1}^{M}\chi_{a}}{\sqrt{2\alpha M } + \delta Q} \Big)\\
\label{eq:ztsce-final-bis}
\bar n &=& \mathbb{E}_{\chi}\big(\frac{1}{M}\sum_{a=1}^{M}\chi_{a}\big)\operatorname{erf}\Big( \frac{\bar n \sum_{a=1}^{M}\chi_{a}}{\sqrt{2\alpha M } + \delta Q} \Big)\\
\delta Q &=& \mathbb{E}_{\chi} \frac{2}{\sqrt{\pi}}\exp \Big[-\Big( \frac{\bar n \sum_{a=1}^{M}\chi_{a}}{\sqrt{2\alpha M } + \delta Q}\Big)^{2} \Big]
\end{eqnarray}
where \(x\) has been set to \(0\), allowing to close the proof.
\end{proof}
The solutions of these equations, as $p$ and $M$ are varied, is captured in the plots of Fig.~\ref{fig:plots_zt_zalpha}.
Remarkably, there exists a crossover at $\tilde{M}(p)$, such that as $M < \tilde{M}(p)$ ($M > \tilde{M}(p)$) the example magnetization $\bar n$ is larger (smaller) than the archetype magnetization $\bar m$. We  would be tempted to label the crossover points $\tilde{M}(p)$ as candidate markers of a phase transition, yet we still need to further inspect the system and to develop the theory by suitably sending both $M$ and $N$ (and $K$ as well in the high storage) to infinity before we can robustly refer to a phase transition; this work will be achieved in the next subsection.

\begin{figure}[tb]
\centerline{\includegraphics[scale=0.8]{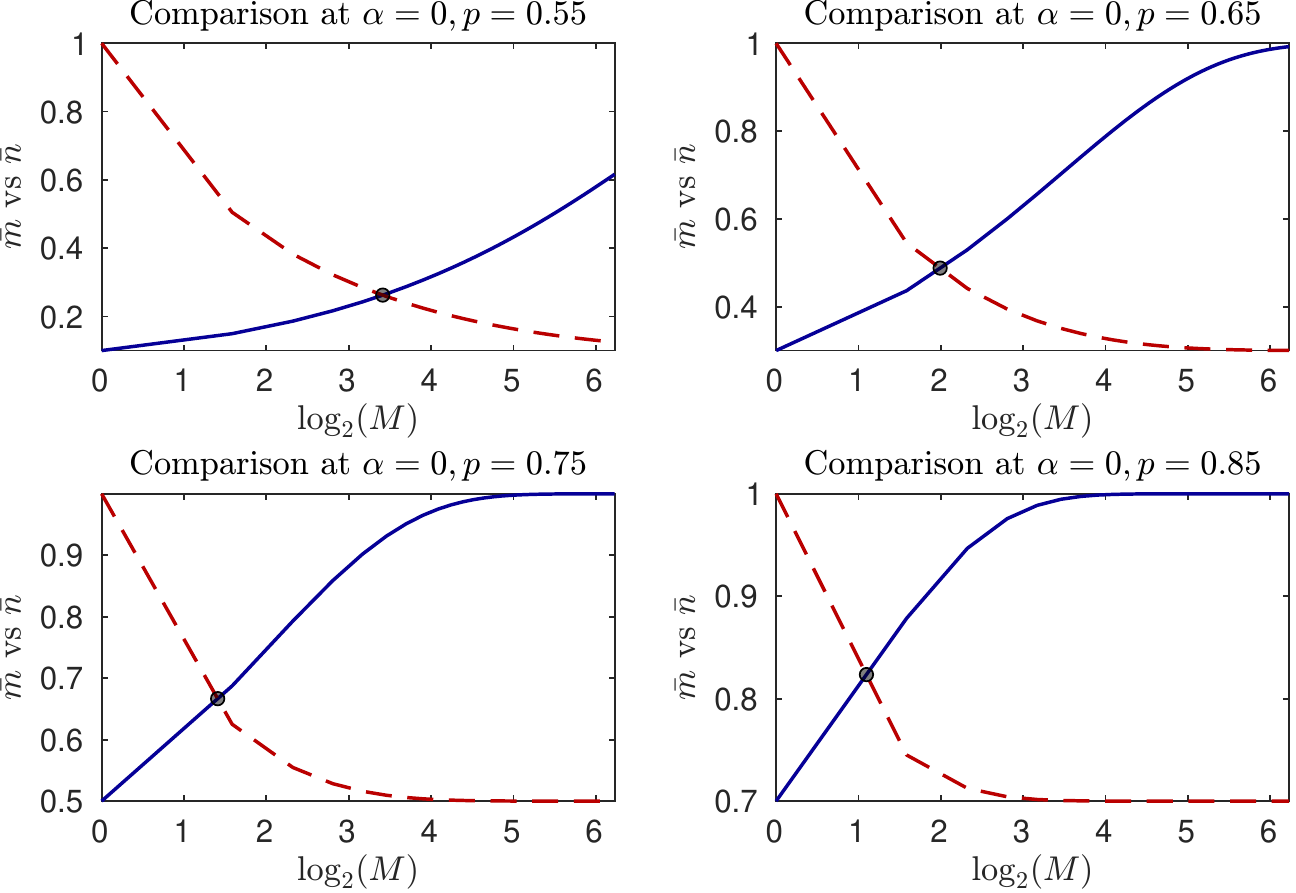}}
\caption[]{We compare the expected magnetizations $\bar m$ (solid line) and $\bar n$ (dashed line), obtained by numerically solving (\ref{eq:ztsce-final}) and (\ref{eq:ztsce-final-bis}), holding in the limit of vanishing temperature $\beta \to \infty$ and infinite size $N \to \infty$, in the low load regime $\alpha=0$. We notice that, as the size $M$ of the dataset increases, the magnetization of the noisy example diminishes while that of the archetype starts to grow; we denote with $\tilde{M}$ the value of $M$ corresponding to the intersection between the two curves. Different values of $p$ are considered, as reported in the title of the panels.  \label{fig:plots_zt_zalpha} }
\end{figure}

\subsection{Network behavior in the large dataset limit $M \to \infty$} \label{ssec:MM}

In the theory developed so far, we assumed that, as $K$ and $N$ are made larger and larger, their ratio $\alpha$ remains finite in such a way that it can be used as an intensive parameter tuning pattern load, however, the parameter $M$ expressing the sample size is still extensive and its tuning is not related to a tuning in the network volume $N$ or in the number $K$ of pattern. In this section we turn the whole theory intensive such that the meaning of the self-consistencies, as well as the nature of the phase transition, can appear manifestly.
\newline
This goal is approached by steps: first, setting $M$ as large (but still retaining the parameter $M$ explicit), via the central limit theorem, we approximate the quantity $\frac{1}{M}\sum_{a=1}^M \chi_{a}$, where, we recall $\mathcal{P}(\chi_{a}) =  p\,\delta(\chi_{a} -1) + (1-p)\,\delta(\chi_{a} +1)$, with a Gaussian random variable, namely
\begin{equation}
  \label{eq:clt}
 \frac{1}{M}\sum_{a=1}^M \chi_{a} \sim 2p-1 +  2\sqrt{\frac{p(1-p)}{M}} Z, \quad\, Z \sim \calN(0,1).
\end{equation}
This expression can be used to considerably simplify the self-consistency equations. Let us focus on the retrieval of the noisy patterns quantified by $\bar{n}$:
\begin{eqnarray}\nonumber
  \bar n  &=& \mathbb{E}_{\phi, Z}\left( 2p-1 +  2\sqrt{\frac{p(1-p)}{M}} Z\right) \tanh \left[\beta M \bar n \left( 2p-1 +  2\sqrt{\frac{p(1-p)}{M}} Z\right)  +\sqrt{\alpha\beta M  \bar p }\phi\right]=\\
  \label{eq:largeMsce}
  &=& (2p-1) \bar m + \beta M \bar n \frac{4p(1-p)}{M} (1-\bar q),
\end{eqnarray}
where the last step has been performed via Wick theorem: $\mathbb{E}_{Z} Z f(Z) = \mathbb{E}_{Z} \partial_{Z} f(Z)$. This equation implies that, for large $M$, beyond $n$, the order parameter $m$ -- assessing the retrieval of archetypes -- also starts to play a fundamental role; in fact, the configurations $\boldsymbol \sigma = \boldsymbol \xi^{\mu}$ emerge as ground states. Indeed, we have
\begin{equation}
  \label{eq:etaSCElargeM}
 \bar n = \frac{\bar m r}{1-\beta  (1-\bar q) (1- r^2)},
\end{equation}
where, for simplicity, we posed $r=2p-1$.
\newline
This equation allows us to get rid of $n$ and rather focus on $m$: by replacing (\ref{eq:etaSCElargeM}) in the remaining self-consistencies we find
\begin{eqnarray}
  \bar p &=& \frac{\beta \bar q}{[1-\beta(1-\bar q)]^{2}}, \ \  G: =  \frac{\beta r^{2}}{1-\beta(1-r^{2})(1-\bar q)},\\
  \bar m &=&  \mathbb{E}_{\phi,Z} \tanh \left[G \bar m M \Big(1+ Z\sqrt{\frac{1-r^{2}}{r^{2} M}}\Big)+ \phi \sqrt{\alpha \beta \bar p M} \right],\\
   \bar q &=&  \mathbb{E}_{\phi,Z} \tanh^{2} \left[G \bar m M \Big(1+ Z\sqrt{\frac{1-r^{2}}{r^{2} M}}\Big)+ \phi \sqrt{\alpha \beta \bar p M} \right],
\end{eqnarray}
where the parameter $G$ has been introduced to lighten the notation.
\newline
For a straight comparison to AGS theory, we introduce a more convenient scale for the temperature, such that
\begin{equation}
  \label{eq:betatransform}
\beta \to \frac{\beta }{\beta  (q-1) \left(r^2-1\right)+r^2}.
\end{equation}
Via this  rescaling the self-consistent equations become
\begin{eqnarray}
\label{eq:prev1}
  \bar m &=&  \mathbb{E}_{\phi,Z} \tanh \left[\beta \bar m M + Z \beta \sqrt{M \frac{1-r^{2}}{r^{2} }{\bar m}^{2}} + \phi \beta\sqrt{\alpha \frac{  \bar q}{r^4 (1-\beta  (1-\bar q))^2}M} \right],\\
  \label{eq:prev2}
 \bar q &=&  \mathbb{E}_{\phi,Z} \tanh^{2} \left[\beta \bar m M + Z \beta \sqrt{M \frac{1-r^{2}}{r^{2} }{\bar m}^{2}} + \phi \beta\sqrt{\alpha \frac{  \bar q}{r^4 (1-\beta  (1-\bar q))^2}M} \right].
\end{eqnarray}
These equations can be further simplified as shown in the next
\bigskip
\begin{proposition}
For the model introduced in Definition (\ref{def:H}), in the thermodynamic limit and for large samples of examples ($M \gg 1$), the order parameters fulfill the following self-consistent equations:
\begin{eqnarray}
\label{eq:bo1}
  \bar m &=&  \mathbb{E}_{Z} \tanh \left[\beta \bar m M + Z \beta \sqrt{M \frac{1-r^{2}}{r^{2} }{\bar m}^{2} + \alpha \frac{  \bar q}{r^4 (1-\beta  (1-\bar q))^2}M} \right],\\
  \label{eq:bo2}
 \bar q &=&  \mathbb{E}_{Z} \tanh^{2} \left[\beta \bar m M + Z \beta \sqrt{M \frac{1-r^{2}}{r^{2} }{\bar m}^{2} + \alpha \frac{  \bar q}{r^4 (1-\beta  (1-\bar q))^2}M} \right].
\end{eqnarray}
\end{proposition}
\bigskip
\begin{proof}
Given a function $F$, we introduce the relation
\begin{equation}
\label{eq:tool}
\mathbb{E}_{X,Y} F(a X+bY +c) = \mathbb{E}_{Z} F(\sqrt{a^{2}+b^{2}}Z+c),
\end{equation}
where $X, Y, Z$ are assumed to be Gaussian random variables.
This relation allows us to reduce any number of averages with the same structure to a single Gaussian average, and, in particular, by appling \eqref{eq:tool} to eqs.~(\ref{eq:prev1})-(\ref{eq:prev2}) we get eqs.~(\ref{eq:bo1})-(\ref{eq:bo2}).
\end{proof}
\bigskip
\begin{remark}
The argument of the hyperbolic tangents in (\ref{eq:bo1})-(\ref{eq:bo2}) includes three contributions (and no longer just two as in the standard Hopfield scenario). Indeed, beyond the signal carried by $\bar{m}$ there are two sources of (slow) noise: a classic one given by the other patterns not retrieved (pattern interference) and a new one given by the examples within the dataset related to the pattern the network is retrieving (example interference).
\end{remark}
\bigskip
\begin{remark}
As a consistency check, we point out that if the network is not provided with datasets, but just patterns (i.e. $M=1$) and those are assumed noiseless (i.e. $r=1$), the whole theory collapses over the standard AGS theory of the Hopfield model as expected.
\end{remark}
\bigskip
\begin{proposition}\label{MSscaling}
To be sure that the archetype is retrieved over the noisy patterns we can use a simple argument, namely we can require that
\begin{equation}
  \label{eq:dis1}
   \beta M \bar m> \beta \sqrt{M} |Z| \sqrt{\frac{1-r^{2}}{ r^{2}} {\bar m}^{2}+\frac{\alpha}{r^4 (1 - \beta (1-\bar q))^2} \bar q} ,\quad Z \sim \calN(0,1)
 \end{equation}
 holds almost surely: a solution in $M$ to the above equation is given by
 \begin{equation}
   \label{eq:solMdis1}
     M > \frac{\gamma^{2}}{r^{2}} \left[1-r^{2}+\frac{q}{\bar m^{2} (1 - \beta (1-\bar q))^2}\frac{\alpha} {r^{2}}\right]
 \end{equation}
 where $\gamma$ establishes the confidence level (indeed the last condition implies $|Z|<\gamma,\quad Z \sim \calN(0,1)$
 which can be satisfied up to an exceedingly small probability at finite $M$): these results recover the scaling behaviour achieved via signal to noise analysis in the previous section. In particular, in the low storage $\alpha=0$ the correct scaling is $M \propto 1/(2p-1)^2$, while in the high storage $\alpha>0$ the correct scaling is $M \propto 1/(2p-1)^4$.

 \begin{proof}
The proof works by requiring that the signal term in the argument of $\tanh$ \eqref{eq:bo1} is on average greater than the noise term, which amounts to the condition:
   \begin{equation}
     \label{eq:pr1}
        \beta \bar m M > |Z| \beta \sqrt{M \frac{1-r^{2}}{r^{2} }{\bar m}^{2} + \alpha \frac{  \bar q}{r^4 (1-\beta  (1-\bar q))^2}M}
      \end{equation}
this condition can be recast as
\begin{equation}
  \label{eq:pr2}
  |Z| < \frac{ \sqrt{M}}{ \sqrt{\frac{1-r^{2}}{r^{2} } + \frac{\alpha}{{\bar m}^{2}} \frac{  \bar q}{r^4 (1-\beta  (1-\bar q))^2}} } =: W(M)
   \end{equation}
if we further require
\begin{equation}
  \label{eq:pr3}
  |Z| < \gamma < W(M)
\end{equation}
by solving $W(m)>\gamma$ w.r.t $M$ we obtain
\begin{equation}
  \label{eq:pr4}
  M > \gamma^2  \big[\frac{1-r^{2}}{r^{2} } + \frac{\alpha}{{\bar m}^{2}} \frac{  \bar q}{r^4 (1-\beta  (1-\bar q))^2}\big]
\end{equation}
concluding the proof.
\end{proof}
\end{proposition}

Now, to further inspect the competition between $m$ and $n$, we resume Theorem \ref{th:2}, see in particular equations \eqref{eq:th2_1}-\eqref{eq:th2_2}, which are used to build Fig.~\ref{fig:si_transition}:
 the ``Fuzzy'' phase corresponds to a region in the parameter space where the retrieval of the examples is more effective than the retrieval of the archetype ($\bar n > \bar m$), no matter how good the retrieval can be. Focusing on the low-load regime, this region is demarcated by the line $\tilde{M}(p):=\tilde{M}(\alpha=0,p)$; beyond that line the retrieval of the archetype is more effective than the retrieval of the example ($\bar m > \bar n$) and, by requiring also a high-quality retrieval (i.e., $|\bar m| >z$), we get the line $\tilde M_{z} (\textcolor{black}{\alpha}, p)$, which detects a region whose volume decreases with $z$. Focusing on the high-load regime, the ``Fuzzy'' region is demarcated by the line $\tilde{M}(\alpha, p)$, which is more restrictive that $\tilde{M}(\alpha, p)$.

\begin{figure}[tb]
\centering
\includegraphics[width=.55\linewidth]{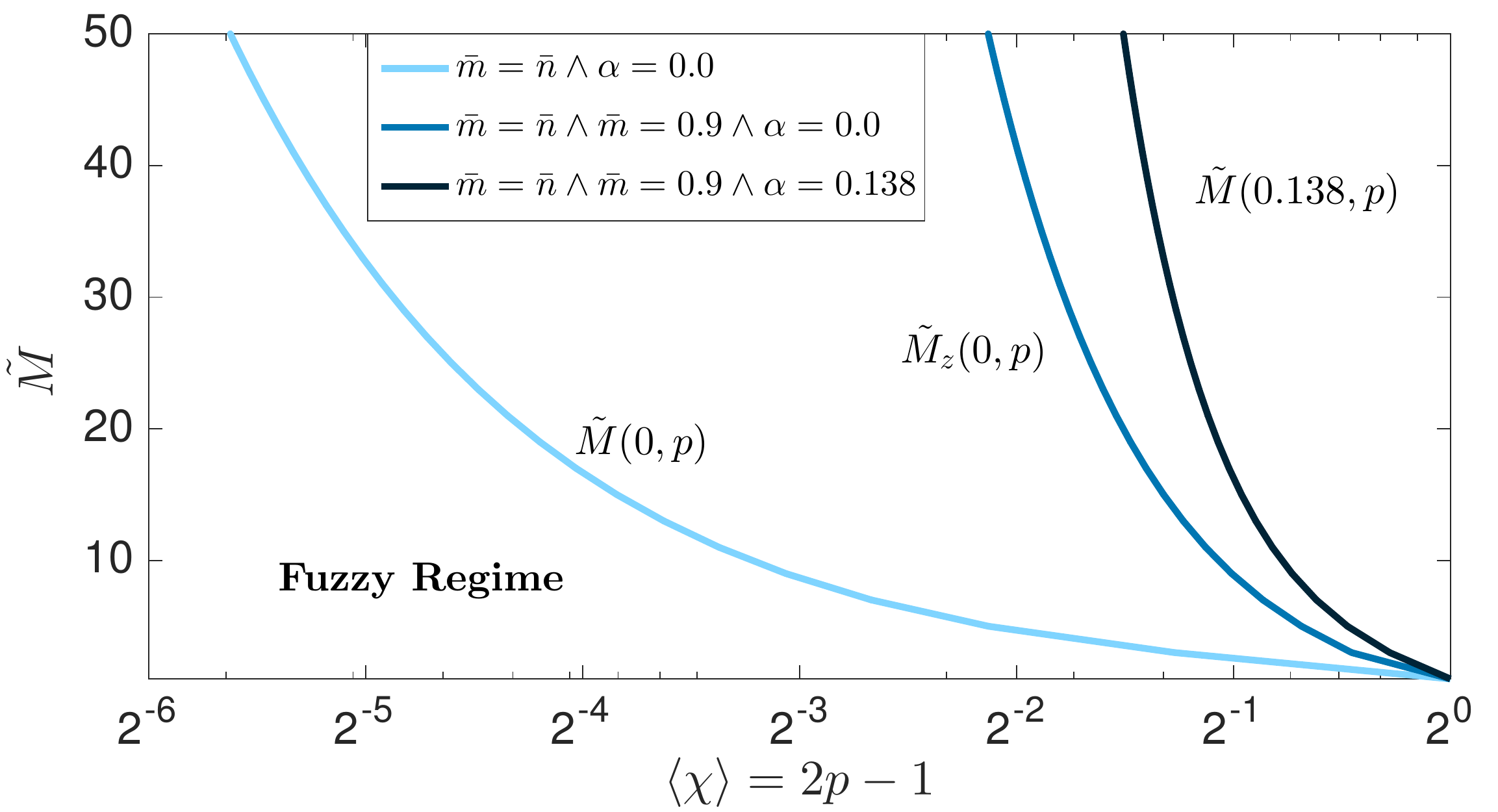}
\caption{\label{fig:si_transition} In this plot we show the crossovers values for $M$ as a function of $2p-1$ and under different conditions.
In particular, given $P=\alpha N$ archetypes and feeding the network with $M \times P$ examples characterized by a noise $p$, as $M>\tilde{M}(\alpha, p)$, then $\bar m > \bar n$. As expected, moving from a low load ($\alpha=0$) to a high load ($\alpha >0$), the region in this parameter space where $\bar m > \bar n$ shrinks. Notice that $\tilde{M}(\alpha, p)$ simply signs a crossover between $\bar n$ and $\bar m$, while no conditions are posed on the magnitude of magnetizations. This kind of information is provided by $\tilde{M}_z(\alpha, p)$ which also requires that $|\bar m| >z$. In this way, we can highlight a region where the pattern is better retrieved than examples \emph{and} with high quality.
%This plot further refines the result show in the front matter, the ``Concept formed'' phase is now broken into \(3\) distinct phases: ``A'' - where the archetype dominates the examples, but still it's magnetization can be arbitrarily low; ``B'' - where the archetype magnetization is at least $90\%$ and ``C'' - where the archetype is also retrievable at high load $(\alpha \leq 0.138)$.
}
\end{figure}

Finally, we want to deepen the possible existence of a genuine phase transition distinguishing between a region where the system can infer the archetype ($\bar m >0$) and a regione where noise -- either fast (i.e., ruled by $T$) or slow (i.e., ruled by a suitable combination of $\alpha$, $r$ and $M$) -- prevails ($\bar m =0$). 
A close look to the self-consistent equations \eqref{eq:bo1}-\eqref{eq:bo2} suggests that a suitable, intensive and tuneable parameter able to trigger the phase transition is given by 
\begin{equation}
\label{eq:vanishR}
\rho:=  \frac{\alpha}{M r^4}.
\end{equation}
In the following analysis we will let $M \to \infty$ and, accordingly, we rescale the temperature as $\beta \to \frac{\beta}{M}$ to ensure the well-definiteness of the model \eqref{eq:Hamiltonian2}; this limit also implies that that we are focusing on the limit of high disorder in the dataset ($p \to 1/2$) so to retain a finite $\rho$.
\bigskip
\begin{proposition}
In the limit of large samples ($M \to \infty$) and high disorder in the dataset ($r \to 0$) a critical behaviour is found as $\rho$ approaches $\rho_{c} = \frac{2}{\pi}$, where $\bar m \sim \sqrt{\frac{3}{\pi}} \sqrt{2 - \pi \rho}$.
\end{proposition}
\bigskip
\begin{proof}
%\sout{Exploiting the scaling  (\ref{eq:vanishR}), the self-consistency equations (\ref{eq:bo1})-(\ref{eq:bo2}) can we rewritten as}
%%
%\begin{eqnarray}
%  \bar m &=&  \mathbb{E}_{Z} \tanh \left[\beta \bar m M + Z \beta M\sqrt{\frac{1- \sqrt{\frac{\alpha}{\rho M} }}{M  \sqrt{\frac{\alpha}{\rho M} }}{\bar m}^{2} +  \frac{ \rho  \bar q}{ (1-\beta  (1-\bar q))^2}} \right],\\
% \bar q &=&  \mathbb{E}_{Z} \tanh^{2} \left[\beta \bar m M + Z \beta M\sqrt{\frac{1- \sqrt{\frac{\alpha}{\rho M} }}{M  \sqrt{\frac{\alpha}{\rho M} }}{\bar m}^{2} +  \frac{ \rho  \bar q}{ (1-\beta  (1-\bar q))^2}} \right].
%\end{eqnarray}
%%
%\sout{In this intensive picture, we can take the limit $M \to \infty$ and, discarding the ergodic $\bar m=\bar q = 0$ solution, we obtain the following limiting equations}
%\begin{eqnarray}
%  \bar m &=&  \mathbb{E}_{Z} \operatorname{sign} (\bar m + Z \sqrt{ \rho  } ),\\
% \bar q &=& 1,
%\end{eqnarray}
%\sout{which can be further simplified to}
%  \begin{equation}
%  \label{eq:criticalSCE}
%  \bar m = \text{erf}\left(\frac{\bar m}{\sqrt{2 \rho}}\right).
%\end{equation}
%
%
%

\begin{figure}
%\centerline{\includegraphics[width=0.5\textwidth]{SCE/critical.pdf}}
\centerline{\includegraphics[width=0.55\textwidth]{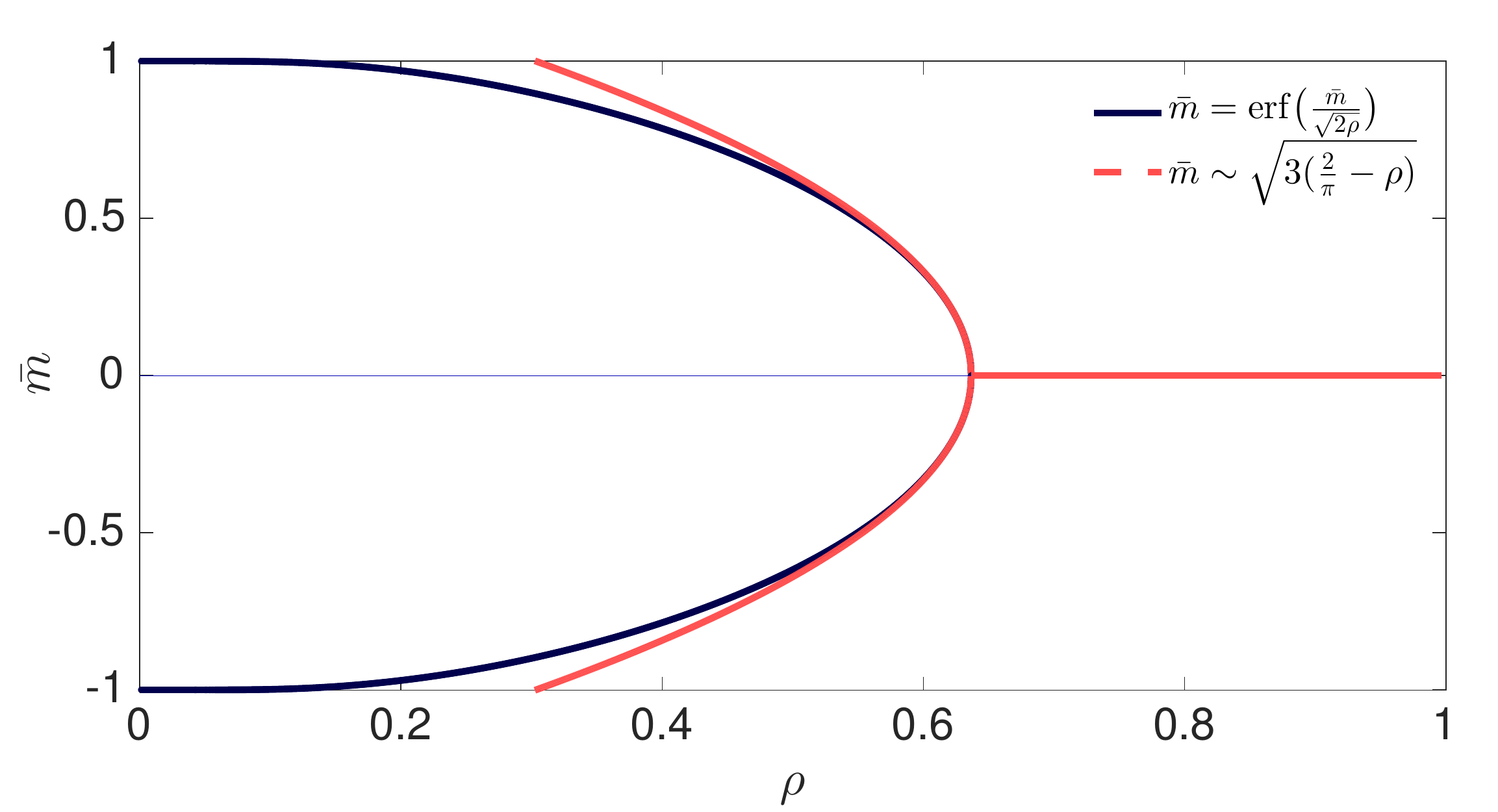}}
\caption{\label{fig:critbehaviour} Zero-temperaure self-consistency for the Mattis magnetization in the limit of $M,N,K \to \infty$ such that $(0,1) \ni \rho := K/(M N r^4)$ is the tunable control parameter for the dataset density (see eq. \ref{eq:vanishR}): for values of $\rho$ smaller than $\rho_c = 2/\pi$ the solely solution is $\bar{m}=0$ while for values of $\rho > \rho_c$ two (gauge-invariant) not-null values of the Mattis magnetization appear. Beyond the exact result given by eq.~(\ref{eq:aaa}), the figure also shows a comparison with the square-root estimate valid nearby the critical point.}
\end{figure}
Taking the large-$M$ self-consistency equations \eqref{eq:bo1}-\eqref{eq:bo2}, all that we have to do is replace $r^{2}$ with $\sqrt{\frac{\alpha}{\rho M}}$ and $\beta$ with $\frac{\beta}{M}$ obtaining:
\begin{eqnarray}
\bar{m}&=&\mathbb{E}_{Z}\tanh\left(\beta\bar{m}+Z\beta\sqrt{\frac{1-\sqrt{\frac{\alpha}{\rho M}}}{M\sqrt{\frac{\alpha}{\rho M}}}+\frac{\rho\bar{q}}{[1-\frac{\beta}{M}(1-\bar{q})]^{2}}}\right),\\
\bar{q}&=&\mathbb{E}_{Z}\tanh^{2}\left(\beta\bar{m}+Z\beta\sqrt{\frac{1-\sqrt{\frac{\alpha}{\rho M}}}{M\sqrt{\frac{\alpha}{\rho M}}}+\frac{\rho\bar{q}}{(1-\frac{\beta}{M}[1-\bar{q})]^{2}}}\right).
\end{eqnarray}
The whole theory now has been rephrased intensive in $M$, allowing us to take the limit $M\to\infty$:
\begin{eqnarray}
\label{eq:sce1}
%\bar{m}&\xrightarrow{M \to \infty}&\mathbb{E}_{Z}\tanh(\beta\bar{m}+\beta Z\sqrt{\rho\bar{q}}),\\
%\bar{q}&\xrightarrow{M \to \infty}&\mathbb{E}_{Z}\tanh^{2}(\beta\bar{m}+\beta Z\sqrt{\rho\bar{q}}).\label{eq:sce2}
\bar{m} &=&\mathbb{E}_{Z}\tanh(\beta\bar{m}+\beta Z\sqrt{\rho\bar{q}})~~~ \textrm{as} ~~ M \to \infty,\\
\bar{q}&=&\mathbb{E}_{Z}\tanh^{2}(\beta\bar{m}+\beta Z\sqrt{\rho\bar{q}}) ~~~ \textrm{as} ~~ M \to \infty.\label{eq:sce2}
\end{eqnarray}
In particular, the zero-temperature limit of the previous equations, where we send $\beta\to\infty$, reads as
\begin{eqnarray}
\label{eq:aaa}
\bar{m} &=& \mathbb{E}_{Z}\text{sign}(\bar{m}+Z\sqrt{\rho\bar{q}}) = \text{erf}\left(\frac{\bar{m}}{\sqrt{2\rho}}\right) ~~~ \textrm{as} ~~ \beta, M \to \infty,\\
\bar{q}&=&\mathbb{E}_{Z}\text{sign}(\bar{m}+Z\sqrt{\rho\bar{q}})^{2}=1 ~~~ \textrm{as} ~~ \beta, M \to \infty.
\end{eqnarray}
%
%In particular, focusing on $\bar m$thus we reach
%\begin{equation}
%\bar{m}=\mathbb{E}_{Z}\text{sign}(\bar{m}+Z\sqrt{\rho})=\text{erf}\left(\frac{\bar{m}}{\sqrt{2\rho}}\right)
%\end{equation}
%
%
By Taylor expanding equation \eqref{eq:aaa} around $\bar m=0$, a critical behaviour is found at $\rho_{c} = \frac{2}{\pi}$ with scaling $\bar m \sim \sqrt{\frac{3}{\pi}} \sqrt{2 - \pi \rho}$ near the critical point.
\end{proof}
\bigskip
The behavior of the magnetization $\bar m$ versus $\rho$, in the limit $M,N,K \to \infty$ is shown in Fig.~\ref{fig:critbehaviour}, where the critical behavior is also corroborated.
\bigskip
\begin{remark}
As a direct consequence of the previous proposition we can state that concepts, namely {\em archetypes} of the experienced examples, are formed by the network via a critical behavior and not abruptly (as, for instance, happens to the Hopfield network when forgetting, i.e. the {\em blackout scenario}).
\end{remark}

\end{document}